\documentclass[11pt]{article}
\usepackage{fullpage}
\usepackage{datetime}
\settimeformat{ampmtime}
\usdate

\usepackage{scrtime}
\usepackage[draft,scrtime]{prelim2e}

\usepackage{amsmath, amsthm, amssymb, bm, calligra}
\usepackage{mathrsfs}
\usepackage{cite}
\usepackage{graphicx}
\usepackage{epstopdf}
\usepackage{multirow}
\usepackage{url}
\usepackage{hyperref, color} 
\usepackage{subfigure}

\newtheorem{thm}{Theorem}
\newtheorem{lem}{Lemma}
\newtheorem{cor}{Corollary}
\newtheorem{prop}{Proposition}
\newtheorem{remark}{Remark}

\newcommand{\vm}[1]{\boldsymbol{#1}}
\renewcommand{\P}[1]{\operatorname{\mathbb{P}}\{#1\}}
\newcommand{\E}{\operatorname{{E}}}

\DeclareMathAlphabet{\mathcalligra}{T1}{calligra}{m}{n}
\newcommand{\argmin}{\operatorname{argmin}}
\newcommand{\T}{{\operatorname{T}}}
\newcommand{\F}{{\operatorname{F}}}

\newcommand{\er}{\mathrm{e}}

\newcommand{\vct}[1]{\bm{#1}}

\def\y{\vm{Y}}
\def\x{\vm{X}}

\def\A{\vm{A}}
\def\z{\vm{Z}}
\def\w{\vm{W}}

\def\ff{\vm{f}}

\def\o{\ell}

\def\hh{\vm{h}}
\def\mm{\vm{x}}
\def\ot{\ell^\prime}
\def\nt{n^\prime}
\def\kt{k^\prime}
\def\bh{{\vm{b}}}

\newcommand{\uu}{\vct{u}}
\newcommand{\vv}{\vct{v}}
\newcommand{\e}{\vm{e}}

\def\PT{\mathcal{R}}
\def\PTc{\mathcal{R}_{\scalebox{0.6}{$\perp$}}}
\def\PO{\mathcal{P}}
\def\POc{\mathcal{P}_{\scalebox{0.6}{$\perp$}}}

\def\cA{\mathcal{A}}

\def\reals{\mathbb{R}}
\def\comps{\mathbb{C}}

\def\mm{\vm{m}}
\def\cc{\vm{c}}
\def\X{\vm{X}}
\def\xx{\vm{x}}
\def\ww{\vm{w}}
\def\l{\ell}

\newcommand{\<}{\langle}
\renewcommand{\>}{\rangle}
 
\begin{document}
\title{Leveraging Diversity and Sparsity in  Blind Deconvolution}
\author{Ali Ahmed and Laurent Demanet\thanks{Ali Ahmed is currently with the Information Technology University, Lahore, Pakistan. He was associated until recently with the Department of Mathematics, MIT, Cambridge, MA.  Laurent Demanet is with the Department of Mathematics, MIT, Cambridge, MA. Email for the corresponding author: alikhan@mit.edu.  Both authors are sponsored by AFOSR grants FA9550-12-1-0328 and FA9550-15-1-0078. LD is also funded by NSF, ONR, and Total SA. We thank Augustin Cosse for interesting discussions. Some preliminary results in this direction were presented in an earlier conference publication, namely,  A convex approach to blind deconvolution with diverse inputs, in Proc. IEEE CAMSAP, Cancun, December 2015. \textcopyright This work has been submitted to the IEEE Transactions on Information Theory for possible publication. Copyright may be transferred without notice, after which this version may no longer be accessible.}}
\date{\today}
\maketitle
\begin{abstract}

	This paper considers recovering $L$-dimensional vectors $\ww$, and $\xx_1,\xx_2, \ldots, \xx_N$ from their circular convolutions $\vm{y}_n = \ww*\xx_n, \ n = 1,2,3, \ldots, N$. The vector $\ww$ is assumed to be $S$-sparse in a known basis that is spread out in the Fourier domain, and each input $\xx_n$ is a member of a known $K$-dimensional random subspace.

	We prove that whenever $K + S\log^2S \lesssim L /\log^4(LN)$, the problem can be solved effectively by using only the nuclear-norm minimization as the convex relaxation, as long as the inputs are sufficiently diverse and obey $N \gtrsim \log^2(LN)$. By ``diverse inputs", we mean that the $\xx_n$'s belong to different, generic subspaces. To our knowledge, this is the first theoretical result on blind deconvolution where the subspace to which $\ww$ belongs is not fixed, but needs to be determined.

	We discuss the result in the context of multipath channel estimation in wireless communications. Both the fading coefficients, and the delays in the channel impulse response $\ww$ are unknown. The encoder codes the $K$-dimensional message vectors randomly and then transmits coded messages $\xx_n$'s over a fixed channel one after the other. The decoder then discovers all of the messages and the channel response when the number of samples taken for each received message are roughly greater than $(K+S\log^2S)\log^4(LN)$, and the number of messages is roughly at least $\log^2(LN)$.

\end{abstract}

 \section{Introduction}
 This paper addresses the problem of recovering a vector $\ww$ from its circular convolutions individually with a series of unknown vectors $\{\xx_n\}_n := \xx_1, \xx_2, \ldots, \xx_N$. Consider a linear, time-invariant (LTI) system, characterized by some unknown impulse response $\ww$. The system is driven by a series of inputs $\{\xx_n\}_n$ and one wants to identify the system by observing only the outputs, which in this case are the convolutions of the inputs with the system impulse response. This problem is referred to as the blind system identification: jointly discover the inputs, and the system impulse response from the outputs, and is one of the core problems in the field of system theory and signal processing. When $\ww$ is expected to be sparse, the problem can be recast in a now standard fashion as the recovery of a simultaneously sparse and rank-1 matrix. We relax this formulation by dropping the sparsity contraint and using nuclear-norm minimization. 
 
 We then leverage results in the well understood area of low-rank recovery from underdetermined systems of equations to give the conditions on the unknown impulse response, and inputs under which they can be deconvolved exactly. 
 
 Roughly, the results say the input vectors $\xx_1,\xx_2,\ldots,\xx_N$, each of which lives in some known ``generic" $K$-dimensional subspace of $\reals^L$, and a vector $\ww \in \reals^L$ that is ``incoherent" in the Fourier domain, and is only assumed to be $S$-sparse is some known basis, are separable with high probability provided $K + S \sim L$, up to log factors, and with appropriate coherences appearing in the constants.

More precisely, we state the problem as follows. Assume that each of the input $\xx_n$ lives in a known $K$-dimensional subspace of $\reals^L$, i.e.,
 \begin{align}\label{eq:x=Cm}
 \xx_n = \vm{C}_n\mm_n,\quad n = 1,2,3,\ldots,N
 \end{align}
for some $L \times K$ basis matrix $\vm{C}_n$ with $K \leq L$, whose columns span the subspace in which $\xx_n$ resides. Moreover, the vector $\ww$ is only assumed to be $S$-sparse in an $L \times L$ basis matrix $\vm{B}$, i.e., 
 \begin{align}\label{eq:w=Bh}
 \ww = \vm{B} \hh, ~\mbox{such that~} \|\hh\|_0 \leq S.
 \end{align}
(It can be convenient to think of $\vm{B}$ as the identity upon first reading.) Given the basis matrices $\vm{B}$ and $\vm{C}_n$, all we need to know are the expansion coefficients $\mm_n$, and $\hh$ to discover the inputs $\xx_n$ for each $n$, and $\ww$. The structural assumptions on $\ww$ are much weaker than on $\xx_n$, in that we only need $\ww$ to be sparse in some known basis, whereas each $\xx_n$ resides in a generic and known subspace.

We observe the circular convolutions:
 \begin{align}\label{eq:circular-convs}
 \vm{y}_n = \ww*\xx_n, ~ n = 1, 2,3,\ldots, N.
 \end{align}
 An entry in the length-$L$ observation vector $\vm{y}_n$, for each $n \in \{1,2,3,\ldots, N\}$ is 
 \begin{align*}
 y_n[\l] = \sum_{\l^\prime = 1}^{L}w[\l^\prime]x_n[\l-\l^\prime + 1 ~\mbox{mod}~ L ],~~(\o,n) \in \{1,2,3,\ldots, L\}\times \{1,2,3,\ldots,N\},
 \end{align*}
 where modulo $L$ is what makes the convolution circular. Given no information about the inputs $\xx_n$ and the impulse response $\vm{w}$, it is clear that both of these quantities cannot be uniquely identified from the observations \eqref{eq:circular-convs}. 
 
 We want to put this result in perspective from the outset by comparing it with a related result in \cite{ahmed2012blind}, where a single ($N=1$) input blind deconvolution problem is analyzed. Mathematically, the main result in \cite{ahmed2012blind} shows that in a single input deconvolution problem, the vectors $\xx$, and $\ww$ can be recovered from the circular convolution $\vm{y} = \ww*\xx$ when $\xx$ lives in a known generic subspace as above, however, unlike above the incoherent vector $\ww$  also lives in a known subspace. In this paper, we do not have a known subspace assumption on $\ww$, which makes it a significant improvement over the results in \cite{ahmed2012blind} and has concrete implications in important applications as will be explained in Section \ref{sec:app}

 \subsection{Notations}
 
We use upper, and lower case bold letters for matrices and vectors, respectively. Scalars are represented by upper, and lower case, non-bold letters. The notation $\xx^
 *$ ($\xx^\T$) denotes a row vector formed by taking the transpose with (without) conjugation of a column vector $\xx$. By $\bar{\vm{x}}$, we mean a column vector obtained from $\vm{x}$ by conjugating each entry. Linear operators are represented using script letters. We repeatedly use the notation $k \sim_K  n$  to indicate that the index $k$ takes value in the range $\{(n-1)K+1,\ldots,nK\}$ for some scalar $K$. We use $[N]$ to denote the set $\{1,2,3,\ldots,N\}$. The notation $\vm{I}_K$ denotes $K \times K$ identity matrix for a scalar $K$. For a set $\Omega \subset [L]$,  $\vm{I}_{L \times \Omega}$ denotes an $L \times |\Omega|$ submatrix of an $L \times L$ identity obtained by selecting columns indexed by the set $\Omega$. Also we use $\vm{D}_n$ to represent a $KN \times KN$ matrix $\vm{I}_K \otimes \vm{e}_n \vm{e}_n^*$ with ones along the diagonal at locations $k \sim_K n$ and zeros elsewhere, where $\{\vm{e}_n\}_n$ denote standard $N$-dimensional basis vectors, and $\otimes$ is the conventional Kronecker product. We write $\mbox{vec}(\vm{A})$ for the vector formed by stacking the columns of a matrix $\vm{A}$. Given two matrices $\vm{A}$, and $\vm{B}$, we denote by $\vm{A} \boxtimes \vm{B}$, the rank-1 matrix: $[\mbox{vec}(\vm{A})][\mbox{vec}(\vm{B})]^*$. Similarly, $\PO: \comps^{L \times M} \rightarrow \comps^{L \times M}$ for some $L$, and $M$ takes an $L\times M$ matrix $\X$ to $\vm{I}_{L \times \Omega}\vm{I}_{L \times \Omega}^* \X$. We will use $L \sim a \rightarrow b$ to show that a variable $L$ varies between scalar $a$, and $b$. Lastly, the operator $\E$ refers to the expectation operator, and $\mathbb{P}$ represents the probability measure. 
 
\subsection{Lifting and convex relaxation}
In this section, we recast the blind system identification from diverse inputs as a simultaneously sparse, and  rank-1 matrix recovery problem, and set up a semidefinite program (SDP) to solve it. Begin with defining $\vm{F}$, the $L \times L$ discrete Fourier transform (DFT) matrix, 
\begin{align}\label{eq:DFT}
F[\omega,\l] = \frac{1}{\sqrt{L}} \mathrm{e}^{-\mathrm{j}2\pi(\omega-1)(\l-1)/L},\quad(\omega,\l) \in [L]\times[L],
\end{align}
and let $\ff_\l^*$ denote the $\l$th row of $\vm{F}$. In the Fourier domain, the convolutions in \eqref{eq:circular-convs} are
\begin{align*}
\hat{\vm{y}}_n = \sqrt{L} \hat{\vm{w}} \odot \hat{\xx}_n, ~~ \mbox{or} ~~ \hat{y}_n[\l] = \sqrt{L}\<\ff_\l,\vm{w}\>\<\ff_\l,\xx_n\>, \quad (\o,n) \in [L]\times[N], 
\end{align*}
where $\hat{\vm{w}} = \vm{F}\vm{w}$, $\hat{\xx}_n = \vm{F}\xx_n$, and $\odot$ denotes the Hadamard product. Using the fact that $\xx_n = \vm{C}_n \mm_n$, and $\vm{w} = \vm{B}\hh$, we obtain 
\begin{align*}
\hat{y}_n[\l] &= \sqrt{L}\<\vm{B}^*\ff_\l,\hh\>\<\vm{C}_n^*\ff_\l,\mm_n\>= \<\bh_\l,\hh\>\<\mm_n,\cc_{\l,n}\>,\quad (\o,n)\in [L]\times[N],
\end{align*}
where the last equality follows by substituting $\bh_\o = \vm{B}^*\ff_\o$, $\bar{\cc}_{\l,n} = \sqrt{L} \vm{C}_n^*\ff_\l$, and using the fact that $\<\xx,\vm{y}\> = \<\vm{y},\xx\>^*$. This can be equivalently expressed as
\begin{align}\label{eq:meas}
\hat{y}_n[\l] = \<\bh_\l\cc_{\l,n}^*,\hh\mm^*_n\> = \<\bh_\l\vm{\phi}_{\l,n}^{*}, \hh\mm^*\>,\quad (\l,n) \in [L]\times[N],
\end{align}
where with matrices as its arguments, the notation $\<\cdot,\cdot\>$ denotes the usual trace inner product, $\mm = [\mm_1^*,\mm_2^*,\ldots,\mm_N^*]^*$, and $\vm{\phi}_{\l,n}$ denotes a length $KN$ vector of zeros except the $\cc_{\l,n}$ in the position indexed by $k \sim_K n$, i.e., 
\begin{align}\label{eq:phi_ln}
\vm{\phi}_{\o,n} = \vm{c}_{\o,n}\otimes \e_n,\quad(\o,n)\in [L]\times[N]
\end{align}
with $\e_n$ denoting the standard $N$-dimensional basis vectors.
It is clear that the measurements are non-linear in $\hh \in \reals^L$ and $\mm \in \reals^{KN}$ but are linear in their outer product $\X_0 = \hh\mm^*$. Since the expansion coefficients $\hh$ are $S$-sparse, this shows that the inverse problem in \eqref{eq:circular-convs} can be thought of as the question of recovering $\hh\mm^*$; a rank-1 matrix with $S$-sparse columns, from its linear measurements obtained by trace inner products against known measurement matrices $\A_{\o,n} = \bh_\l\vm{\phi}_{\l,n}^*$.

Define a linear map $\cA: \reals^{L \times KN} \rightarrow \comps^{LN}$ as
\begin{align}\label{eq:cA}
\cA(\X) : &=\{\<\bh_\l\vm{\phi}_{\l,n}^*,\X\>~| ~(\l,n)\in[L]\times[N]\}
\end{align}
The number of unknowns in $\X$ are $LKN$ and there are only $LN$ linear measurements available. This means that the linear map $\cA$ is severely underdetermined except in the trivial case when $K=1$. In all other cases when $K > 1$, infinitely many candidate solutions satisfy the measurements constraint owing to the null space of $\cA$. 

Of course, we can take advantage of the fact that the unknown matrix will always be simultaneously sparse, and rank-1 and hence the inherent dimension is much smaller. Information theoretically speaking, the number of unknowns is only $\sim S\log L+KN$, and if we can effectively solve for the simultaneously, sparse and rank-1 matrices then inverting the system of equations for $\xx_n$'s, and $\ww$ might be possible for a suitable linear map $\cA$ when $LN \gtrsim S\log L + KN$. If it were possible, a single unknown input ($N=1$) under certain structural assumptions would suffice to identify the system completely. 

However, it is only known how to individually relax the low-rank and sparse structures \cite{fazel2001rank,candes09ex,recht10gu,candes2006robust,candes2006stable}, namely using nuclear and $\ell_1$ norms, but it remains an open question to efficiently relax those structures simultaneously.\footnote{The most natural choice of combining the nuclear and $\ell_1$ norms to constitute a convex penalty for simultaneously sparse and low-rank is known to be suboptimal \cite{oymak2015simultaneously}. In fact, for the special case of rank-1, and sparse matrices, no effective convex relaxation exists \cite{aghasi2013tightest}. Thus, even if the low-rank and sparse structures can be individually handled with effective convex relaxations, no obvious convex penalty is known for the simultaneously spase, and low-rank structure.} Instead, if we ignore the sparsity altogether and only cater to the rank-1 structure, the problem remains in principle solvable because the inherent number $L+KN$ of unknowns in this case become smaller than the number $LN$ of observations as soon as the number $N$ of inputs exceeds $\frac{L}{L-K} > 1$. Therefore, {\em the main idea of this paper is to use multiple inputs, which allow us to forego the use of a sparsity penalty in the relaxed program.}

Before we formulate the optimization program, it is worth mentioning that the recovery of the rank-1 matrix $\X_0$ only guarantees the recovery  of $\hh$ and $\mm$ to within a global scaling factor $\alpha$, i.e., we can only recover $\tilde{\hh} = \alpha \hh$, and $\tilde{\mm} = \alpha^{-1}\mm$, which is not of much concern in practice. 

The inverse problem in \eqref{eq:circular-convs} can be cast into a rank-1 matrix recovery problem from linear measurements as follows:
\begin{align*}
&\quad\mbox{find}\quad\qquad\X\\
&\mbox{subject to}\quad \hat{y}_n[\l] = \<\bh_\l\vm{\phi}_{\l,n}^*,\X\>, \quad (\l,n) \in [L]\times[N]\\
&~~~~~~~~~~~~~~\mbox{rank}(\X) = 1.
\end{align*}
The optimization program is non convex and in general NP hard due to the combinatorial rank constraint. Owing to the vast literature \cite{fazel02ma,candes09ex,gross11re,recht10gu} on solving optimization programs of the above form, it is well known that a good convex relaxation is 
\begin{align}\label{eq:convex-relaxation}
\hat{\vm{X}} := &~~\underset{\vm{X}}{\argmin}~~\quad\|\X\|_* \\
&\mbox{subject to}\quad\hat{y}_n[\l] = \<\bh_\l\vm{\phi}_{\l,n}^*,\X\>, \quad (\l,n) \in [L]\times[N]\notag,
\end{align}
where the nuclear norm $\|\X\|_*$ is the sum of the singular values of $\X$. The system identification problem is successfully solved if we can guarantee that the minimizer to the above convex program equals $\hh\mm^*$. Low-rank recovery from under determined linear map has been of interest lately in several areas of science and engineering, and a growing literature \cite{candes09ex,gross11re,recht10gu,recht11si} has been concerned with finding the properties of the linear map $\cA$ under which we can expect to obtain the true solution after solving the above optimization program.

\subsection{Main results}\label{sec:Main_results}

In this section, we state the main result claiming that the optimization program in \eqref{eq:convex-relaxation} can recover the sparse, and rank-1 matrix $\hh\mm^*$ almost always when the inputs $\xx_n$'s reside in relatively dense ``generic'' $K$ dimensional subspaces of $\reals^L$, and that $\ww \in \reals^L$ satisfies the nominal conditions of $S$-sparsity in some known basis, and ``incoherence" in the Fourier domain. Before stating our main theorem, we define the terms ``generic" and ``incoherence" concretely below. Recall that $\vm{w} = \vm{B} \hh$.

The incoherence of the basis $\hat{\vm{B}} = \vm{F}\vm{B}$ introduced in \eqref{eq:w=Bh} is quantified using a coherence parameter $\mu_{\max}^2$,
\begin{align}\label{eq:coherence_B}
\mu_{\max}^2 : = L \cdot \|\hat{\vm{B}}\|_\infty^2,
\end{align}
where $\| \cdot \|_{\infty}$ is the entrywise uniform norm. Using the fact that $\hat{\vm{B}}$ is an $L \times L$ orthonormal matrix, it is easy to see that $1 \leq \mu_{\max}^2 \leq L$. A simple example of a matrix that achieves minimum $\mu_{\max}^2$ would be the DFT matrix. 

The incoherence of $\ww$ in the Fourier domain is measured by $\mu_0^2$,
\[
\mu_0^2 : = L \cdot \max \left\{ \frac{\| \hat{\vm{B}} \hh \|^2_\infty}{\| \hh \|_2^2} , \frac{\| \hat{\vm{B}} \hh'_{n,p} \|^2_\infty}{\| \hh'_{n,p} \|_2^2} , \frac{\| \hat{\vm{B}} \hh''_{n,n'} \|^2_\infty}{\| \hh''_{n,n'} \|_2^2} \right\},
\]
where each ratio is a measure of diffusion in the Fourier domain. The spirit of the definition is mainly captured by the first term, $L \cdot \frac{\| \hat{\vm{B}} \hh \|^2_\infty}{\| \hh \|_2^2}$ --- scaled peak value of $\ww$ in the Fourier domain.  The other terms involve quantities $\hh'_{n,p}$ and $\hh''_{n,n'}$ that are defined in the sequel (they are random perturbations of $\hh$), and are only present for technical reasons. Notice that the first term is small, $\mathcal{O}(1)$, when $\ww$ is diffuse in the frequency domain, and can otherwise be as large as $L$.

To keep our results as general as possible, we introduce an extra incoherence parameter $\rho_0^2$ that quantifies the distribution of energy among the inputs $\{\mm_n\}_n$, and is defined as
\begin{align}\label{eq:coherence_m}
\rho_0^2 := N\cdot \max_{n} \frac{\|\mm_n\|_2^2}{\|\mm\|_2^2}, 
\end{align}
which is bounded as $1 \leq \rho_0^2 \leq N$.  The coherence $\rho_0^2$ achieves the lower bound when the energy is equally distributed among the inputs, and the upper bound is attained when all of the energy is localized in one of the inputs, and the rest of them are all zero. 

As mentioned earlier, we want each of the inputs to reside in  some ``generic" $K$-dimensional subspace, which we realize by choosing $\vm{C}_n$'s to be iid Gaussian matrices, i.e.,
\begin{align}\label{eq:Cn}
C_n[\o,k] \sim \text{Normal}\big(0,\tfrac{1}{L}\big)~~\forall (\l,n,k) \in [L]\times[N]\times [K].
\end{align}
A  ``generic" $K$-dimensional subspace refers to most of the $K$-dimensional subspaces in the entire continuum of $K$-dimensional subspaces of $\mathbb{R}^L$, however, one must also be mindful that such generic subspaces may not arise naturally in applications, and may have to be introduced by design as will be demonstrated in a stylized channel-estimation application in Section \ref{sec:app}. 

Ultimately, we are working with the rows $\vm{c}_{\l,n}$'s of the matrix $\sqrt{L}\vm{F}\vm{C}_n$ as defined in \eqref{eq:meas}. As the columns of $\vm{C}_n$ are real and $\vm{F}$ is an orthonormal matrix, the columns of $\vm{F}\vm{C}_n$ are also Gaussian vectors with a conjugate symmetry. Hence, the rows $\vm{c}_{\l,n}$ are distributed as\footnote{The construction in \eqref{eq:cl-construction} is explicitly for even $L$ but can be easily adapted to the case when $L$ is odd.}
\begin{align}\label{eq:cl-construction}
\cc_{\l,n} &= \begin{cases}
\mbox{Normal}(0,\vm{I}) & \l = 1,(L/2)+1,~n \in [N]\\
\mbox{Normal}(0,2^{-1/2}\vm{I})+j\mbox{Normal}(0,2^{-1/2}\vm{I})& \l = 2,\ldots, (L/2),~n\in [N]
\end{cases}\\
\cc_{\l,n} &= \bar{\vm{c}}_{L-\l+2,n},~~ \l = (L/2)+2,\ldots, L,~n\in[N].\notag
\end{align}
Note that the vectors $\vm{c}_{\o,n}$'s are independently instantiated for every $ n \in [N]$. On the other hand, the vectors $\vm{c}_{\o,n}$ are no longer independent for every $\o \in [L]$, rather the independence is retained only for $\o \in \{1, 2,\ldots,L/2+1\}$. However, the $\vm{c}_{\o,n}$'s are still uncorrelated for $\forall \o \in [L]$; a fact which is crucial in the analysis to follow later. We are now ready to state the main result. 

\begin{thm}\label{thm:main}
	Suppose the bases $\{\vm{C}_n\}_{n=1}^N$ are constructed as in \eqref{eq:cl-construction} , and the coherences $\mu_{\max}^2$, $\mu_0^2$, and $\rho_0^2$ of the basis matrix $\vm{B}$, and the expansion coefficients $\hh$, and $\{\mm_n\}_{n=1}^N$ are as defined above. Furthermore, to ease the notation, set 
	 \[
     \alpha_1 = \log(K\log(LN)), \quad \text{and}\quad \alpha_2 = \log(S\log(LN)).
     \]
     Then for a fixed $\beta \geq 4$, there exists a constant $C^\prime_\beta = \mathcal{O}(\beta)$, such that if 
	\[
	\max(\mu_0^2\alpha_1 K, \mu_{\max}^2S \alpha_2\log^2S) \leq \frac{L}{C_{\beta}^\prime \alpha_1 \log^2(LN)}~\mbox{and}~ N \geq C^\prime_{\beta}\rho_0^2\alpha_1\log(LN),
	\]
	then $\X_0 = \hh(\mm_1^*,\mm_2^*,\dots,\mm_N^*)$ is the unique solution to  \eqref{eq:convex-relaxation} with probability at least $1- \mathcal{O}((LN)^{4-\beta})$, and we can recover $N$ inputs $\{\xx_n\}_{n=1}^N$ and $\ww$ (within a scalar multiple) from $N$ convolutions $\{\vm{y}_n = \ww * \xx_n\}_{n=1}^N$.
\end{thm}

The result above crudely says that in an $L$ dimensional space, an incoherent vector, $S$-sparse in some known basis, can be separated successfully almost always from $N$ vectors (with equal energy distribution) lying in known random subspaces of dimension $K$ whenever $ K+S\log^2S \lesssim L /\log^4(LN)$, and $N \gtrsim \log^2(LN)$.

\subsection{Application: Blind channel estimation using random codes}\label{sec:app}

A stylized application of the blind system identification directly arises in multipath channel estimation in wireless communications. The problem is illustrated in Figure \ref{fig:channel_estimation}. A sequence of length-$K$ messages $\mm_1,\mm_2,\ldots, \mm_N$ are coded using taller $L \times K$ coding matrices $\vm{C}_1,\vm{C}_2,\ldots, \vm{C}_N$, respectively. The coded messages $\xx_n = \vm{C}_n\mm_n$, $n \in 1,2,3,\ldots, N$ are then transmitted one after the other over an unknown multipath channel, characterized by a sparse impulse response $\ww \in \reals^L$.  The transmitted message $\xx_n$ arrives at the receiver through multiple paths. Each path introduces its own delay and fading. All the delayed and scaled copies of $\xx_n$ overlap in the free space communication medium. The received signal is modeled as the convolution of $\xx_n$ with $\ww$. This action is repeated with same delay and fading coefficients for every $\xx_n$. In other words, we are assuming here that the channel's impulse response is more or less fixed over the duration of the transmission of these $N$ coded messages, which justifies the use of a fixed impulse response $\ww$ in each of the convolutions. The task of the decoder is to discover both the impulse response and the messages by observing their convolutions $\vm{y}_n = \ww * \xx_n, \ n = 1,2,3,\ldots, N$, and using the knowledge of the coding matrices.

Our main result in Theorem \ref{thm:main} took $\ww$ as a vector that is sparse in some incoherent basis $\vm{B}$. In the application discussed in the last paragraph, we can simply take the basis $\vm{B}$ to be the standard basis; perfectly incoherent. The location of each non-zero entry in $\ww$ depicts the delay in the arrival time of  a copy of coded message at the receiver from a certain path and the value of the entry known as the fading coefficient incorporates the attenuation and the phase change encountered in that path. The coherence parameter $\mu_0^2$ is roughly just the peak value of the normalized frequency in the spectrum of the channel response. For this particular application, we can assume that $\rho_0^2 \approx 1$ as the transmitter energy is equally distributed among the message signals. Our results prove that if each of the message is coded using a random coding matrix, and the channel response has approximately a flat spectrum, then we can recover the messages and the channel response jointly almost always by solving \eqref{eq:convex-relaxation}, whenever the length $K$ of the messages, the sparsity $S$ of the channel impulse response $\ww$, and the codeword length $L$ obey $ K+S\log^2S\lesssim L/\log^4 (LN)$, and the number $N$ of messages that convolve with the same instantiation of the channel roughly exceed $ \log^2(LN)$.

Our results here can be thought of as an extension to the blind deconvolution result that appeared in \cite{ahmed2012blind}, where we only have a one-time look at the unknown channel --- we observe only a single convolution of the impulse response with a randomly coded message. Consequently, only fading coefficients could be resolved in \cite{ahmed2012blind} and not the delays in the impulse response $\ww$ of the channel. In other words, one needs to know the subspace or support of $\ww$ in advance. In general, both fading coefficient, and delays are equally important pieces of information to decipher the received message in wireless communications.  In this paper, we take  advantage of several looks at the same channel as it remains fixed during the transmission of $N$ messages. This enables us to estimate both the fading coefficients and the unknown delays at the same time. In general, we do not assume that the vector $\ww$ lives in a known subspace as was the case in \cite{ahmed2012blind}.

\begin{figure}
	\begin{center}
		\includegraphics[trim = 2cm 8cm 2cm 3cm, scale = 0.65 ]{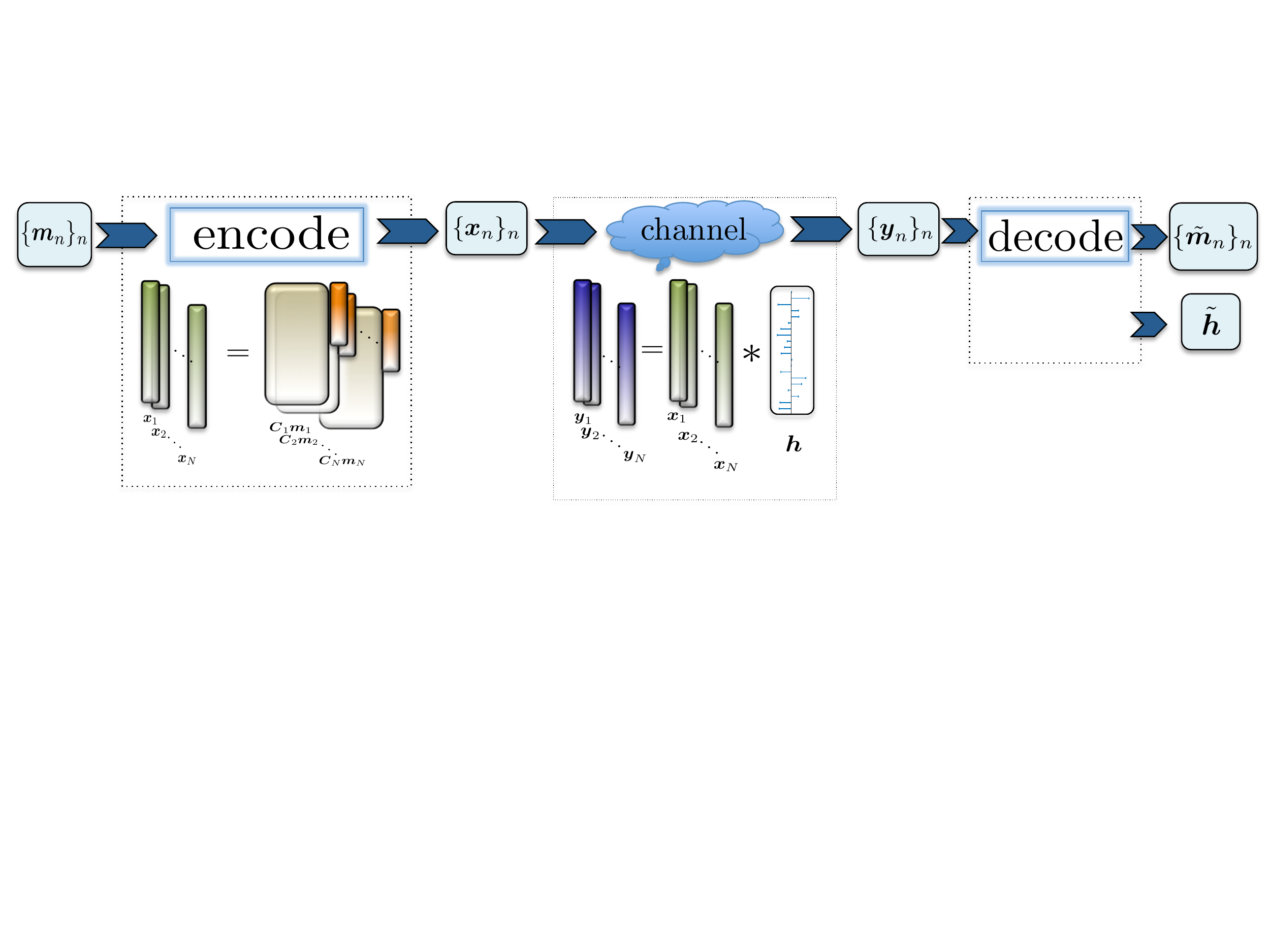} 
	\end{center}
	\caption{\small\sl Blind channel estimation. Each meassage $\mm_n$ in the block $\{\mm_n\}_n = \{\mm_1,\mm_2,\ldots,\mm_N\}$ of messages is coded with a corresponding tall coding matrix $\vm{C}_n$ and the block of coded messages $\{\xx_n\}_n$ is sequentially transmitted over an arbitrary unknown channel. This results in convolution of each of the coded messages $\{\xx_n\}_n$ with an unknown impulse response $\hh$. The decoder receives the convolutions $\{\vm{y}_n\}_n$ and it discovers both the messages $\{\mm_n\}_n$ and the unknown channel $\hh$ (within a global scalar). }
	\label{fig:channel_estimation}
\end{figure}
\subsection{Related work}\label{sec:related-work}

In a nutshell, and to our knowledge, this paper is the first in the literature to theoretically deal with an impulse response that belongs to a low-dimensional subspace that is not fixed ahead of time but needs to be discovered.

The lifting strategy to linearize the bilinear blind deconvolution problem was proposed in \cite{ahmed2012blind}, and it was rigorously shown that two convolved vectors in $\reals^L$ can be separated blindly if their $K$, and $S$ dimensional subspaces are known and one of the subspace is generic and the other is incoherent in the Fourier domain. It is further shown using the dual certificate approach in low-rank matrix recovery literature \cite{candes09ex,gross11re,recht11si} that both the vectors can be deconvolved exactly when $(K+S) \lesssim L/\log^3L$.  This paper extends the single input blind deconvolution result to multiple diverse inputs, where we observe the convolutions of $N$ vectors with known subspaces with a fixed vector only known to be sparse in some known basis. 

A natural question that arises is whether multiple ($N>1$) inputs $\xx_n$'s are necessary in our problem to identify $\ww$ in \eqref{eq:w=Bh}. The answer is no in this specific case as even in the single input case $N=1$, under the same random subspace assumption on $\xx_1$, and replacing the nuclear norm in \eqref{eq:convex-relaxation} with the standard $\ell_1$ norm (sum of absolute entries) will separate $\xx_1$, and $\ww$, however, the sample complexity $L$ will be suboptimal, and of the order of $SK$ to within log factors. In the general single input case; under no random subspace assumption on $\xx_1$, it is shown in \cite{choudhary2014sparse} that $\ww$, and $\xx_1$ are not identifiable from $\vm{y} = \ww*\xx_1$. 

A related question, in a sense dual to that presented in the previous section, is multichannel blind deconvolution. See Figure \ref{fig:Multi_channel}. In discrete time this problem can be modeled as follows. An unknown noise source $\vm{w} \in \reals^L$ feeds $N$ unknown multipath channels characterized by $K$-sparse impulse responses $\xx_n \in \reals^L, \ n = 1,2,3, \ldots, N$.  The receiver at each channel observes several delayed copies of $\ww$ overlapped with each other, which amounts to observing the convolutions $\vm{y}_n = \ww * \xx_n,\ n = 1,2,3,\ldots, N$. The noise $\ww$ can be modeled as a Gaussian vector, and is well dispersed in the frequency domain, i.e., the vector $\ww$ is incoherent according to the definition \eqref{eq:coherence_h}. The fading coefficients of the multipath channels are unknown, however, we assume that the delays are known. This amounts to knowing the subspace of the channels and the unknown impulse responses can be expressed as $\xx_n = \vm{C}_n\vm{m}_n$ for every $n \in [N]$, where the columns of the known $L \times K$ coding matrices are
now the trivial basis vectors and $\mm_n$ contain $K$ unknown fading coefficients in each channel. The indices (delays) of the non-zeros of every impulse response $\xx_n$ can be modeled as random, in which case the coding matrices are composed of the random subset of the columns of the identity matrix. With the coding matrix known and random, the multichannel blind deconvolution problem is in spirit the dual of the blind system identification from diverse inputs presented in this paper, where the roles of the channel and the source signal are reversed. However, the results in Theorem \ref{thm:main} are explicitly derived for dense Gaussian coding matrices and not for random sparse matrices. It is worth mentioning here that in many practical situations the non zeros in the channel impulse response are concentrated in the top few indices making the assumption of known subspaces (delays) plausible.

After \cite{ahmed2012blind}, a series of results on blind deconvolution appeared under different sets of assumptions on the inputs. For example, the result in \cite{bahmani2015lifting} considers an image debluring problem, where the receiver observes the $N$ subsampled circular convolutions of an $L$-dimensional image $\xx$, modulated with random binary waveforms, with an $L$-dimensional bandpass blur kernel $\hh$ that lives in a known $K$-dimensional subspace.  Then it is possible to recover both the image and a incoherent blur kernel using lifting and nuclear norm minimization, whenever $N \gtrsim \log^3 L \log \log K$, where $N/L$ is also the number of subsampling factor of each convolution. The result shows that it is possible to deconvolve two unknown vectors by observing multiple convolutions---each time one of the vectors is randomly modulated and is convolved with the other vector living in a known subspace. We are also observing multiple convolutions but one of the vectors in the convolved pair is changing every time and the subspace of the other is also unknown, this makes our result much broader.  

Another relevant result is blind deconvolution plus demixing \cite{ling2015blind}, where one observes sum of $N$ different convolved pairs of $L$-dimensional vectors lying in $K$, and $S$ dimensional known subspaces; one of which is generic and the other is incoherent in the Fourier domain. Each generic basis is chosen independently of others. The blind deconvolution plus the demixing problem is again cast as a rank-$N$ matrix recovery problem. The algorithm is successful when $N^2(K+S) \lesssim L / \log^4 L$. 

An important recent article from the same group settles the recovery guarantee for a regularized gradient descent algorithm for blind deconvolution, in the single-input case and with the scaling $K + S \lesssim L / \log^2 L$ \cite{li2016rapid}. This result, however, makes the assumption of a fixed subspace for the sparse impulse response. Note that gradient descent algorithms are expected to have much more favorable runtimes than semidefinite programming, when their basin of attraction can be established to be wide enough, as in \cite{li2016rapid}.

The multichannel blind deconvolution was first modeled as a rank-1 recovery problem in \cite{romberg2013multichannel} and the experimental results show the successful joint recovery of Gaussian channel responses with known support that are fed with a single Gaussian noise source.  Other interesting works include \cite{xu1995least,gurelli1995evam}, where a least squares method is proposed. The approach is deterministic in the sense that the input statistics are not assumed to be known though the channel subspaces are known. Some of the results with various assumptions on input statistics can be found in \cite{tong1994blind}. Owing to the importance of the blind deconvolution problem, an expansive literature is available and the discussion here cannot possibly cover all the related material, however, an interested reader might start with the some nice survey articles \cite{levin2011understanding,liu1996recent,tong1998multichannel} and the references therein. 

It is also worth mentioning here a related line of research in the phase recovery problem from phaseless measurements \cite{candes2013phaselift,candes2015phase}, which happen to be quadratic in the unknowns. As in bilinear problems, it is also possible to lift the quadratic phase recovery problem to a higher dimensional space, and solve for a positive-definite matrix with minimal rank that satisfies the measurement constraints. 

\begin{figure}[!ht]
	\begin{center}
		\includegraphics[trim = 2cm 8cm 2cm 0cm, scale = 0.55 ]{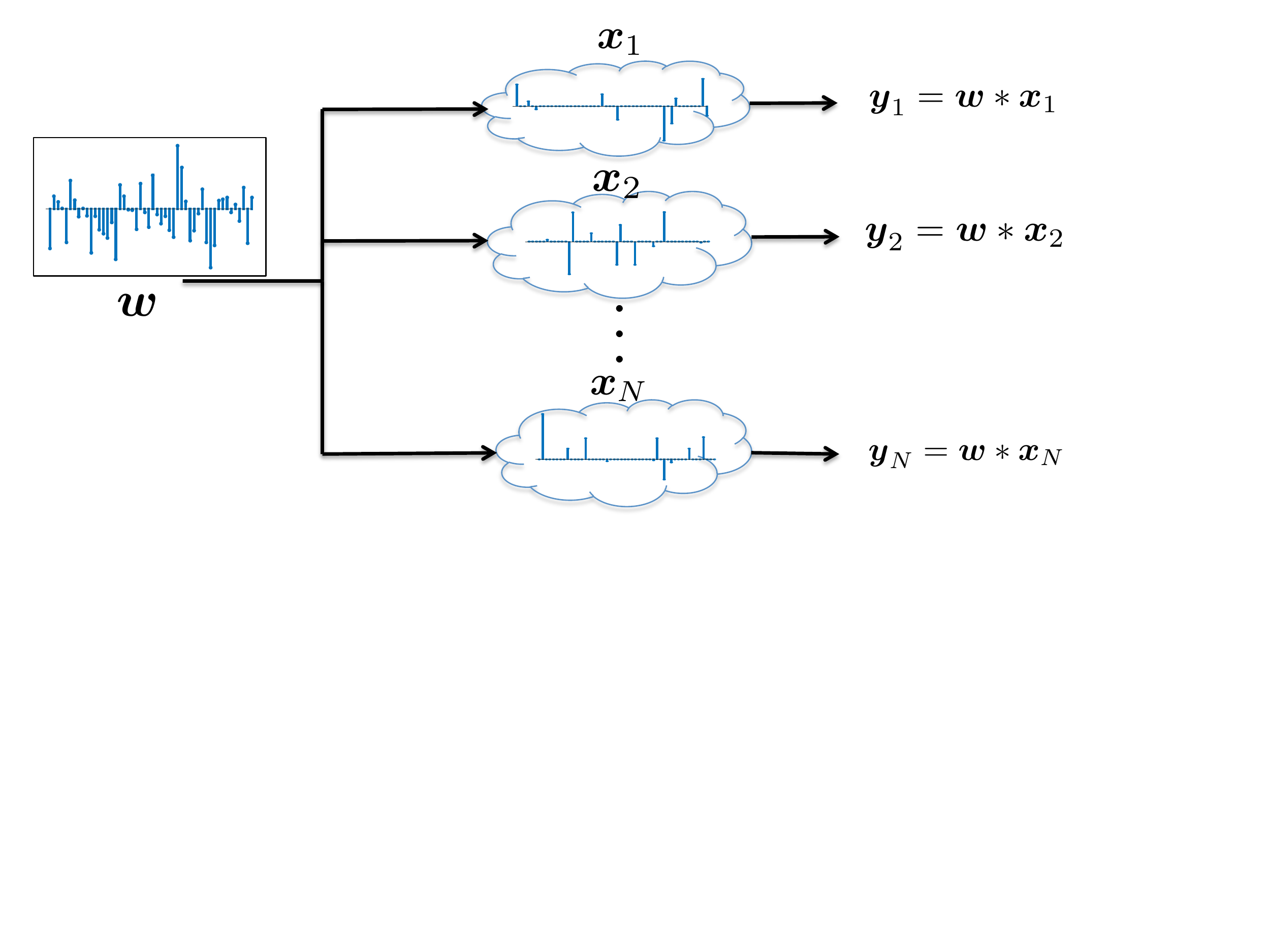} 
	\end{center}
	\caption{\small\sl Blind multichannel estimation. An unknown noise source $\vm{w}$ feeds  $N$ unknown multipath channels characterized by sparse impulse responses $\{\xx_1,\xx_2,\ldots, \xx_N\}$. We observe the convolutions at the receivers and the task is to recover the channel responses together with the noise signal. The problem can be thought of as the dual to the blind channel estimation problem where the roles of channels and the source signals are reversed: a fixed incoherent vector is now fed into all the channels. The channel impulse responses can be reliably modeled with Bernoulli Gaussian distribution. }
	\label{fig:Multi_channel}
\end{figure}
\section{Numerical Simulations}\label{sec:Exps}
As an alternative to the computationally expensive semidefinite program in \eqref{eq:convex-relaxation}, we rely on a heuristic non-linear program: 
\begin{align}\label{eq:non-linear-program}
\hat{\vm{H}}, \hat{\vm{M}} := ~&\underset{\vm{H}, \vm{M}}{\argmin}\quad \|\vm{H}\|_{\F}^2 + \|\vm{M}\|_{\F}^2\\
&\text{subject to} \quad \hat{y}_n[\ell] = \<\vm{b}_{\ell}\vm{\phi}_{\ell,n}^*, \vm{H}\vm{M}^*\> , \quad (\ell,n) \in [L]\times[N]\notag,
\end{align}
which solves for matrices $\vm{H} \in \comps^{L \times R}$, and $\vm{M} \in \comps^{KN \times R}$. The semidefinite constraint in \eqref{eq:convex-relaxation} is always satisfied under the substitution $\vm{X} = \vm{H}\vm{M}^*$. The non-linear program was proposed in \cite{burer2003nonlinear}, and the results therein showed that  all the local minima of \eqref{eq:non-linear-program} are the global minima of \eqref{eq:convex-relaxation} when $R > \mbox{rank}(\hat{\vm{X}})$, where $\hat{\vm{X}}$ is the optimal solution of \eqref{eq:convex-relaxation}. Since in our case the optimal solution $\hh\mm^*$ is rank-1, we solve \eqref{eq:non-linear-program} with $R = 2$, and declare recovery when $\hat{\vm{H}}$, and $\hat{\vm{M}}$ are rank deficient. The best rank-1 approximation of $\hat{\vm{H}}\hat{\vm{M}}^*$ constitutes the solution of \eqref{eq:convex-relaxation}.  The non-linear program considerably speeds up the simulations as instead of operating in the lifted space like \eqref{eq:convex-relaxation} with $LKN$ variables involved, it operates almost in the natural parameter space with much fewer number $2(L+KN)$ of variables. We use an implementation of LBFGS  available in \cite{MINFUNC} to solve \eqref{eq:non-linear-program}. An additional advantage of \eqref{eq:non-linear-program} is that no suitable initialization is required.  Comparatively, the recently proposed gradient descent scheme \cite{li2016rapid} for bilinear problems not only requires to solve a separate optimization program to initialize well but also the gradient updates involve additional unnatural regularizer to control the incoherence.

%
%
%

We present phase transitions that validate the sample complexity results in Theorem \ref{thm:main}. The shade in the phase transitions represents the probability of failure determined by counting the frequency of failures in twenty five experiments for each pixel in the phase transitions. We classify the recovered solution $\hat{\vm{X}}$ as a failure if $\|\hat{\vm{X}}-\hh\mm^*\|_{\F} > 10^{-1}$.

In all of the phase transitions, we take $\ww$, and $\mm$ to be Gaussian vectors. Observe that $\ww$ is constructed to be a dense vector with no sparse model. Recall that Theorem \ref{thm:main} restricts $\ww$ to be a sparse vector, however, our simulation results show successful recovery in a more general case of dense $\ww$. This observation is in conformation with our belief that the sparsity assumption on $\ww$ is a result of merely a technical requirement due to the proof method. Very similar phase transitions can be obtained under restrictive sparse model on $\ww$. 

We will present two sets of phase transitions. Each set contains three phase transition diagrams; in each diagram, we fix one of the variables $L$, $K$, and $N$, and vary the other two in small increments and compute the probability of failure every time as outlined earlier in this section.

 In the first set, we mimic the channel estimation problem discussed in Section \ref{sec:app}, and shown in Figure \ref{fig:channel_estimation}. We take $\vm{C}_n$'s to be Gaussian matrices as in \eqref{eq:Cn}. Figure \ref{fig:RandomGauss1} shows that for a fixed $N = 40$, we are able to recover all of the inputs $\xx_n$'s, and $\ww$ as soon as $L \geq 10K$. The phase diagrams in Figure \ref{fig:RandomGauss2} and \ref{fig:RandomGauss3} mainly show that the performance of the algorithm become roughly oblivious to the number $N$ of inputs as soon as $N \geq 10$ for the particular range of $K$, and $L$ considered in the phase transition diagrams. 

In the second set shown in Figure \ref{fig:RandomEye}, we simulate the blind channel estimation problem discussed in Section \ref{sec:related-work}, and shown in Figure \ref{fig:Multi_channel}. This set contains similar phase diagrams as in first set under the same assumptions, the only difference is that the matrices $\vm{C}_n$'s are now the random subsets of the columns of identity. In other words, we take the support of $\xx_n$'s to be random and known. The results are almost exactly the same as in the first set. 

\begin{figure}[!ht]
	\centering
	\subfigure[]{
		\includegraphics[trim = 4cm 5cm 0cm 1.9cm, angle = 90, scale = 0.367]{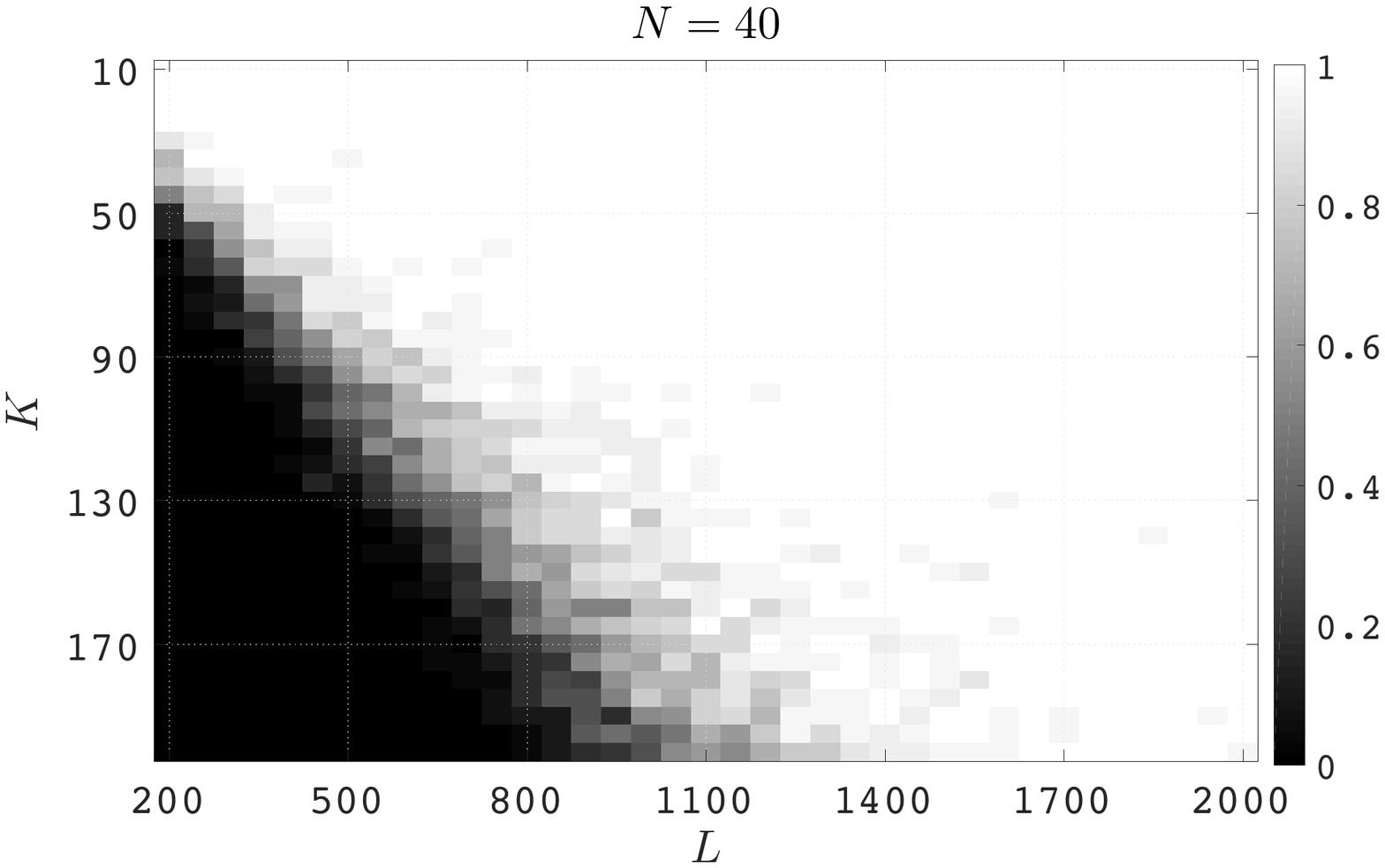}
		\label{fig:RandomGauss1}}
	\subfigure[]{
		\includegraphics[ trim = 4cm 4.5cm 0cm 1cm, angle = 90, scale = 0.36]{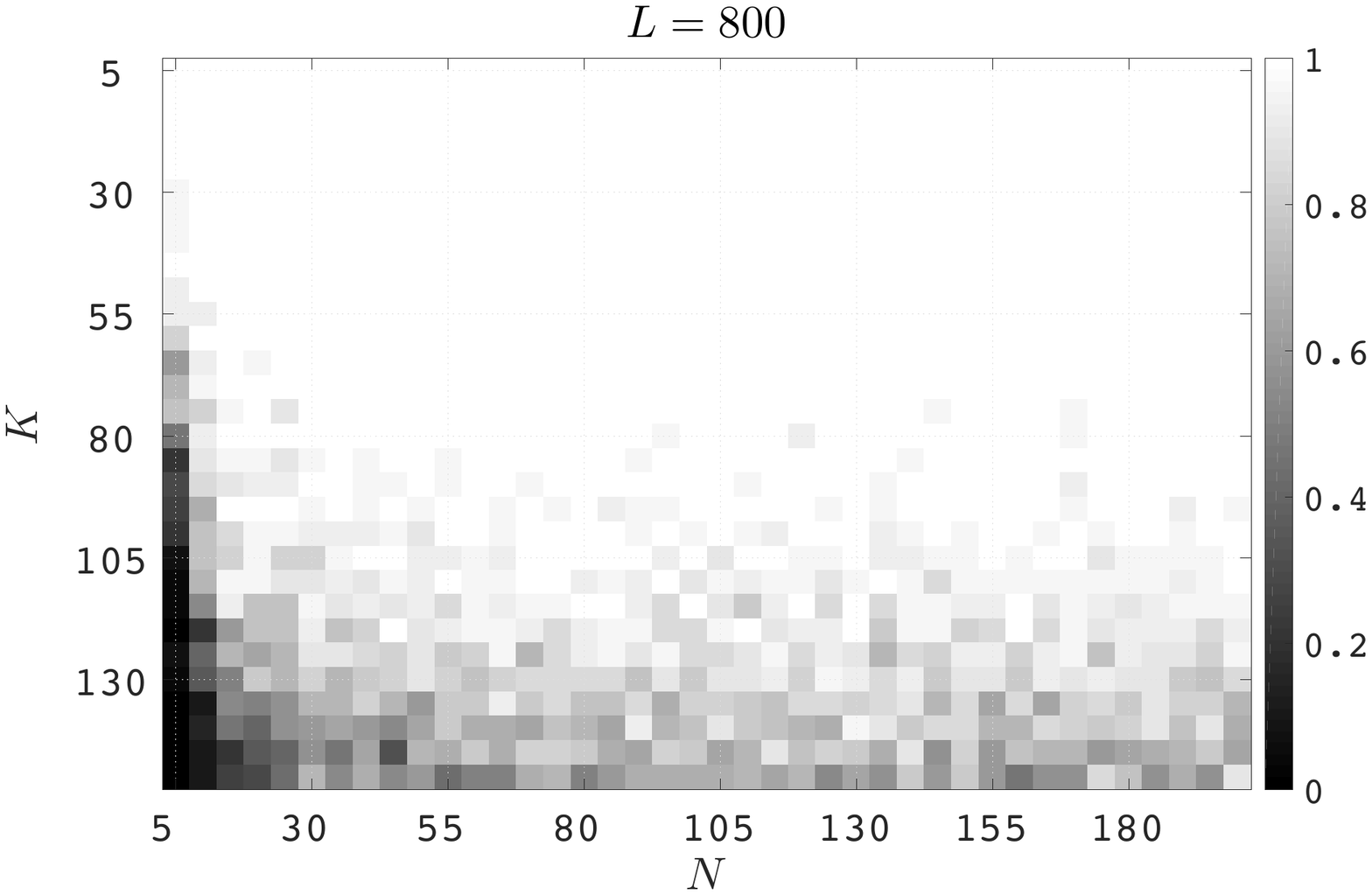}
		\label{fig:RandomGauss2}}
	\subfigure[]{
		\includegraphics[trim = 4cm 4.5cm 4cm -1.5cm, angle = 90, scale = 0.42]{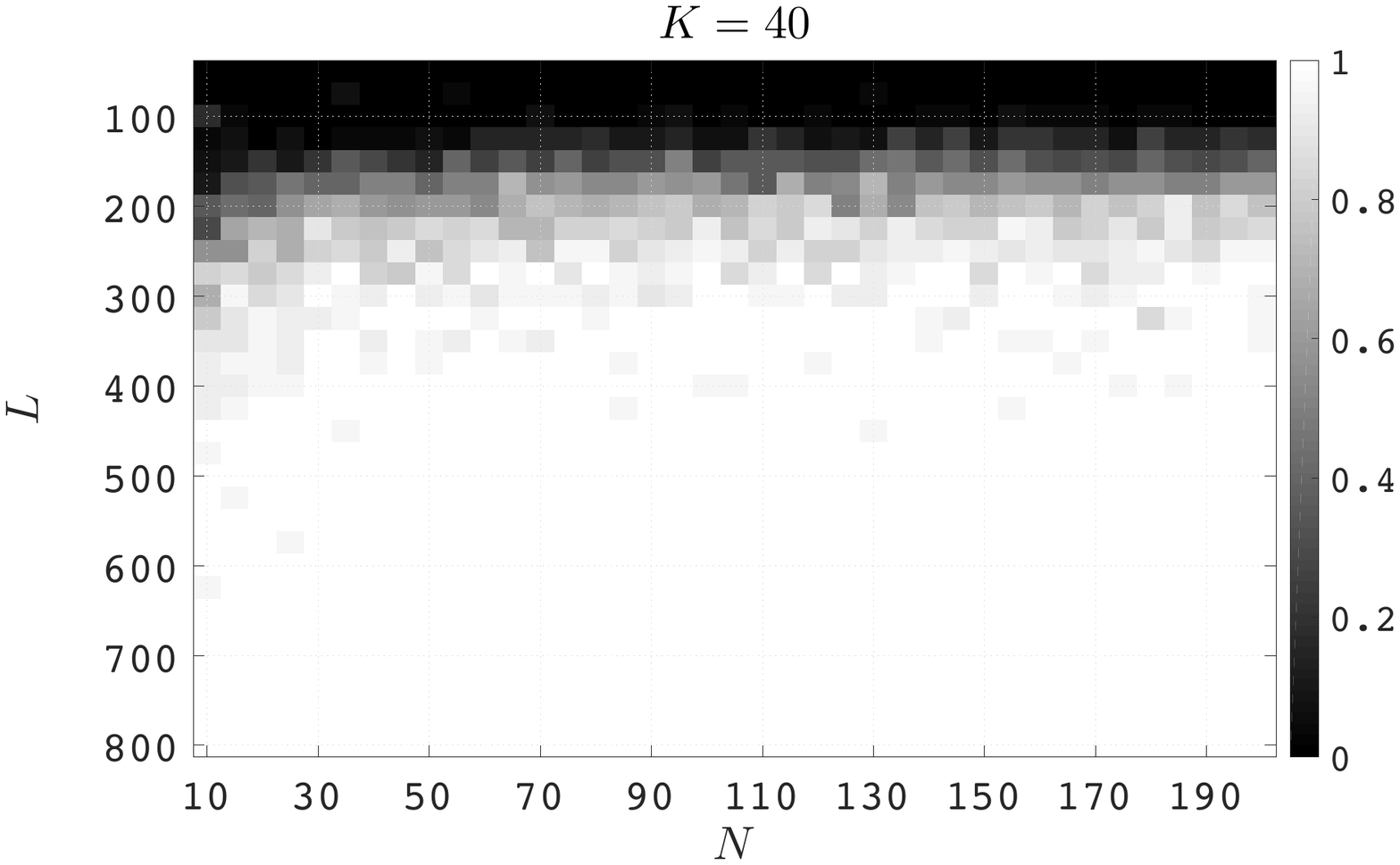}
		\label{fig:RandomGauss3}}
	\caption{\small\sl Empirical success rate for the deconvolution of  $\ww$ and $\xx_1,\xx_2,\ldots,\xx_N$.  Recall that $\xx_n = \vm{C}_n\mm_n$ for every $n = 1,2,3,\ldots, N$. In these experiments, the vectors $\ww$, $\mm_n$ are Gaussian, and $L \times K$ matrices $\vm{C}_n, n =1,2,3,\ldots,N$ are also independent, and Gaussian.  (a) Fix $N = 40$, and vary $L \sim 200 \rightarrow 2000$, $K \sim 10 \rightarrow 190$. Successful reconstruction is obtained with probability one when $L \geq 10K$. (b) Fix $L = 800$, and vary $K \sim 5 \rightarrow 150$, $N \sim 5 \rightarrow 200$. Successful reconstruction is obtained with probability one when $K \leq L/10$, and $N \geq 10$. (c) Fix $K = 40$, and vary $L \sim 50\rightarrow 800$, $N \sim 10 \rightarrow 200$. Successful reconstruction occurs with probability one when $L \geq 10K$, and $N \geq 10$.}
	\label{fig:RandomGauss}
\end{figure}
\begin{figure}[!ht]
	\centering
	\subfigure[]{
		\includegraphics[trim= 2cm 6.5cm -1cm 4.5cm, scale = 0.4]{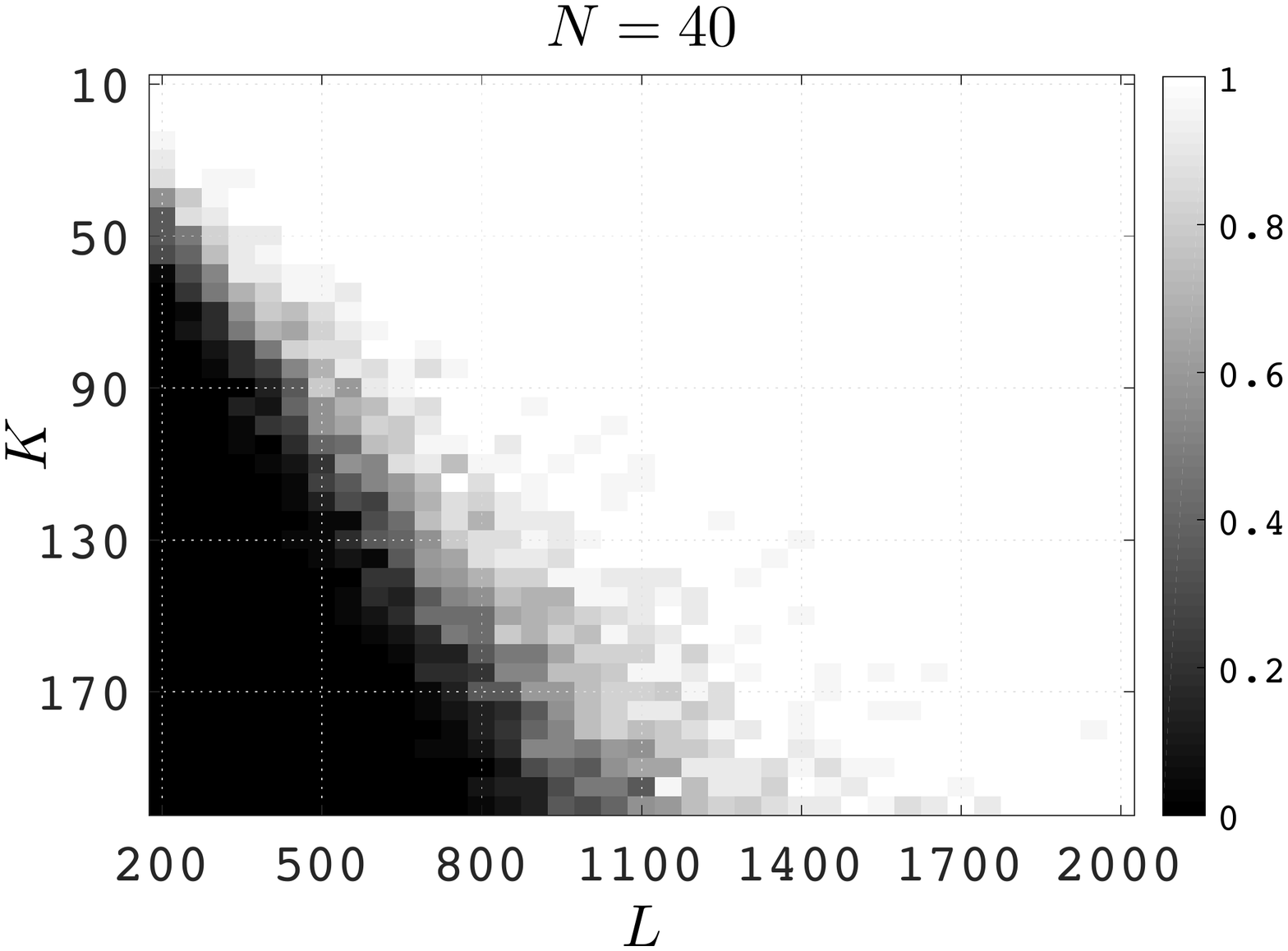}
		\label{fig:RandomEye1}}
	\subfigure[]{
		\includegraphics[trim= 0cm 6.5cm 2.5cm 4.5cm,scale = 0.4]{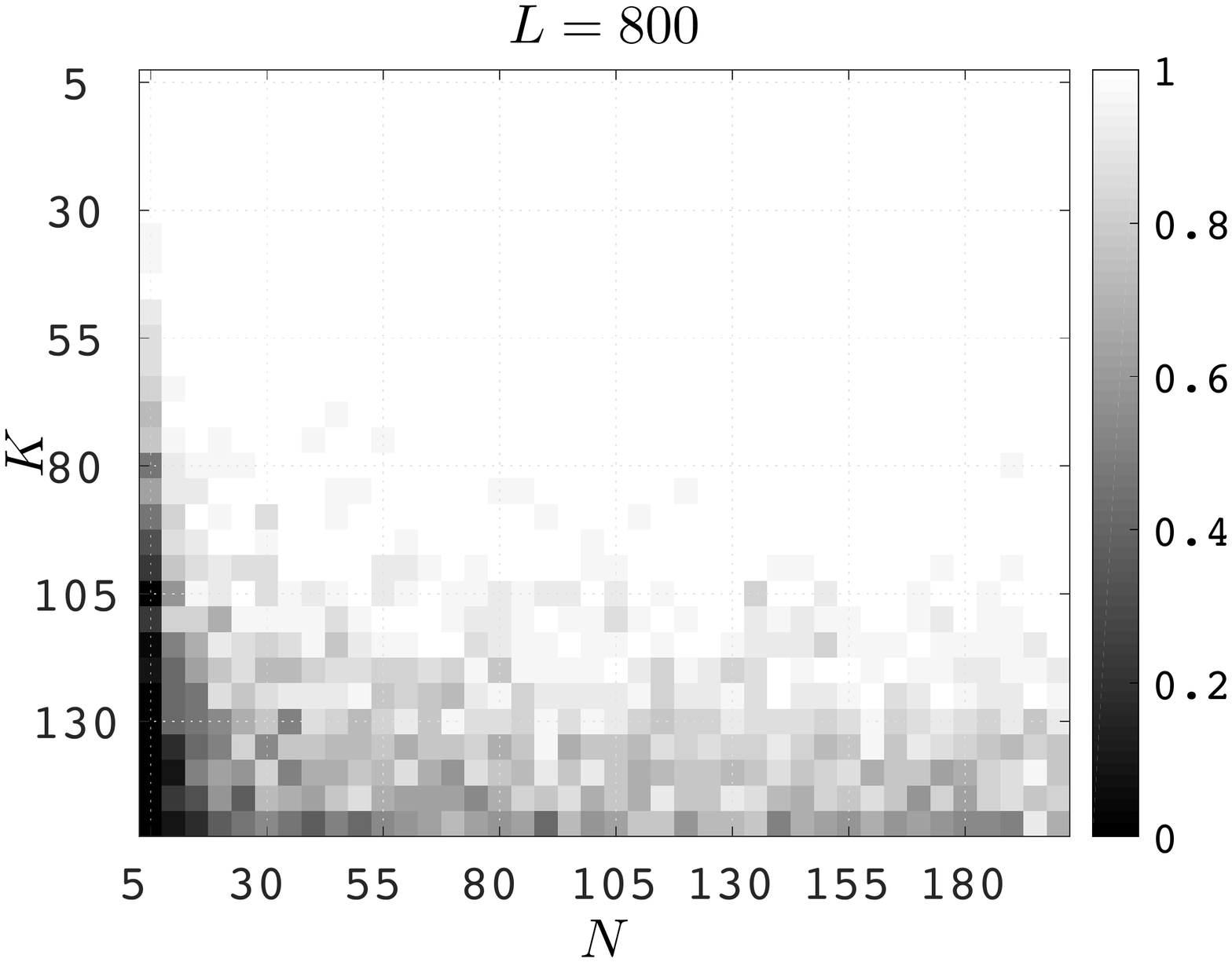}
		\label{fig:RandomEye2}}
	\subfigure[]{
		\includegraphics[trim=0.7cm 6.5cm 2.5cm 6.5cm,scale = 0.4]{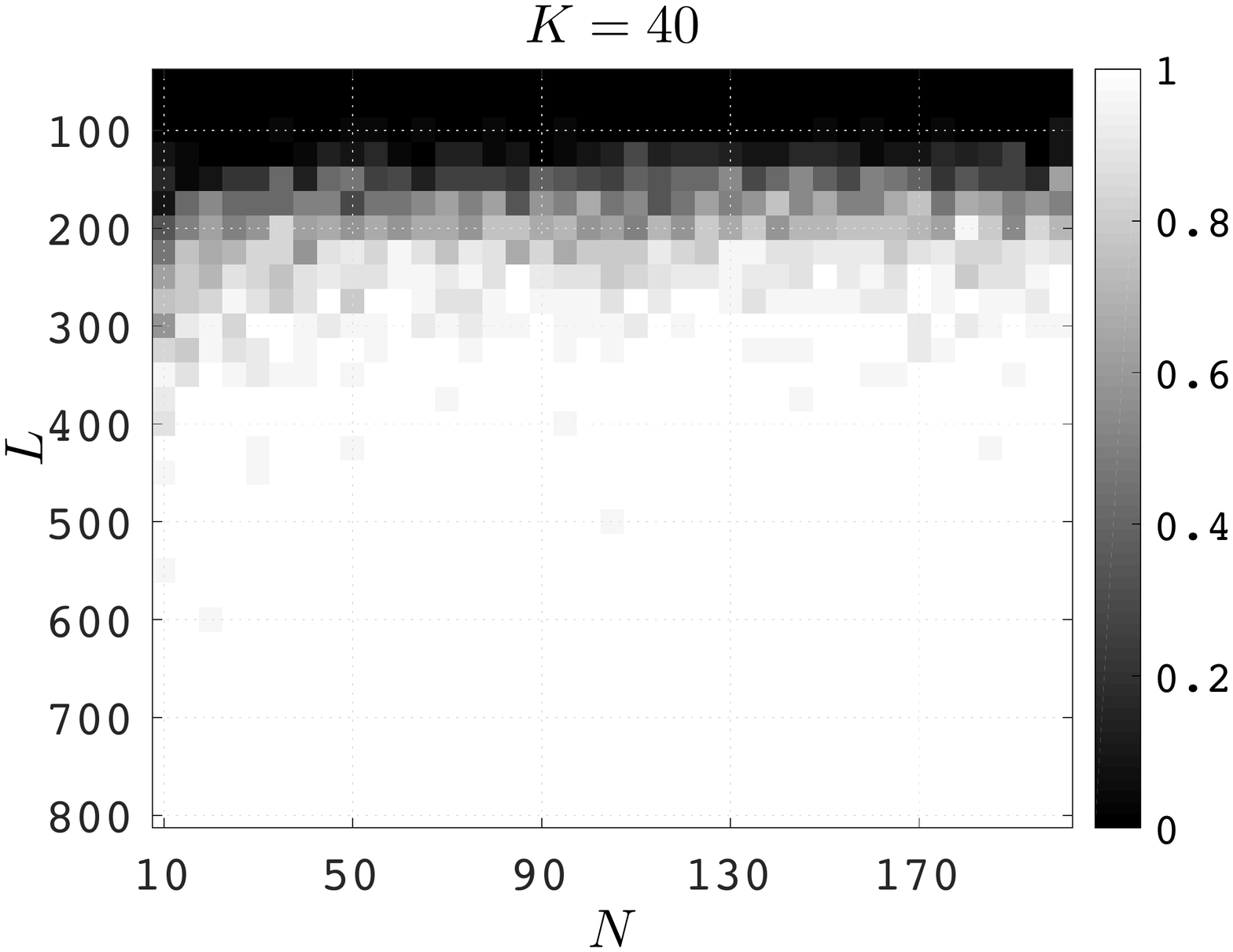}
		\label{fig:RandomEye3}}
	\caption{\small\sl Empirical success rate for the deconvolution of  $\ww$ and $\xx_1,\xx_2,\ldots,\xx_N$.  In these experiments, the vector $\ww$ is Gaussian, and every $\xx_n$ is a sparse vector with random support and its $K$ non-zero entries are Gaussian.  (a) Fix $N = 40$, and vary $L \sim 200 \rightarrow 2000$, $K \sim 10 \rightarrow 190$. Successful reconstruction is obtained with probability one when $L \geq 10K$. (b) Fix $L = 800$, and vary $K \sim 5 \rightarrow 150$, $N \sim 5 \rightarrow 200$. Successful reconstruction is obtained with probability one when $K \leq L/10$, and $N \geq 10$. (c) Fix $K = 40$, and vary $L \sim 50\rightarrow 800$, $N \sim 10 \rightarrow 200$. Successful reconstruction occurs with probability one when $L \geq 10K$, and $N \geq 10$.}
	\label{fig:RandomEye}
\end{figure}

\section{Proof of Theorem \ref{thm:main}}
We are observing the linear measurements as in \eqref{eq:meas} of an unknown $\vm{X}_0 = \hh\mm^*$ such that $\hh$ is an $S$-sparse vector. Define $\Omega: = \mbox{supp}(\hh)$, where $|\Omega| = S$. To show that the solution to the SDP in \eqref{eq:convex-relaxation} equals $\X_0$ with high probability, we establish the existence of a  valid dual certificate \cite{candes09ex,gross11re}. The proof of low-rank recovery using dual certificate method is a standard approach now and has been employed in the literature \cite{ahmed2012blind,recht11si,gross11re} many times before. 
Our construction of the dual certificate uses the golfing scheme \cite{gross11re}, but is unusually technical in that there is a probabilistic dependence between the iterates, which in turn precludes the use of matrix concentration inequalities.

  
  Let $\uu \in \reals^L$, and $\vv \in \reals^{KN}$ be arbitrary vectors, and $\hh$, and $\mm$ be as defined earlier. Let $T_1$ be the linear space of matrices with rank at most two defined as
  \[
  T_1 :=\{\X|\X= \alpha \hh\vv^*+\beta \uu\mm^*, ~\mbox{and}~ \alpha, \beta \in \reals\},
  \]
  and $T_2$ be the space of matrices with rows supported on index set $\Omega$, and is defined as 
 \[
  T_2 : =\{ \X~ | ~\X \in \reals^{L \times KN}, X[\ell, j] = 0 ~\mbox{for} ~ \ell \in \Omega^\perp\},
\]  
where $ \Omega^\perp = [L]\backslash \Omega$. Then we have
\[
T_1\cap T_2:=\{\X|\X= \alpha \hh\vv^*+\beta \uu\mm^*,  \uu_{\Omega^\perp} = \vm{0}, ~\mbox{and}~ \alpha, \beta \in \reals\}.
\]
Note that the matrix of interest $\X_0$ is a member of the space $T_1 \cap T_2$. Let us now define the related projection operators. We start by defining $\PO$ that takes a matrix or a vector with $L$ rows and sets all the rows that are not indexed by the index set $\Omega \subset [L]$ to zero. Mathematically, we can define the projection on the index set $\Omega$ as
  \begin{equation}\label{eq:POmega}
  	\PO(\X):= \vm{I}_{L \times \Omega} \vm{I}_{L \times \Omega}^*\X
  \end{equation}
  where $\vm{I}_{L \times \Omega}$ denotes the submatrix of the $L \times L$ identity matrix with columns indexed by set $\Omega$. The orthogonal projector $\PT$ onto $T_1 \cap T_2$ is then defined as
\begin{equation}\label{eq:PT-def}
\PT(\z) := \hh\hh^*\z + \PO(\z\mm\mm^*)  - \hh\hh^*\z\mm\mm^*,
\end{equation}
and the projector $\PTc$ onto the orthogonal complement $T_1^{\perp}\cup T_2^\perp$ of $T_1\cap T_2$ is then simply $\PTc(\z) = \vm{Z}-\PT(\vm{Z})$. Note that in the definition of the projection above, we assume without loss of generality that $\|\mm\|_2 = \|\hh\|_2 = 1$ as the optimality conditions, presented in Lemma \ref{lem:optimality-conditions} below, for the success of nuclear norm minimization \ref{eq:convex-relaxation}  only involve normalized $\vm{h}$, and $\vm{m}$.

The following lemma gives sufficient conditions on a dual certificate $\vm{Y} \in \text{Range}(\cA^*)$ to guarantee that nuclear-norm minimization program in \eqref{eq:convex-relaxation} produces $\X_0 = \hh\mm^*$ as the solution. The proof of the lemma is almost exactly the same as in \cite{candes09ex}, the only difference is that now instead of just working with a space of rank-2 matrices $T_1$, we are dealing with the space $T_1 \cap T_2$ of rank-2 and row-sparse matrices. We repeat the proof here to show that all the details in \cite{candes09ex} also work out for the space $T_1\cap T_2$.
\begin{lem}[Optimality Conditions]\label{lem:optimality-conditions}
	Let $\cA$ be as defined in \eqref{eq:cA}, and $\gamma > 0$ be a positive number such that $\|\cA\|\leq \gamma$, and 
	\begin{equation}\label{eq:uniqueness-cond2}
	\sqrt{2}\|\cA\PT(\vm{Z})\|_\F  \geq \|\PT(\vm{Z})\|_\F
	\end{equation}
	for all $\vm{Z} \in \mbox{Null}(\cA)$.  Then the matrix $\X_0= \hh\mm^*$ is the unique minimizer of \eqref{eq:convex-relaxation} if there exists a $\vm{Y} \in \mbox{Range}(\cA^*)$ such that 
	\begin{equation}\label{eq:uniqueness-cond}
	\|\hh\mm^*-\PT(\vm{Y})\|_\F \leq \frac{1}{4\gamma},  \qquad	 \|\PTc(\y)\| \leq \frac{1}{2}.
	\end{equation}
\end{lem}
\begin{proof}
	Let $\hat{\vm{X}}$ denote the solution to the optimization program in \eqref{eq:convex-relaxation}. This implies that $\|\hat{\vm{X}}\|_* \leq \|\X_0\|_*$. Given this, it is enough to show
	\[
	\|\X_0 + \vm{Z}\|_* > \|\X_0\|_*, ~ \forall \vm{Z} \in \mbox{Null}(\cA), ~\vm{Z} \neq \vm{0},
	\]
	where $\vm{Z} = \hat{\vm{X}}-\vm{X}_0$, to establish exact recovery, since the two conflicting requirements on $\hat{\vm{X}}$ above would directly mean that $\hat{\vm{X}} = \vm{X}_0$, or $\vm{Z} = \vm{0}$. 
	
	The sub-differential of nuclear norm at point $\x_0$ is (see \cite{watson1992characterization} for details)
	\[
	\partial\|\X_0\|_* := \{ \hh\mm^*+ \w : \PTc(\w) = \w,~ \mbox{and} ~ \|\w\| \leq 1 \}.
	\]
	Since by the definition of sub-differentials
	\[
	\|\X_0+\vm{Z}\|_* -\|\X_0\|_* \geq \<\vm{\Delta}, \vm{Z}\>, ~ \mbox{ for every } \vm{\Delta} \in \partial \|\X_0\|_*,
	\]
	we obtain 
	\begin{align*}
	\|\X_0+\vm{Z}\|_* -\|\X_0\|_* &\geq \< \hh\mm^*, \vm{Z}\> + \< \PTc(\w), \vm{Z}\>\\
	& =\< \hh\mm^*-\y, \vm{Z}\> + \< \w, \PTc(\vm{Z})\>, ~  \mbox{ for every } \y \in \mbox{Range}(\cA^*).
	\end{align*}
	 Using the fact that  $\hh\mm^*\in T_1 \cap T_2$, and also maximizing the inner product $\< \w, \PTc(\vm{Z})\>$ with respect to $\|\w\| \leq 1$ gives us 
	\begin{align*}
	\|\X_0+\vm{Z}\|_* -\|\X_0\|_* &\geq \< \hh\mm^*-\PT(\y), \PT(\vm{Z})\> - \<\PTc(\y), \PTc(\vm{Z})\> + \|\PTc(\vm{Z})\|_* \\
	 & \geq - \|\hh\mm^* -\PT(\y)\|_\F \|\PT(\vm{Z})\|_\F -\|\PTc(\y)\|\|\PTc(\vm{Z})\|_* + \|\PTc(\vm{Z})\|_*.
	\end{align*}
	Now the inequality $\|\cA\PT(\vm{Z})\|_\F \geq 2^{-1/2} \|\PT(\vm{Z})\|_\F$ for any $\vm{Z} \in \mbox{Null}(\cA)$ implies that 
	\begin{align*}
	0 = \|\cA(\vm{Z})\|_\F &\geq  \|\cA\PT(\vm{Z})\|_\F - \|\cA\PTc(\vm{Z})\|_\F  \\
	& \geq 2^{-1/2} \|\PT(\vm{Z})\|_\F - \gamma \|\PTc(\vm{Z})\|_\F.
	\end{align*}
	The above inequality implies firstly that $\|\vm{Z}\|_\F^2 \leq (2\gamma^2+1)\|\PTc(\vm{Z})\|_\F^2$, which in turn means $\PTc(\vm{Z}) \neq \vm{0}$ whenever $\vm{Z} \neq \vm{0}$;  secondly, $\|\PT(\vm{Z})\|_\F \leq \sqrt{2}\gamma \|\PTc(\vm{Z})\|_*$.  Using these results, the bound in the earlier inequality gives
	\[
	\|\X_0+\vm{Z}\|_* - \|\X_0\|_* \geq \left(-\|\hh\mm^* - \PT(\y) \|_\F \sqrt{2}\gamma - \|\PTc(\y)\| +1 \right) \|\PTc(\vm{Z})\|_*.
	\]

	Now under the conditions in \eqref{eq:uniqueness-cond}, the right hand side above is strictly positive, which means $\|\X_0 + \vm{Z}\|_* > \|\X_0\|_*$; enough to exhibit the uniqueness. 
\end{proof}
The next lemma provides an upper bound $\gamma$ on the operator norm of the linear map $\cA$. 
\begin{lem}[Operator norm of $\cA$]\label{lem:cA-operator-norm}
	Let $\cA$ be as defined in \eqref{eq:cA}. For any $\beta \geq 1$, 
	\[
	\|\cA\| \leq \sqrt{\beta K \log (LN)}
	\]
	with probability at least $1-(LN)^{-\beta+2}$.
\end{lem}
\begin{proof}
	The operator norm can be calculated by using the fact that $\<\A_{\o_1,n_1} ,\A_{\o_2,n_2}\> = 0$,~ $\forall \o_1 \neq \o_2 \in [L]$, or $\forall n_1 \neq n_2 \in [N]$. This implies that $\|\cA\| = \max_{\o,n}\|\A_{\o,n}\|_\F$. We can write $\|\A_{\o,n}\|^2_\F= \|\bh_\o\|^2_2 \|\vm{\phi}_{\o,n}\|^2_2 =  \sum_{k \sim_K n}|c_{\o,n}[k]|^2$, where $|c_{\o,n}[k]|^2$ are chi-squared random variables with degree 2 when $ 2 \leq \o \leq L/2$, and degree 1  when $\o=1$, or $\o = L/2+1$. In both cases $\E |c_{\o,n}[k]|^2 = 1$, and 
	\begin{align*}
	\P{|c_{\o,n}[k]|^2 > u } \leq \mathrm{e}^{-u}.
	\end{align*}
	The maximum is taken over $(L/2+1)\cdot2\cdot KN$ unique $c_{\o,n}[k]$ and 
	\begin{align*}
	\P{\max_{\o,n,k} |c_{\o,n}[k]|^2 \geq u} \leq  (L+2)KN\cdot \mathrm{e}^{-u} \Rightarrow \P{\max_{\o,n}\|\A_{\o,n}\|_\F^2 \geq Ku} \leq (LN)^2 e^{-u},
	\end{align*} 
	where the last inequality follows from the fact that $K \leq L$.
	Choose $u = \beta\log (LN)$, which gives $\|\cA\| \leq \gamma : = \sqrt{\beta K \log(LN)} $ with probability at least $1-(LN)^{-\beta+2}$. 
\end{proof}
In the following section, we focus on constructing a dual certificating using golfing scheme,  which is then shown to satisfy the uniqueness conditions in the lemma above.

\subsection{Linear operators on golfing partition}\label{sec:Golfing}
To prove the uniqueness conditions in \eqref{eq:uniqueness-cond}, we use a dual certificate $\vm{Y}$ constructed using a variation of the golfing scheme \cite{gross11re}. To this end, we  partition the index set $[L]\times [N]$ into $P$ disjoint sets $\Gamma_1,\Gamma_2,\ldots,\Gamma_P$ defined as
\begin{align}\label{eq:Gammap}
\Gamma_p := \left\{ (\o,n) ~|~ (\o,n) \in \{(\Delta_{n,p},n)\}_{n=1}^N\right\},
\end{align}
where $\Delta_{n,p}$ is a subset of $[L]$, chosen uniformly at random for every $n$ and $p$, such that $|\Delta_{n,p}| = Q = L/P$, and for every $n$, $\Delta_{n,p} \cap \Delta_{n,p^\prime} = \emptyset$ for $p \neq p^\prime$, and $\bigcup_p\Delta_{n,p} = [L]$. The parameter $Q$ is adjusted through the proof, and we assume here without loss of generality that it is an integer\footnote{We are assuming here that $P$ is a factor of $L$; this can be achieved in the worst case by increasing the number $L$ of measurements in each convolution by no more than a factor 2, which only affects the measurements bounds in Theorem \ref{thm:main} by a multiplicative constant.}. In words, the partition of the set $\{1,2,3,\ldots, L\} \times \{1,2,3,\ldots, N\}$ is obtained by dividing $\{1,2,3,\ldots,L\}$ into $P$ randomly chosen disjoint sets for a given $n$, and then we repeat this process independently for every $n$ to obtain a total of $PN$ sets $\Delta_{1,1},\ldots, \Delta_{N,1},\Delta_{1,2}, \ldots, \Delta_{N,P}$. We then define $\Gamma_p$ as in \eqref{eq:Gammap}.  For every $n$, the disjointness among the sets $\Delta_{n,1}, \Delta_{n,2},\ldots, \Delta_{n,P}$ is of critical importance as it ensures that no dependence arises due to reuse of the same $\vm{\phi}_{\o,n}$ in different partitioned sets. We define a linear map $\cA_p: \reals^{L \times KN} \rightarrow \comps^{QN}$ that returns the measurements indexed by $(\o,n) \in \Gamma_{p}$:
\begin{align}\label{eq:cAp}
\cA_p(\vm{Z}) := \{\<\bh_\o\vm{\phi}_{\o,n}^*,\vm{Z}\>: (\o,n) \in \Gamma_p\}, ~~~~~~\cA_p^*\cA_p(\vm{Z}) := \sum_{(\o,n) \in \Gamma_p}\<\bh_\o\vm{\phi}_{\o,n}^*,\vm{Z}\>\bh_\o\vm{\phi}_{\o,n}^*.
\end{align}
In tandem with this partitioning of the measurements, we also require to partition the $L$ rows of $\hat{\vm{B}}$ into $P$ sets of $Q \times L$ submatrices that behave roughly as an isometry on the sparse vectors. Quantitatively, we want the rows $\bh_\o^*$ of matrix $\hat{\vm{B}}$ in each of the sets $\Delta_{n,p}$ to obey 
\begin{align}\label{eq:RIP}
(1-\delta)\|\vm{z}\|_2^2 \leq \frac{L}{Q}\sum_{\o \in \Delta_{n,p}} |\bh_\o^*\vm{z}|^2 \leq (1+\delta)\|\vm{z}\|_2^2
\end{align}
for all vectors $\vm{z}$ supported on the set $\Omega$ such that $|\Omega| \leq S$. A reader familiar with compressive sensing will readily recognize this as the restricted isometry property (RIP) \cite{candes2008restricted} on the submatrices of $\hat{\vm{B}}$. A result \cite{rudelson2008sparse} from compressive sensing says that for each submatrix of rows $\bh_\o$  with coherence $\mu^2_{\max}$, defined in \eqref{eq:coherence_B} if the index set $\Delta_{n,p}$ is chosen uniformly at random, then there exists a constant $C$ such that for any $0 < \delta < 1$, and $0 < \epsilon < 1$,
\begin{align*}
Q \geq C \delta^{-1}\mu_{\max}^2\left(\frac{S \log L}{\epsilon^2}\right)\log \left(\frac{S\log L}{\epsilon^2}\right)\log^2 S
\end{align*}
implies that the RIP in \eqref{eq:RIP} holds with probability exceeding $1-e^{-c(\delta/\epsilon)^2}$. Given the partition $\{\Delta_{n,p}\}_{n,p}$ is chosen uniformly at random, the above result with $\delta = 1/4$, and $ \epsilon^{-2} = \beta\log (LN) $ means that if 
 \begin{align}\label{eq:QRIPbound}
 Q \geq C\beta\mu_{\max}^2 S\log^2 (LN)\log (\beta S\log (LN)) \log^2 S,
 \end{align}
 then 
 \begin{align*}
 \sup_{|\Omega| \leq S}\sum_{\o \in \Delta_{n,p}} |\bh_\o^*\vm{z}|^2 \leq \frac{5Q}{4L}\|\vm{z}\|_2^2, 
 \end{align*}
 with probability at least $1-\mathcal{O}((LN)^{-\beta})$, for all $\vm{z}$ supported on $\Omega$ for a given $\Delta_{n,p}$. A simple union bound over all $PN$ number of sets $\{\Delta_{n,p}\}_{n,p}$ shows that 
 \begin{align}\label{eq:RIPresults}
 \max_{n,p}\left[\sup_{|\Omega| \leq S}\sum_{\o \in \Delta_{n,p}} |\bh_\o^*\vm{z}|^2 \right]\leq \frac{5Q}{4L}\|\vm{z}\|_2^2 \implies \max_{n,p}\left\|\sum_{\o \in \Delta_{n,p} } \PO \bh_{\o}\bh_{\o}^* \PO - \frac{Q}{L}\PO\right\| \leq \frac{Q}{4L} 
 \end{align}
 holds with probability at least $1-(PN)\mathcal{O}((LN)^{-\beta})\geq 1-\mathcal{O}((LN)^{1-\beta})$. In the rest of the article, we take the results in \eqref{eq:RIPresults} as given. 
 
Before constructing a dual certificate, we define some nomenclature. Let 
\begin{align}\label{eq:Snp}
&\vm{S}_{n,p} := \sum_{\o \in \Delta_{n,p}} (\PO\bh_\ell)(\PO \bh_\ell)^*,  ~ \mbox{and} ~  \vm{S}_{n,p}^\ddagger  := \vm{I}_{L \times \Omega}\left(\vm{I}_{L \times \Omega}^*\vm{S}_{n,p}\vm{I}_{L\times \Omega}\right)^{-1}\vm{I}_{L \times\Omega}^*,
\end{align}
where $\vm{I}_{L \times\Omega}$ is an $L \times |\Omega|$ matrix containing the columns of an $L \times L$ identity matrix indexed by set $\Omega$. A direct conclusion of \eqref{eq:RIPresults} gives
\begin{align}\label{eq:S-eigen-bounds}
\max_{n,p}\|\vm{S}_{n,p}\| \leq \frac{5Q}{4L}, ~~ \max_{n,p}\|\vm{S}_{n,p}^\ddagger\| \leq \frac{4L}{3Q}.
\end{align}
In addition, set $\vm{D}_n = \vm{I}_K \otimes \vm{e}_n \vm{e}_n^*.$ A linear operator $\mathcal{S}_p^\ddagger$ is then defined by its action on an $L \times KN$ matrix $\X$ as follows
\begin{equation}\label{eq:cSp}
\mathcal{S}_p^\ddagger(\X) = \sum_n \vm{S}_{n,p}^\ddagger \X \vm{D}_n.
\end{equation}

\subsection{Coherence}
With all the development above, we are now in position to precisely define the coherence parameter $\mu_0^2$ that was first introduced in Section \ref{sec:Main_results}. The diffusion of the impulse response $\ww$ is quantified using the following definition 
\begin{align}\label{eq:coherence_h}
\mu_0^2 : = L \cdot\max \left\{\frac{\|\hat{\vm{B}}\hh\|_\infty^2}{\|\hh\|_2^2}, \frac{Q^2}{L^2}\cdot\max_{n,p}\frac{\|\hat{\vm{B}}\vm{S}_{n,p}^\ddagger\hh\|_\infty^2}{\|\hh\|_2^2},  \max_{n,n^\prime}\frac{\|\hat{\vm{B}}\vm{S}^\ddagger_{n,2}\vm{S}_{n^\prime,1}\vm{h}\|_\infty^2}{\|\hh\|_2^2}\right\}.
\end{align}
The quantities $\hh''_{n,n'}$ and $\hh''_{n,n'}$ from Section \ref{sec:Main_results} can be easily read off from this expression, and from $\ell_2$ norm equivalences resulting from (\ref{eq:S-eigen-bounds})\footnote{So that $\mu_0^2$ as defined here, and its illustration listed in the introduction, are equivalent to within an inconsequential multiplicative factor.}. The following lemma presents upper and lower bounds on $\mu_0^2$.

\begin{lem}[Range of $\mu_0^2$]\label{lem:mu0-range}
	Let $\mu_0^2$, and $\mu_{\max}^2$ be as defined in \eqref{eq:coherence_h}, and \eqref{eq:coherence_B}, respectively. Assume that \eqref{eq:S-eigen-bounds} holds. Then 
	\[
	\frac{5}{3} \leq \mu_0^2 \leq \frac{20}{9}\mu_{\max}^2 S.
	\]
\end{lem}
\begin{proof}
	We assume without loss of generality that $\|\hh\|_2^2 = 1$. Since $\hat{\vm{B}}$ is an orthonormal matrix, it is easy to see that $ 1/L \leq \|\hat{\vm{B}}\hh\|_\infty^2 \leq 1$.  As far as $\|\hat{\vm{B}}\vm{S}_{n,p}^\ddagger \hh\|_\infty^2$ is concerned, an upper bound on it is 
	\begin{align*}
	\max_{n,p}\|\hat{\vm{B}}\vm{S}_{n,p}^\ddagger \hh\|_\infty^2 &= \max_{\o}\max_{n,p} |\bh_\o^*\vm{S}_{n,p}^\ddagger \hh|^2 \leq \max_{\o}\|\bh_\o\|_\infty^2 \max_{n,p}\|\vm{S}_{n,p}^\ddagger \hh\|_1^2\\
	& \leq \mu_{\max}^2 \frac{1}{L} S \|\vm{S}_{n,p}^\ddagger\hh\|_2^2 \leq  \frac{16}{9}\mu_{\max}^2 S\frac{L}{Q^2},
	\end{align*}
	where the first inequality is H\"{o}lder's result; the second one follows from the definition of coherence $\mu_{\max}^2$ in \eqref{eq:coherence_B}, the equivalence of $\ell_1$, and $\ell_2$-norms, and the fact that the vector $\vm{S}_{n,p}^\ddagger \hh$ is $S$-sparse; and the last one is the result of \eqref{eq:S-eigen-bounds}. A lower bound can be obtained by summing over $\o \in [L]$ as follows:
	\[
	L \cdot \max_{\o} \max_{n,p}|\bh_\o^*\vm{S}_{n,p}^\ddagger \hh|^2 \geq \max_{n,p}\sum_{\o}|\bh_\o^*\vm{S}_{n,p}^\ddagger \hh|^2 = \max_{n,p}\|\vm{S}_{n,p}^\ddagger \hh\|_2^2 \geq \frac{16}{25} \frac{L^2}{Q^2},
	\]
	where the equality is due to the fact that $\hat{\vm{B}}$ is an orthonormal matrix, and last inequality follows from \eqref{eq:S-eigen-bounds}. In a similar manner, we can compute the upper and lower bounds on $\|\hat{\vm{B}}\vm{S}_{n,2}^\ddagger \vm{S}_{\nt,1} \hh\|_\infty^2$, and the result is
	\[
	\frac{5}{3}\leq \max_{n,\nt}\|\hat{\vm{B}}\vm{S}_{n,2}^\ddagger \vm{S}_{\nt,1} \hh\|_\infty^2 \leq \frac{20}{9}\mu_{\max}^2 S.
	\]
	Combining all these results, the claim in the lemma follows.
\end{proof}

The spirit of the incoherence is captured by the first term in the maximum; namely, $\|\hat{\vm{B}}\hh\|_\infty^2$, which is small when $\ww$ is diffuse in the frequency domain, and large otherwise. The other two terms are mainly due to technical reasons in the proof presented later. It is a hard question to characterize them exactly, but they are qualitatively expected to be of the same order as the first term because the matrices, $\vm{S}_{n,p}$, and $\vm{S}_{n,p}^{\ddagger}$ are random by construction and are not expected to make the vectors $\vm{S}_{n,p}^\ddagger\hh$, and $\vm{S}_{n,2}^\ddagger\vm{S}_{n^\prime,1}\hh$ more aligned with the rows of $\hat{\vm{B}}$ than $\hh$ was.

\subsection{Construction of a dual certificate via golfing}\label{sec:Golfing1}

We now iteratively build a dual certificate $\vm{Y} \in \mbox{Range}(\cA^*)$ in $P$ iterations with initial value $\vm{Y}_0 = \mathbf{0}$ as follows
\begin{align}
\label{eq:iteration}
\y_1 &= \frac{L}{Q}\cA_1^*\cA_1(\hh\mm^*)~\quad\mbox{and}\quad~\y_p = \y_{p-1}+\cA_p^*\cA_p\mathcal{S}^\ddagger_p\left(\hh\mm^*-\PT(\y_{p-1})\right) \quad 
\mbox{for}~ p \geq 2.
\end{align}

Note that  in \eqref{eq:iteration}, ~$\y_p \in \mbox{Range}(\cA^*)$ for every $p$.
This approach to build the dual certificate was first developed in \cite{gross11re}. Projecting both sides on $T_1 \cap T_2$ results in
\begin{align*}
\PT(\y_1) &= \frac{L}{Q}\PT\cA_1^*\cA_1(\hh\mm^*)\\
\PT(\y_p) &= \PT(\y_{p-1})-\PT\cA_p^*\cA_p\mathcal{S}_p^\ddagger\left(\PT(\y_{p-1})-\hh\mm^*\right), ~ \quad p \geq 2.
\end{align*}
Denoting 
\begin{align}\label{eq:Wp-def}
\w_0 = -\hh\mm^*, \quad \text{and}\quad \w_p :& = \PT (\y_p) -\hh\mm^* ~ \mbox{ for }  p \geq 1
\end{align}
results in a recursion
\begin{align}\label{eq:Wp-def2}
\w_1 = \left(\frac{L}{Q}\PT\cA_1^*\cA_1\PT- \PT\right)\w_0,~ \qquad \w_p = (\PT\cA_p^*\cA_p\mathcal{S}_p^\ddagger\PT-\PT)\w_{p-1} \mbox{ for} ~ p \geq 2,
\end{align}
which in turn implies 
\begin{align}\label{eq:Wp-fro-bound-intermediate}
\w_{p_0} = -\left[\prod_{p=2}^{p_0}(\PT\cA_p^*\cA_p\mathcal{S}_p^\ddagger\PT-\PT)\right]\left[\frac{L}{Q}\PT\cA_1^*\cA_1\PT- \PT\right]\w_0.
\end{align}

Running the iteration \eqref{eq:iteration}  till $p = P$ gives us our candidate for dual certificate $\vm{Y}: = \vm{Y}_P$. To establish that $\X_0$ is the unique solution to \eqref{eq:convex-relaxation}, we need only show that $\|\w_P\|_\F \leq 1/4\gamma$ and $\|\PTc(\vm{Y}_P)\| < 1/2$ in light of \eqref{eq:uniqueness-cond}. 

The Frobenius norm of $\w_P$ is upper bounded by 
\[
\|\w_P\|_\F\leq \left[\prod_{p=2}^P\| \PT\cA_p^*\cA_p\mathcal{S}_p^\ddagger \PT -\PT \|\right] \; \left\| \frac{L}{Q}\PT\cA_1^*\cA_1\PT- \PT \right\| \; \|\w_0\|_\F. 
\]
The difference in the construction of iterates for $p=1$, and $p\geq 2$ is mainly to avoid technical difficulties that arise in the proofs later owing to the dependencies between $\y_{p-1}$ and $\cA_p$ for $p \geq 2$. The dependencies stem from the fact that although for every $p$, the random set $\Delta_{n,p}$ is independent of $\Delta_{\nt,p}$, where $n \neq \nt$. However, for every $n$, the sets $\Delta_{n,p}$, and $\Delta_{n,p^\prime}$ are dependent by construction. This directly implies that the sets $\Gamma_1, \ldots, \Gamma_P$ are dependent; therefore, $\y_{p-1}$, and hence $\w_{p-1}$ are dependent on $\cA_p$ for $p \geq 2$. As is shown in detail in the proofs to follow that the introduction of $\mathcal{S}_p^\ddagger$ for $p \geq 2$ is to avoid this dependence problem and it ensures that
\[
\E \mathcal{R}\cA_p^*\cA_p\mathcal{S}_p^\ddagger(\w_{p-1}) = \w_{p-1},
\]
which is of critical importance in controlling some of the random quantities in the proofs to follow. For $p=1$, however, we do not introduce $\mathcal{S}_1^\ddagger$ as unlike $\w_p$ for $p \geq 1$, the matrix $\w_0 = \hh\mm^*$ is fixed, and there is no bias between $(L/Q)\PT\cA_1^*\cA_1(\hh\mm^*)$, and $\hh\mm^*$.

A bound on the operator norm for the term with $p=1$ above can be achieved by using a simple triangle inequality followed by an application of Lemma \ref{lem:injectivitype1part1}, and \ref{lem:injectivity1part2}, in Section \ref{sec:key-lemmas} below, to obtain
\begin{align}\label{eq:injectivitypart1-complete}
\left\| \frac{L}{Q} \PT\cA_1^*\cA_1\PT - \PT\right\| &\leq \frac{L}{Q} \left\| \PT\cA_1^*\cA_1\PT - \E \PT\cA_1^*\cA_1\PT \right\|+ \left\| \frac{L}{Q} \E \PT\cA_1^*\cA_1\PT  - \PT\right\|\notag\\
&\leq \frac{1}{4}\sqrt{\frac{Q}{L}}
\end{align}
with probability at least $1-2(LN)^{-\beta+1}$ (the constant 2 upfront comes from the union bound),  for a  parameter $\beta > 2$ that controls the choice of $Q$, and $N$. Note that the expectation $\E$ is only w.r.t. the random construction of $\vm{c}_{\o,n}$ in \eqref{eq:cl-construction}, and not w.r.t. to the randomness due to the sets $\Delta_{n,p}$. The operator norm of the remaining terms ($p \geq 2$) in the expression \eqref{eq:Wp-fro-bound-intermediate} can all be bounded using Lemma \ref{lem:injectivity} in Section \ref{sec:key-lemmas} below to conclude that 
\begin{equation}\label{eq:Wpfro-bound}
\|\w_p\|_\F \leq 2^{-p-1} \sqrt{\frac{Q}{L}}, ~\mbox{for every} ~ p \in \{1,\ldots, P\},
\end{equation}
holds with probability at least $1-\mathcal{O}(L^{-\beta+1})$. This means using the crude union bound that $\|\w_P\|_\F\leq 2^{-P-1}\sqrt{Q/L}$ holds with probability at least $1-\mathcal{O}(PL^{-\beta+1}) \geq 1-\mathcal{O}(L^{-\beta+2})$. Now choosing $P= (L/Q)= 0.5\log_2(4\beta K \log (LN))$ is more than sufficient to imply that $\|\w_P\|_\F\leq (4\gamma)^{-1}$, where the value of $\gamma$ is dictated by Lemma \ref{lem:cA-operator-norm}. This proves the first half of \eqref{eq:uniqueness-cond}.

To prove the second half of \eqref{eq:uniqueness-cond}, use the construction in \eqref{eq:iteration} to write
\[
\y_P = -\left(\frac{L}{Q}\cA_1^*\cA_1(\hh\mm^*) +\sum_{p=2}^P \cA_p^*\cA_p \mathcal{S}_p^\ddagger(\w_{p-1})\right).
\]
Since $\w_p \in T_1\cap T_2,$ for every $p \in \{1, 2,3,\ldots, P\}$, this means 
\[
\PTc(\y_P)  = - \PTc\left(\frac{L}{Q}\cA_1^*\cA_1(\hh\mm^*)-\hh\mm^*\right)-\sum_{p=2}^P \PTc\left( \cA_p^*\cA_p\mathcal{S}_p^\ddagger(\w_{p-1})-\w_{p-1}\right).
\]
Taking the operator norm and a triangle inequality, followed by an application of the fact that $\|\PTc\| \leq 1$, shows that 
\begin{align*}
\|\PTc(\y_P)\| &\leq \left[\left\| \frac{L}{Q}\cA_1^*\cA_1(\hh\mm^*)-\hh\mm^* \right\| + \sum_{p=2}^P \left\|\cA_p^*\cA_p\mathcal{S}_p^\ddagger (\w_{p-1})-\w_{p-1}\right\|\right].
\end{align*}
Note that by adding and subtracting $(L/Q)\E\cA_1^*\cA_1(\hh\mm^*)$ in the first term and similarly adding and subtracting $\E\cA_p^*\cA_p\mathcal{S}_p^\ddagger(\w_{p-1})$ from every term in the summation above, and subsequently using the triangle inequality, we obtain 
\begin{align*}
\|\PTc(\y_P)\| & \leq \frac{L}{Q}\left\|\cA_1^*\cA_1(\hh\mm^*)-\E\cA_1^*\cA_1(\hh\mm^*)\right\|+\left\| \frac{L}{Q}\E\cA_1^*\cA_1(\hh\mm^*) - \hh\mm^*\right\| + \\
& \sum_{p=2}^P \left[ \left\|\cA_p^*\cA_p\mathcal{S}_p^\ddagger (\w_{p-1})-\E\cA_p^*\cA_p\mathcal{S}_p^\ddagger (\w_{p-1})\right\|+ \left\| \E\cA_p^*\cA_p\mathcal{S}_p^\ddagger (\w_{p-1})-\w_{p-1}\right\|\right].
\end{align*}
Note that the expectation $\E$ is only w.r.t. the random construction of $\vm{c}_{\o,n}$ in \eqref{eq:cl-construction}, and not w.r.t. to the randomness due to the sets $\Delta_{n,p}$. Now each of the term on the right hand side above can be controlled using Corollary \ref{cor:concentration}; and Lemma \ref{lem:concentrationpe1part2}, \ref{lem:concentration}, and \ref{lem:conectrationpg2part2} in Section \ref{sec:key-lemmas} below, respectively that hold under the choices of $Q$, and $N$ that conform to Theorem \ref{thm:main} to give us the upper bound 
\[
\|\PTc(\y_P)\| \leq \frac{1}{8} + \frac{1}{8} + \frac{1}{2}\sum_{p=2}^P 2^{-p} \leq \frac{1}{2},
\]
which, using the union bound, holds with probability at least $1-\mathcal{O}((LN)^{-\beta+1})$. 

\begin{remark}\label{rem:dense-w}
	Lemma \ref{lem:conectrationpg2part2} to control $\big\| \E\cA_p^*\cA_p\mathcal{S}_p^\ddagger (\w_{p-1})-\w_{p-1}\big\|$ relies on the uniform result in \eqref{eq:RIP} for sparse vectors to overcome the statistical dependence between $\cA_p$, and $\w_{p-1}$. Being able to control this term is one of the main reasons of working with $T_1 \cap T_2$ instead of $T_1$. Since $\w_{p-1} \in T_1\cap T_2$, the columns of $\w_{p-1}$ are always $S$-sparse enabling us to employ the uniform result. This technical requirement restricts our results to only $S$-sparse $\ww$. However, we think that the proof technique may be improved to work for a completely dense vector $\ww$ as is suggested by the numerical experiments in Section \ref{sec:Exps}. 
\end{remark} 

We also need to show that  \eqref{eq:uniqueness-cond2} holds. To that end, note that by Corollary \ref{cor:injectivitype1part1} in Section \ref{sec:key-lemmas} below, the linear map $\cA$ is well-conditioned on $T_1\cap T_2$, and hence, for $\z \in \mbox{Null}(\cA)$, we have 
\begin{align*}
\|\cA\PT(\z)\|_\F^2 &= |\<\z, \PT\cA^*\cA\PT(\z)\>| \\
&\geq |\<\z,(\PT\cA^*\cA\PT-\PT)\z\>+\<\z,\PT(\z)\>|\\
& \geq \|\PT(\z)\|_\F^2 - \|\PT\cA^*\cA\PT-\PT\|\|\PT(\z)\|_\F^2.
\end{align*}
Corollary \ref{cor:injectivitype1part1} shows that $\|\PT\cA^*\cA\PT-\PT\| \leq 1/8$. Using this fact in the inequality above proves \eqref{eq:uniqueness-cond2}. 

Finally, the choice of upper bounds on $L$ and $N$ in the statement of the theorem is the tightest upper bound that conforms to all the lemmas and corollaries, and uses the fact that $L/Q = P = 0.5\log_2(4\beta K \log (LN))$ derived above. The nuclear norm minimization recovers the true solution when all of the above conclusion hold true. The failure probability of each of the lemmas and corollaries is less than or equal to $(LN)^{-\beta+4}$, and hence, using the union bound, the probability that none of the above items fail is $1-\mathcal{O}((LN)^{-\beta+4})$.

This completes the proof of Theorem \ref{thm:main}.

\subsection{Concentration Inequalities}\label{sec:Conc-Ineq}
Most of the lemmas below require an application of either the uniform version, or the Orlicz-norm version of the matrix Bernstein inequality to control the operator norm of the sum of independent random matrices. Before stating them, we give an overview of the Orlicz-norm results that are used later in the exposition.

We begin by giving basic facts about subgaussian and subexponential random variables that are used throughout the proofs. The proofs of these facts can be found in any standard source; see, for example, \cite{vershynin10in}. 

The Orlicz norms of a scalar random $X$ are defined as 
\[
\|X\|_{\psi_\alpha} := \inf \left\{ u > 0 : \E \exp \left(\frac{|X|^\alpha}{u^\alpha}\right) \leq 2 \right\},~ \alpha \geq 1.
\]
The Orlicz-norm of a vector $\|\vm{z}\|_{\psi_\alpha}$, and a matrix $\|\vm{Z}\|_{\psi_\alpha}$ are then defined by setting $X = \|\vm{z}\|_2$, and $X = \|\vm{Z}\|$, respectively, in the above definition. Therefore, we restrict our discussion below to scalar random variables, and it can trivially extended to vectors and matrices using above mentioned equivalence. 

Some of the key facts relating the Orlicz norms of subgaussian and subexponential random variables are as follows. A subgaussian random variable $X$ can be characterized by the fact that its Orlicz-2 norm is always finite, i.e., $\|X\|_{\psi_2} < \infty$.  Similarly, for a subexponential r.v., we have $\|X\|_{\psi_1} < \infty$. A random variable $X$ is subgaussian iff $X^2$ is subexponential. Furthermore, 
\begin{equation}\label{eq:Square-Subgaussian}
\|X\|_{\psi_2}^2 \leq \|X^2\|_{\psi_1} \leq 2\|X\|_{\psi_2}^2.
\end{equation}
At some points in the proof, we are interested in bounding $\|X-\E X\|_{\psi_\alpha}$. A coarse bound is obtained by using the triangle inequality, 
\[
\|X-\E X\|_{\psi_\alpha} \leq \|X\|_{\psi_\alpha} + \|\E X\|_{\psi_\alpha} 
\]
followed by Jensen's inequality,  $\|\E X\|_{\psi_\alpha} \leq \E \|X\|_{\psi_\alpha} = \|X\|_{\psi_\alpha}$, which further implies that 
\begin{equation}\label{eq:Orlicz-norm-centered-X}
\|X-\E X\|_{\psi_\alpha}  \leq 2 \|X\|_{\psi_\alpha}.
\end{equation}

We also find it handy to have a  generalized version of the above fact; namely, the product of two  subgaussian random variables $X_1$, and $X_2$ is subexponential, and 
\begin{equation}\label{eq:Product-Subgaussians}
\|X_1X_2\|_{\psi_1} \leq C \|X_1\|_{\psi_2}\|X_2\|_{\psi_2},
\end{equation}
for some $C > 0$.

As we are working with Gaussian random variables mostly in the proofs, some of the useful identities for a Gaussian vector $\vm{g} \sim \text{Normal}(\vm{0}, \vm{I}_M)$ are: For a fixed vector $\vm{z}$, the random variable $\<\vm{g}, \vm{z}\>$ is also Gaussian, and hence, $|\<\vm{g},\vm{z}\>|^2$ must have a subexponential tail behavior, and it can be easily verified that
\[
\P{|\<\vm{g}, \vm{z}\>|^2 > \lambda} \leq e^{-\lambda / \|\vm{z}\|_2^2}
\]
for every scalar $\lambda \geq 0$. Moreover, the $\ell_2$-norm of $\vm{g}$ is strongly concentrated about its mean, and there exists $C > 0$ such that 
\[
\P{\|\vm{g}\|_2^2 > \lambda M} \leq Ce^{-\lambda}.
\]
This tail behavior of a random variable completely defines its Orlicz-norm. Specifically, for a subexponential random variable $X$, 
\begin{equation}\label{eq:tail-Orlicz-relation}
\P{X > u}\leq \alpha e^{-\beta u} \implies \|X\|_{\psi_1} \leq (1+\alpha)/\beta. 
\end{equation}
This completes the required overview.

We now state the matrix Bernstein inequalities, which is heavily used in the proofs below. 
\begin{prop}[Uniform Version \cite{tropp12us}, \cite{koltchinskii10nu}]\label{prop:matrixBernstein-uniform}
		Let $\vm{Z}_1,\vm{Z}_2,\ldots,\vm{Z}_Q$ be iid random matrices with dimensions $M \times N$ that satisfy $\E (\vm{Z}_q) = 0$. Suppose that $\|\z_q\|< U$ almost surely for some constant $U$, and all $q = 1,2,3, \ldots, Q$. Define the variance as 
		\begin{align}\label{eq:matbernsigma}
		&\sigma^2_{\vm{Z}}  = \max \left\{\left\|\sum_{q = 1}^Q (\E \vm{Z}_q\vm{Z}_q^*)\right\|,\left\|\sum_{q = 1}^Q (\E \vm{Z}_q^*\vm{Z}_q)\right\| \right\}.
		\end{align}
		Then, there exists $C > 0$ such that , for all $t>0$, with probability at least $1-\er^{-t}$,
		\begin{align}
		&\left\|\vm{Z}_1+\vm{Z}_2+\cdots+\vm{Z}_Q\right\| \lesssim \max\left\{\sigma_{\vm{Z}}\sqrt{t+\log(M+N)},U(t+\log(M+N))\right\}.
		\end{align}
	\end{prop}
 The version of Bernstein listed below depends on the Orlicz norms $\|\z\|_{\psi_\alpha},~ \alpha \geq 1$ of a matrix $\z$, defined as 
\begin{equation}\label{eq:matpsinorm}
\|\z\|_{\psi_\alpha} = \inf \{ u>0: \E \exp \bigg(\frac{\|\z\|^\alpha}{u^\alpha}\bigg) \leq 2\}, \quad \alpha \geq 1.
\end{equation}
\begin{prop}[Orlicz-norm Version \cite{koltchinskii10nu}]\label{prop:matbernpsi}
	Let $\vm{Z}_1,\vm{Z}_2,\ldots,\vm{Z}_Q$ be iid random matrices with dimensions $M \times N$ that satisfy $\E (\vm{Z}_q) = 0$. Suppose that $\|\z_q\|_{\psi_\alpha} \leq U_{\alpha}$ for some constant $U_{\alpha} > 0$, and  $q  = 1,2,3, \ldots, Q$. Define the variance $\sigma_{\vm{Z}}^2$ as in \eqref{eq:matbernsigma}. Then, there exists $C > 0$ such that, for all $t>0$, with probability at least $1-\er^{-t}$,
	\begin{align}\label{eq:matbernpsi}
	&\left\|\vm{Z}_1+\vm{Z}_2+\cdots+\vm{Z}_Q\right\| \lesssim \max\left\{\sigma_{\vm{Z}}\sqrt{t+\log(M+N)},U_{\alpha}\log^{\tfrac{1}{\alpha}}\left(\frac{QU_{\alpha}^2}{\sigma_{\vm{Z}}^2}\right)(t+\log(M+N))\right\}.
	\end{align}
\end{prop}
\subsection{Key lemmas}\label{sec:key-lemmas}
This section provides the important lemmas that constitute the main ingredients to establish the uniqueness conditions \eqref{eq:uniqueness-cond2}, and \eqref{eq:uniqueness-cond} for our construction of the dual certificate $\y$ in the previous section. 

\textbf{Conditioning on $T_1 \cap T_2$}

The results in this section concern the conditioning of the linear maps $\cA$, and $\cA_p$ when restricted to the space $T_1 \cap T_2$. 
\begin{lem}\label{lem:injectivity}
	Let the coherences $\mu_{\max}^2$, $\mu_0^2$, and $\rho_0^2$ be as defined in \eqref{eq:coherence_B}, \eqref{eq:coherence_h}, and \eqref{eq:coherence_m}, respectively. Fix $\beta \geq 2$. Choose the subsets $\Gamma_p : = \{(\Delta_{n,p},n)\}_n$ constructed as in Section \ref{sec:Golfing}, so that, for $p \in \{2,3,\ldots,P\}$,
	\[
	|\Delta_{n,p}| = Q \geq C\beta (\mu_0^2K+\mu_{\max}^2 S)\log^2(LN),
	\]
	for some sufficiently large $C > 0$. Then the linear operators $\cA_p$, and $\mathcal{S}_p^\ddagger$ defined in \eqref{eq:cAp} and \eqref{eq:cSp}, obey 
	\begin{align}\label{eq:injectivity-result1}
	\max_{2\leq p \leq P} \left\|\PT\cA_p^*\cA_p\mathcal{S}_p^\ddagger\PT - \PT \right\| \leq \frac{1}{2}
	\end{align}
	with probability at least $1-(P-1)(LN)^{-\beta} \geq 1- (LN)^{-\beta+1}$. 
\end{lem}
\begin{lem}\label{lem:injectivitype1part1}
Let the coherences $\mu_{\max}^2$, $\mu_0^2$, $\rho_0^2$ be as in Lemma \ref{lem:injectivity}. Fix $\beta \geq 1$. Choose $\Delta_{n,1}$ such that
	\begin{equation}\label{eq:Q-bound-injectivity3}
	|\Delta_{n,1}| = Q \geq C\beta (\mu_0^2K+\mu_{\max}^2S)(L/Q)^{1/2}\log^2(LN),
	\end{equation}
	for some sufficiently large $C > 0$. Then the linear operator $\cA_1$ defined in \eqref{eq:cAp} obeys
	\begin{equation}
	\frac{L}{Q}\left\|\PT\cA_1^*\cA_1\PT - \E\PT\cA_1^*\cA_1\PT \right\| \leq \frac{1}{8}{\sqrt{\frac{Q}{L}}}
	\end{equation}
	with probability at least $1-(LN)^{-\beta}$. 
\end{lem}
\begin{cor}[Corollary of Lemma \ref{lem:injectivitype1part1}]\label{cor:injectivitype1part1}
	Let the coherences $\mu_{\max}^2$, $\mu_0^2$, and $\rho_0^2$ be as in Lemma \ref{lem:injectivity}. Fix $\beta \geq 1$. Assume that
	\[
	L \geq C\beta (\mu_0^2 K + \mu_{\max}^2 S) \log^2 (LN), 
	\]
	for some sufficiently large $C > 0$. Then the linear operator $\cA$ defined in \eqref{eq:cA} obeys
	\[
	\|\PT\cA^*\cA\PT-\PT \| \leq  \frac{1}{8}
	\]
	with probability at least $1-(LN)^{-\beta}$. 
\end{cor}
\begin{lem}\label{lem:injectivity1part2}
Let $\mu_{\max}^2$ be as in Lemma \ref{lem:injectivity}. Fix $\beta \geq 2$. Then there exists a constant $C$ such that 
	\[
	|\Delta_{n,1}| = Q \geq C\beta \mu_{\max}^2 S(L/Q) \log (LN) 
	\]
	implies that the linear operator $\cA_1$ defined in \eqref{eq:cAp} obeys
	\[
	\left\|\frac{L}{Q}\E\PT\cA_1^*\cA_1\PT-\PT\right\| \leq \frac{1}{8}\sqrt{\frac{Q}{L}}
	\]
	with probability at least $1-(LN)^{-\beta+1}$. 
\end{lem}
\textbf{Coherences of iterates}\label{sec:coherence-lemmas}

In this section, we define the coherences of the iterates $\w_p$ in \eqref{eq:Wp-def2}, and show that these coherences can be bounded in terms of $\mu_0$, and $\rho_0$ in \eqref{eq:coherence_h}, and \eqref{eq:coherence_m}, respectively.  The coherences are defined, among other variables, in terms of  $\vm{S}_{n,p}^\ddagger$ in \eqref{eq:Snp}, and $\vm{D}_n := \vm{I}_K \otimes \vm{e}_n\vm{e}_n^*$. The partition $\{(\Delta_{n,p},n)\}_{n,p}$ is as defined in Section \ref{sec:Golfing}, and we assume the implications of restricted isometry property in \eqref{eq:RIPresults} as given. Moreover, we also take the results of Lemma \ref{lem:injectivity}, \ref{lem:injectivitype1part1}, and \ref{lem:injectivity1part2} as true. The following results are in order then. 

\begin{lem}\label{lem:rho-p-bound}
Define\footnote{We use the index variables $\ell^\prime$, and $n^\prime$ as $\ell$ and $n$ are already reserved to index the set $\Gamma_p$ in the proofs of these lemmas.}
	\begin{align}\label{eq:rho_Wp}
	\rho^2_p &:= \frac{Q}{L}N\max_{1\leq \nt\leq N}\left[ \sum_{\ot \in \Delta_{\nt,p+1}}\left\|\bh_{\ot}^* \vm{S}^\ddagger_{\nt,p+1}\w_p\vm{D}_{\nt}\right\|_2^2\right] ~ \mbox{for every}~ p \in \{1,2,3, \ldots, P\}.
	\end{align}
	Then
	\begin{equation}\label{eq:rhop-bound}
	\rho_p \leq 2^{-p}\sqrt{\frac{Q}{L}}\rho_0 ~\mbox{for every}~ p \in \{1,2,3, \ldots, P\}.
	\end{equation}
\end{lem}
\begin{lem}\label{lem:nup-bound}
	 Define
	\begin{align}\label{eq:nu_Wp}
	\nu^2_p &:= \frac{Q^2}{L}N\max_{1\leq\nt\leq N}\left[\max_{\ot\in \Delta_{\nt,p+1}}\left\|\bh_{\ot}^* \vm{S}^\ddagger_{\nt,p+1}\w_p\vm{D}_{\nt}\right\|_2^2\right], ~\mbox{for every}~ p \in \{ 1,2,3, \ldots, P\}. 
	\end{align}
	 Fix $\beta \geq 1$. Then there exists a constant $C$ such that 
	\begin{equation*}
	Q \geq C\beta(\mu_0^2 K + \mu_{\max}^2 S)\log^2(LN)
	\end{equation*}
	implies 
	\begin{equation}\label{eq:nup-bound}
	\nu_p \leq 2^{-p+3}\mu_0\rho_0 ~\mbox{for every}~ p \in \{1, 2,3,\ldots, P\}
	\end{equation}
	with probability at least $1-(LN)^{-\beta}$. 
\end{lem}
\begin{lem}\label{lem:mup-bound}
	Define 
	\begin{align}\label{eq:mup-def}
	\mu_p^2 := \frac{Q^2}{L}  \sum_{\nt=1}^N\left[\max_{\ot \in \Delta_{\nt,p}}\left\|\bh_{\ot}^* \vm{S}_{\nt,p+1}^\ddagger \w_p\vm{D}_{\nt}\right\|_2^2 \right] ~\mbox{for every}~ p \in \{1,2,3, \ldots, P\}.
	\end{align}
    Let $Q$, and $N$ be the same as in Lemma \ref{lem:nup-bound} for sufficiently large $C$. Then
	\begin{equation}\label{eq:mup-bound}
	\mu_p \leq 2^{-p+2} \mu_0~ \mbox{for every}~ p \in \{1,2,3,\ldots, P\}
	\end{equation}
	with probability at least $1- (LN)^{-\beta}$.
\end{lem}
\textbf{Range of $\cA^*$}

Finally, the results in this section help us establish that the dual certificate $\y$ mostly lies in $T_1\cap T_2$, that is, $\|\PTc(\y)\| \leq 1/2$; one of the uniqueness conditions in \eqref{eq:uniqueness-cond}. Let $\cA_p$, $\mathcal{S}_p^\ddagger$ and $\w_p$ be as in \eqref{eq:cAp}, \eqref{eq:cSp}, and \eqref{eq:Wp-def2}, respectively. In addition, let the coherences 
$\mu_p$, $\rho_p$, and $\nu_p$ be as defined in \eqref{eq:mup-def}, \eqref{eq:rho_Wp}, and \eqref{eq:nu_Wp}, respectively. We shall take \eqref{eq:mup-bound}, \eqref{eq:rhop-bound}, and \eqref{eq:nup-bound} as given. The following results are in order then. 
\begin{lem}\label{lem:concentration}
	Fix $\beta \geq 4$. Then there exists a constant $C$ such that
	\begin{equation*}
	Q \geq C \beta \mu_0^2 K \max\{(L/Q),\log^2(LN)\}\log (LN), ~ \mbox{and} ~ N \geq C\beta \rho_0^2(L/Q)\log(LN)
	\end{equation*}
	implies that 
	\[
	 \left\| \cA_p^*\cA_p\mathcal{S}_p^\ddagger (\w_{p-1}) - \E\cA_p^*\cA_p\mathcal{S}_p^\ddagger (\w_{p-1})\right\| \leq 2^{-p-1} ~\mbox{for all} ~ p \in \{2,3,\ldots,P\}
	\]
	holds with probability at least  $1- (LN)^{-\beta+4}$
\end{lem}
\begin{cor}[Corollary of Lemma \ref{lem:concentration}]\label{cor:concentration}
	Let $\cA_1$ be as in \eqref{eq:cAp}, and coherences $\mu_0^2$, and $\rho_0^2$ be as in \eqref{eq:coherence_h}, and \eqref{eq:coherence_m}, respectively. Fix $\beta \geq 1$. Then there exist a constant $C$ such that 
	\begin{equation*}
	Q \geq C \beta \mu_0^2 K \max\{(L/Q),\log^2(LN)\}\log (LN), ~ \mbox{and} ~ N \geq C\beta \rho_0^2 (L/Q)\log(LN)
	\end{equation*}
	implies that
	\[
	\frac{L}{Q}\left\|\cA_1^*\cA_1(\hh\mm^*) - \E\cA_1^*\cA_1(\hh\mm^*)\right\| \leq \frac{1}{8}
	\]
	with probability exceeding $1-(LN)^{-\beta}$. 
\end{cor}
\begin{lem}\label{lem:conectrationpg2part2}
	Assume further that the restricted isometry property in \eqref{eq:RIPresults} holds. Then 
	\[
	\|\E\cA_p^*\cA_p\mathcal{S}_p^\ddagger (\w_{p-1})-\w_{p-1}\| \leq 2^{-p-1} ~\mbox{for all}~ p \in \{2,3, \ldots, P\}
	\]
	with probability at least $1-(LN)^{-\beta+1}$
\end{lem}
\begin{lem}\label{lem:concentrationpe1part2}
	Let $\cA_1$, $\mu_{\max}^2$, and $\rho_0^2$ be as in \eqref{eq:cAp}, \eqref{eq:coherence_B}, and \eqref{eq:coherence_m}, respectively. Fix $\beta \geq 1$. Then there exists a constant $C$ such that 
	\[
	L \geq C\beta \mu_{\max}^2 S (L/Q)\log(LN), ~\mbox{and}~ N \geq C\beta \rho_0^2(L/Q)\log(LN)
	\]
	implies that 
	\[
	\left\|\frac{L}{Q}\E\cA_1^*\cA_1(\hh\mm^*) - \hh\mm^*\right\| \leq \frac{1}{8}
	\]
	with probability at least $1- (LN)^{-\beta}$. 
\end{lem}
\section{Proofs of the Key Lemmas}
This section provides the proofs of all the key lemmas laid out in Section \ref{sec:Main_results}. All of the main lemmas involve bounding the operator norm of a sum of independent random matrices with high probability. The matrix Bernstein inequality is used repeatedly to compute such probability tail bounds.

In the proof of the lemma below, the following calculations come in handy. Using the definition of the projection $\PT$ in \eqref{eq:PT-def}, we can see that 
\[
\PT(\bh_{\o}\vm{\phi}_{\o,n}^*) = (\hh\hh^*\bh_\o)\vm{\phi}_{\o,n}^*+ \PO\bh_\o(\mm\mm^*\vm{\phi}_{\o,n})^*-(\hh\hh^* \bh_\o)(\mm\mm^*\vm{\phi}_{\o,n})^*.
\]
It also follows from the definition \eqref{eq:PT-def} that $\PO \PT= \PT\PO =  \PT$. Another quantity of interest is $\|\PT(\bh_{\o}\vm{\phi}_{\o,n}^*)\|_{\F}^2 = \<\PT(\bh_{\o}\vm{\phi}_{\o,n}^*),\bh_{\o}\vm{\phi}_{\o,n}^*\>$, which can be expanded as
\begin{align}\label{eq:PTA-fro}
\left\|\PT(\bh_{\o}\vm{\phi}_{\o,n}^*)\right\|_{\F}^2&=\<\hh\hh^*\bh_\o\vm{\phi}_{\o,n}^*,\bh_\o\vm{\phi}_{\o,n}^*\> +\<\PO\bh_\o\vm{\phi}_{\o,n}^*\mm\mm^*,\bh_\o\vm{\phi}_{\o,n}^*\>-\<\hh\hh^* \bh_\o\vm{\phi}_{\o,n}^*\mm\mm^*,\bh_\o\vm{\phi}_{\o,n}^*\>\notag\\
& = \|\vm{\phi}_{\o,n}\|_2^2 |\bh_\o^*\hh|^2 + \|\PO \bh_\ell\|_2^2 |\mm^*\vm{\phi}_{\o,n}|^2 - |\bh_\o^*\hh|^2 |\mm^*\vm{\phi}_{\o,n}|_2^2\notag\\
& \leq \|\vm{\phi}_{\o,n}\|_2^2 |\bh_\o^*\hh|^2 + \|\PO \bh_\ell\|_2^2 |\mm^*\vm{\phi}_{\o,n}|^2.
\end{align}
Moreover, 
\begin{align}
\left\|\PT(\vm{S}_{n,p}^\ddagger\bh_{\o}\vm{\phi}_{\o,n}^*)\right\|_{\F}^2&=\<\hh\hh^*\vm{S}_{n,p}^\ddagger\bh_{\o}\vm{\phi}_{\o,n}^*,\vm{S}_{n,p}^\ddagger\bh_{\o}\vm{\phi}_{\o,n}^*\> +\<\PO\vm{S}_{n,p}^\ddagger\bh_{\o}\vm{\phi}_{\o,n}^*\mm\mm^*,\vm{S}_{n,p}^\ddagger\bh_{\o}\vm{\phi}_{\o,n}^*\>\notag\\
&~~~~~~~~~~~~-\<\hh\hh^*\vm{S}_{n,p}^\ddagger \bh_{\o}\vm{\phi}_{\o,n}^*\mm\mm^*,\vm{S}_{n,p}^\ddagger\bh_{\o}\vm{\phi}_{\o,n}^*\>\notag\\
& = \|\vm{\phi}_{\o,n}\|_2^2 |\bh_\o^*\vm{S}_{n,p}^\ddagger\hh|^2 + \| \vm{S}_{n,p}^\ddagger\bh_\ell\|_2^2 |\mm^*\vm{\phi}_{\o,n}|^2 - |\bh_\o^*\vm{S}_{n,p}^\ddagger\hh|^2 |\mm^*\vm{\phi}_{\o,n}|_2^2\notag\\
& \leq \|\vm{\phi}_{\o,n}\|_2^2 |\bh_\o^*\vm{S}_{n,p}^\ddagger\hh|^2 + \| \vm{S}_{n,p}^\ddagger\bh_\ell\|_2^2 |\mm^*\vm{\phi}_{\o,n}|^2.
\end{align}
We are now ready to move on to the proof of Lemma \ref{lem:injectivity} given below. 
\subsection{Proof of Lemma \ref{lem:injectivity}}
\begin{proof}
	 The lemma concerns bounding the quantity $\|\PT\cA_p^*\cA_p\mathcal{S}_p^\ddagger\PT - \PT\|$ for $p \geq 2$. Using the definition of $\cA_p$, and $\mathcal{S}^\ddagger_p$ in \eqref{eq:cAp}, and \eqref{eq:cSp}, respectively, we expand the quantity $\PT\cA_p^*\cA_p\mathcal{S}_p^\ddagger\PT$ and evaluate its expectation with the sets $\Gamma_p$ fixed as follows 
\begin{align*}
\E\PT\cA_p^*\cA_p\mathcal{S}_p^\ddagger\PT  &= \sum_{(\o,n) \in \Gamma_p}\PT\left[\bh_\o\bh_\o^*\vm{S}_{n,p}^\ddagger \otimes \E\vm{\phi}_{\o,n}\vm{\phi}_{\o,n}^*\right]\PT\\
&= \sum_{(\o,n) \in \Gamma_p}\PT\left[\bh_\o\bh_\o^* \vm{S}_{n,p}^\ddagger \otimes \vm{D}_n\right]\PT.
\end{align*}
From the construction of the sets $\Gamma_p$ in Section \ref{sec:Golfing}, it is clear that $(\o,n) \in \Gamma_p$ means that for every $n$, the index $\o$ traverses the set $\Delta_{n,p}$. This means we can split the summation over $\Gamma_p$ into an outer sum over $n$ and an inner sum over $\o \in \Delta_{n,p}$. Moreover, by definition $\PT = \PT\PO$, and $\PO \vm{S}_{n,p}^\ddagger = \vm{S}_{n,p}^\ddagger$.  The following equality is now in order.
\[
\E\PT\cA_p^*\cA_p\mathcal{S}_p^\ddagger\PT  = \sum_n \sum_{\o \in \Delta_{n,p}}\PT\left[\PO\bh_\o\bh_\o^*\PO \vm{S}_{n,p}^\ddagger \otimes \vm{I}_{KN}\right]\PT.
\]
It is now easy to see by using the definition of $\vm{S}_{n,p}^\ddagger$ in Section \ref{sec:Main_results} that for $p \geq 2$, 
\[
\sum_{\o \in \Delta_{n,p}}\PO\bh_\o\bh_\o^*\PO\vm{S}_{n,p}^\ddagger = \PO,
\]
 which implies that $\E\PT\cA_p^*\cA_p\mathcal{S}_p^\ddagger\PT  = \PT$. Given this, we only need to control the term $ \|\PT\cA_p^*\cA_p\mathcal{S}_p^\ddagger\PT -  \PT(\E\cA_p^*\cA_p)\mathcal{S}_p^\ddagger\PT\|$. By definitions of $\cA_p$, $\mathcal{S}_{p}^\ddagger$, and $\PT$ in Section \ref{sec:Golfing}, we can write 
\[
\PT\cA_p^*\cA_p\mathcal{S}_p^\ddagger\PT  = \sum_{(\o,n)\in \Gamma_p}\PT\left[\bh_\o\bh_\o^*\vm{S}_{n,p}^\ddagger \otimes \vm{\phi}_{\o,n}\vm{\phi}_{\o,n}^*\right]\PT.
\]
Note that the action of the linear map $\mathcal{Z}_{\o,n}:=\PT(\bh_\ell\bh_\ell^*\vm{S}_{n,p}^\ddagger\otimes\vm{\phi}_{\o,n}\vm{\phi}_{\o,n}^*)\PT$ on an $L \times KN$ matrix $\X$ is  $\mathcal{Z}_{\o,n}(\X) = \PT\left[\bh_\ell\bh_\ell^*\vm{S}_{n,p}^\ddagger\otimes\vm{\phi}_{\o,n}\vm{\phi}_{\o,n}^*\right]\PT(\X) = \<\X,\PT(\vm{S}_{n,p}^\ddagger\bh_\ell\vm{\phi}_{\ell,n}^*)\>\PT(\bh_\ell\vm{\phi}_{\ell,n}^*)$. It is clear from this definition that $\mathcal{Z}_{\o,n}$ are rank-1 operators. Thus, we are asking the question of bounding the operator norm of sum of independent operators $\mathcal{Z}_{\o,n}$. Subtracting the expectation, we get
\begin{align*}
\PT(\cA_p^*\cA_p-\E\cA_p^*\cA_p)\mathcal{S}_p^\ddagger\PT &= \sum_{(\o,n)\in \Gamma_p}\PT \left[\bh_\ell\bh_\ell^*\vm{S}_{n,p}^\ddagger\otimes\left(\vm{\phi}_{\o,n}\vm{\phi}_{\o,n}^*-\E\vm{\phi}_{\o,n}\vm{\phi}_{\o,n}^*\right)\right]\PT\\
&=\sum_{(\o,n)\in \Gamma_p} \mathcal{Z}_{\ell,n} -\E \mathcal{Z}_{\o,n}.
\end{align*}
The operator norm of the sum can be controlled using Bernstein's inequality. The variance $\sigma_{\mathcal{Z}}^2$; the main ingredient to compute the Bernstein bound, is in this case 
	\begin{align}\label{eq:variance-lemma1}
	\sigma_{\mathcal{Z}}^2  &:= \max \left\{\left\|  \sum_{(\o,n)\in \Gamma_p}\E\mathcal{Z}_{\o,n}\mathcal{Z}_{\o,n}^* -(\E\mathcal{Z}_{\o,n}) (\E\mathcal{Z}_{\o,n})^*\right\|, \left\| \sum_{(\o,n)\in \Gamma_p} \E\mathcal{Z}_{\o,n}^*\mathcal{Z}_{\o,n} -(\E\mathcal{Z}_{\o,n})^* (\E\mathcal{Z}_{\o,n})\right\|\right\} \notag\\
	&\leq \max \left\{\left\|  \sum_{(\o,n)\in \Gamma_p} \E\mathcal{Z}_{\o,n}\mathcal{Z}_{\o,n}^*\right\|,\left\| \sum_{(\o,n)\in \Gamma_p} \E\mathcal{Z}_{\o,n}^*\mathcal{Z}_{\o,n}\right\|\right\},
	\end{align}
where the last inequality follows from the fact that for two positive semidefinite (PSD) matrices $\vm{A}$, and $\vm{B}$, $\|\vm{A}-\vm{B}\| \leq \|\vm{A}\|$, whenever $\vm{A}-\vm{B}$ is a PSD matrix. The first term in the maximum in the variance expression is simplified below. As mentioned earlier, the linear operator $\mathcal{Z}_{\o,n}$ can be visualized as a rank-1 matrix  $\PT(\bh_\o\vm{\phi}_{\o,n}^*) \boxtimes \PT(\vm{S}_{n,p}^\ddagger\bh_\o\vm{\phi}_{\o,n}^*)$, and $\mathcal{Z}_{\o,n}(\X)$ is just the product of the rank-1 matrix above with the vectorized $\vm{X}$, i.e., 
\[
\mathcal{Z}_{\o,n}(\X) = \left[\PT(\bh_\o\vm{\phi}_{\o,n}^*) \boxtimes \PT(\vm{S}_{n,p}^\ddagger\bh_\o\vm{\phi}_{\o,n}^*)\right]\mbox{vec}(\X). 
\]
It is then easy to see that 
\[
\mathcal{Z}_{\o,n}\mathcal{Z}_{\o,n}^* = \|\PT(\vm{S}_{n,p}^\ddagger \bh_\o\vm{\phi}_{\o,n}^*)\|_\F^2\cdot\left[\PT(\bh_\o\vm{\phi}_{\o,n}^*) \boxtimes \PT(\bh_\o\vm{\phi}_{\o,n}^*)\right], 
\]
where 
\[
\left[\PT(\bh_\o\vm{\phi}_{\o,n}^*) \boxtimes \PT(\bh_\o\vm{\phi}_{\o,n}^*)\right] = \PT\left[\bh_\o\bh_\o^* \otimes \vm{\phi}_{\o,n}\vm{\phi}_{\o,n}^*\right]\PT.
\]
Thus,
\begin{align*}
\left\| \sum_{(\o,n)\in \Gamma_p}\E\mathcal{Z}_{\o,n}\mathcal{Z}_{\o,n}^*\right\| =  \left\|\sum_{(\o,n)\in \Gamma_p}\E   \|\PT(\vm{S}_{n,p}^\ddagger\bh_\o\vm{\phi}_{\o,n}^*)\|_{\F}^2\PT\left[\bh_\o\bh_\o^*\otimes \vm{\phi}_{\o,n}\vm{\phi}_{\o,n}^*\right]\PT  \right\|.
\end{align*}
We remind the reader that the expectation is only over $\vm{\phi}_{\o,n}$, not over the randomness in the construction of the partition $\Gamma_p$. Using the expansion in \eqref{eq:PTA-fro}, the above quantity is upper bounded by 
\begin{align}
&\left\|\sum_{(\o,n)\in \Gamma_p} \E\left(\|\vm{\phi}_{\o,n}\|_2^2|\bh_\o^*\vm{S}_{n,p}^\ddagger\hh|^2+ \|\vm{S}_{n,p}^\ddagger \bh_\o\|_2^2|\mm^*\vm{\phi}_{\o,n}|^2\right)\PT\left[\bh_\o\bh_\o^* \otimes \vm{\phi}_{\o,n}\vm{\phi}_{\o,n}^*\right]\PT\right\|\notag\\
&~~~\leq \max_n\|\hat{\vm{B}}\vm{S}_{n,p}^\ddagger\hh\|_\infty^2\left\|\sum_{(\o,n)\in \Gamma_p}\E\|\vm{\phi}_{\o,n}\|_2^2\PT\left[\bh_\o\bh_\o^* \otimes \vm{\phi}_{\o,n}\vm{\phi}_{\o,n}^*\right]\PT\right\|\notag\\
&~~~~~~~~~~~~~~~~~~~~~~~~+ \max_{\o,n}\|\vm{S}_{n,p}^\ddagger \bh_\o\|_2^2 \left\|\sum_{(\o,n)\in \Gamma_p}\E|\mm^*\vm{\phi}_{\o,n}|^2\PT\left[\bh_\o\bh_\o^* \otimes \vm{\phi}_{\o,n}\vm{\phi}_{\o,n}^*\right]\PT\right\|\notag \\
&~~~= \max_n\|\hat{\vm{B}}\vm{S}_{n,p}^\ddagger\hh\|_\infty^2\left\|\sum_{(\o,n)\in \Gamma_p} \PT\left[\bh_\o\bh_\o^* \otimes \E(\|\vm{\phi}_{\o,n}\|_2^2\vm{\phi}_{\o,n}\vm{\phi}_{\o,n}^*)\right]\PT\right\|\notag\\
&~~~~~~~~~~~~~~~~~~~~~~~+ \frac{16}{9}\mu_{\max}^2\frac{SL}{Q^2}\left\|\sum_{(\o,n)\in \Gamma_p} \PT\left[\bh_\o\bh_\o^* \otimes \E|\mm^*\vm{\phi}_{\o,n}|^2\vm{\phi}_{\o,n}\vm{\phi}_{\o,n}^*\right]\PT\right\|\label{eq:Variance-injectivity},
\end{align}
where in the last equality, we made use of the linearity of Kronecker operator, and the bound
\[
\max_{\o,n}\|\vm{S}_{n,p}^\ddagger \bh_\o\|_2^2 \leq \max_{\o,n} \|\PO\bh_\o\|_2^2\|\vm{S}_{n,p}^\ddagger\|_2^2 \leq \frac{16}{9}\mu_{\max}^2 \frac{SL}{Q^2},
\]
where we have in turn used the Cauchy Schwartz inequality; the fact that the non-zero rows and columns of $\vm{S}_{n,p}^\ddagger$ are supported on the index set $\Omega$; and the last inequality is the result of \eqref{eq:coherence_B}, together with \eqref{eq:S-eigen-bounds}. It is then easy to verify that 
\[
\E \|\vm{\phi}_{\o,n}\|_2^2\vm{\phi}_{\o,n}\vm{\phi}_{\o,n}^* =  ( K+2)\vm{D}_n,
\]
and
\[
 \E|\mm^*\vm{\phi}_{\o,n}|^2\vm{\phi}_{\o,n}\vm{\phi}_{\o,n}^* = (\|\mm_n\|_2^2 +2\mm_n\mm_n^*) \vm{D}_n.
\]
Applying this and the facts that $\PT= \PT\PO$, and $\|\PT\| \leq 1$, the above expression simplifies to 
\begin{align*}
\left\| \sum_{(\o,n)\in \Gamma_p}\E\mathcal{Z}_{\o,n}\mathcal{Z}_{\o,n}^*\right\| &\leq (K+2)\left(\max_n\|\hat{\vm{B}}\vm{S}_{n,p}^\ddagger\hh\|_\infty^2\right) \left\|\sum_{(\o,n)\in \Gamma_p} \PO\bh_\o\bh_\o^*\PO \otimes \vm{D}_n\right\|\\
&~~~~~~~+ \frac{16}{3}\mu_{\max}^2\frac{SL}{Q^2}\left(\max_{n} \|\mm_n\|_2^2\right)\left\|\sum_{(\o,n)\in \Gamma_p}\PO\bh_\o\bh_\o^*\PO \otimes  \vm{D}_n\right\|.
\end{align*}
Using the definition of coherences in \eqref{eq:coherence_h}, and \eqref{eq:coherence_m} and the fact that $\sum_{n} \vm{I}_K\otimes \vm{e}_n\vm{e}_n^* = \vm{I}_{KN}$, and that for an arbitrary matrix $\vm{A}$, $\|\vm{A}\otimes \vm{I}\| = \|\vm{A}\|$,
\begin{align*}
 \left\|\sum_{n} \sum_{\o \in \Delta_{n,p}} \E\mathcal{Z}_{\o,n}\mathcal{Z}_{\o,n}^*\right\| &\leq \left(\mu_0^2 \frac{(K+2)L}{Q^2} +\mu_{\max}^2\rho_0^2 \frac{SL}{Q^2N}\right) \cdot\max_n\left\| \sum_{\o \in \Delta_{n,p}}\PO\bh_\o\bh_\o^*\PO\right\|\notag\\
 & \leq \left(\frac{5}{4}\mu_0^2 \frac{K+2}{Q} +\frac{20}{3}\mu_{\max}^2\rho_0^2 \frac{S}{QN}\right),
\end{align*}
where the last inequality follows from \eqref{eq:RIPresults}. The computation for the second term in the maximum of variance expression \eqref{eq:variance-lemma1} follows a very similar route. In short, 
\begin{align}\label{eq:var-injectivity}
\left\| \sum_{(\o,n)\in \Gamma_p}\E\mathcal{Z}_{\o,n}^*\mathcal{Z}_{\o,n}\right\| &=  \left\|\sum_{(\o,n)\in \Gamma_p}\E \|\PT(\bh_\o\vm{\phi}_{\o,n}^*)\|_{\F}^2\PT(\bh_\o\bh_\o^*\vm{S}_{n,p}^\ddagger\otimes \vm{\phi}_{\o,n}\vm{\phi}_{\o,n}^*)\PT\right\|\notag\\
& \leq 3\left(\mu_0^2\frac{K}{L}+\mu_{\max}^2\rho_0^2\frac{S}{LN}\right)\cdot\max_n\left\|\sum_{\o \in \Delta_{n,p}}\PO\bh_\o\bh_\o^*\PO \vm{S}_{n,p}^\ddagger\right\|\notag\\
& \leq 3\left(\mu_0^2\frac{K}{L}+\mu_{\max}^2\rho_0^2\frac{S}{LN}\right),
\end{align}
the last line is the result of the definition of $\vm{S}_{n,p}^\ddagger$ in \eqref{eq:Snp}. This completes the calculation of the variance term. 

The second ingredient in the Bernstein bound is the calculation of the Orlicz norms of the summands. From \eqref{eq:Orlicz-norm-centered-X}, it follows that
\begin{align}\label{eq:orliczbound}
\|\mathcal{Z}_{\o,n} - \E \mathcal{Z}_{\o,n}\|_{\psi_1} & \leq  2\|\mathcal{Z}_{\o,n}\|_{\psi_1}.
\end{align}
Since $\mathcal{Z}_{\o,n}$ is rank-1, its operator norm simplifies to 
\[
\|\mathcal{Z}_{\o,n}\| = \|\PT(\bh_\o\vm{\phi}_{\o,n}^*)\|_{\F}\|\PT(\vm{S}_{n,p}^\ddagger \bh_\o\vm{\phi}_{\o,n}^*)\|_\F,
\]
and it is also easy to show that $\|\E\mathcal{Z}_{\o,n}\| = 1$. Let us begin by showing that the random variable $\|\PT(\bh_\o\vm{\phi}_{\o,n}^*)\|_{\F}$ is subgaussian. For this it is enough to prove that $\|\| \PT(\bh_\o\vm{\phi}_{\o,n}^*)\|_{\F}\|_{\psi_2}^2 < \infty$. Now using \eqref{eq:Square-Subgaussian}, we have
\begin{align}\label{eq:PTAfro-Orlicz}
\|\| \PT(\bh_\o\vm{\phi}_{\o,n}^*)\|_{\F}\|_{\psi_2}^2  &\leq \|\| \PT(\bh_\o\vm{\phi}_{\o,n}^*)\|_{\F}^2\|_{\psi_1}\notag\\
&\leq \left(\|\hh\|_\infty^2 \cdot\| \|\vm{\phi}_{\o,n}\|_2^2\|_{\psi_1}+ \mu_{\max}^2 \frac{S}{L} \||\<\mm,\vm{\phi}_{\o,n}\>|^2\|_{\psi_1}\right)\notag\\
&\leq C \left(\mu_0^2 \frac{K}{L} + \mu_{\max}^2\rho_0^2\frac{S}{LN}\right).
\end{align}
In a very similar manner, one shows that 
\begin{align*}
\|\| \PT(\vm{S}_{n,p}^\ddagger\bh_\o\vm{\phi}_{\o,n}^*)\|_{\F}\|_{\psi_2}^2  &\leq \|\| \PT(\vm{S}_{n,p}^\ddagger\bh_\o\vm{\phi}_{\o,n}^*)\|_{\F}^2\|_{\psi_1}\\
&\leq C\left(\|\vm{B}\vm{S}_{n,p}^\ddagger\hh\|_\infty^2 \cdot\| \|\vm{\phi}_{\o,n}\|_2^2\|_{\psi_1}+ \mu_{\max}^2 \frac{S}{L} \||\<\mm,\vm{\phi}_{\o,n}\>|^2\|_{\psi_1}\right)\\
&\leq C \left(\mu_0^2 \frac{KL}{Q^2} + \mu_{\max}^2\rho_0^2\frac{SL}{Q^2N}\right).
\end{align*}
Using \eqref{eq:Product-Subgaussians}, we now obtain 
\begin{align*}
\|\mathcal{Z}_{\o,n}\|_{\psi_1} &\leq C\|\| \PT(\vm{S}_{n,p}^\ddagger\bh_\o\vm{\phi}_{\o,n}^*)\|_{\F}\|_{\psi_2}\cdot\|\| \PT(\bh_\o\vm{\phi}_{\o,n}^*)\|_{\F}\|_{\psi_2}\\
& \leq C\frac{L}{Q}\left(\mu_0^2 \frac{K}{L}  + \mu_{\max}^2\rho_0^2\frac{S}{LN}\right).
\end{align*}
Thus $\|\mathcal{Z}_{\o,n} - \E \mathcal{Z}_{\o,n}\|$ is a sub-exponential random variable; hence, $\alpha = 1$ in \eqref{eq:matbernpsi}, which gives 
\begin{equation}\label{eq:orlicznorm-injectivity}
\log\left(\frac{(QN) U_1^2}{\sigma^2_{\mathcal{Z}}}\right) \leq C\log M, 
\end{equation}
where $M = \mu_0^2KN+\mu_{\max}^2\rho_0^2S \leq C(LN)$, the last inequality being the result of  $\mu_0^2 \leq L$, $\mu_{\max}^2 \leq L$, and $\rho_0^2 \leq N$, and that $L \geq S, K$. 

Plugging the upper bound on the variance in \eqref{eq:var-injectivity}, and \eqref{eq:orlicznorm-injectivity} in \eqref{eq:matbernpsi}, we have by using $t = \beta\log (LN)$ that
\begin{align}\label{eq:injectivity-rate}
&\|\PT\cA_p^*\cA_p\mathcal{S}_p^\ddagger\PT-\PT\| \leq \notag\\ &C\max\left\{\sqrt{\beta\left(\mu_0^2 \frac{K}{Q} + \mu_{\max}^2\rho_0^2\frac{S}{QN}\right)\log (LN)},  \beta\left(\mu_0^2 \frac{K}{Q} + \mu_{\max}^2\rho_0^2\frac{S}{QN}\right)\log^2 (LN)\right\}
\end{align} 
with probability at least $1-(LN)^{-\beta}$. The result in the statement of the lemma follows by choosing $Q$, and $N$ as in the statement of the lemma for a sufficiently large constant $C$. 
\end{proof}
\subsection{Proof of Lemma \ref{lem:injectivitype1part1}}
\begin{proof}
Lemma \ref{lem:injectivity} considers bounding the quantity $\| \PT\cA_p^*\cA_p\mathcal{S}_p^\ddagger\PT-\PT\|$ for $p\geq 2$. The proof of this corollary is very similar to the proof of Lemma \ref{lem:injectivity}, and basically follows by replacing $\cA_p$ with $\cA_1$, and $\mathcal{S}_p^\ddagger$ with $(L/Q)\PO$. Therefore, we lay out the steps very briefly. Start by expressing the quantity of interest as a sum of independent random linear maps 
\[
\frac{L}{Q}\PT(\cA_1^*\cA_1-\E\cA_1^*\cA_1)\PT = \frac{L}{Q}\sum_{(\o,n) \in \Gamma_1}\PT \left(\bh_\ell\bh_\ell^*\otimes\left(\vm{\phi}_{\o,n}\vm{\phi}_{\o,n}^*-\E\vm{\phi}_{\o,n}\vm{\phi}_{\o,n}^*\right)\right)\PT,
\]
Define $\mathcal{Z}_{\o,n} : = (L/Q)\PT \left(\bh_\ell\bh_\ell^*\otimes\vm{\phi}_{\o,n}\vm{\phi}_{\o,n}^*\right)\PT$. Since the linear maps $\mathcal{Z}_{\o,n}$ can be thought of as symmetric matrices 
\[
\mathcal{Z}_{\o,n} = \PT(\bh_\o\vm{\phi}_{\o,n}^*) \boxtimes \PT(\bh_\o\vm{\phi}_{\o,n}^*),
\]
this means $\mathcal{Z}_{\o,n}^*\mathcal{Z}_{\o,n} = \mathcal{Z}_{\o,n}\mathcal{Z}_{\o,n}^*$. Therefore, for the variance it suffices to compute 
\[
\left\|\sum_{(\o.n) \in \Gamma_1}\mathcal{Z}_{\o,n}^*\mathcal{Z}_{\o,n}\right\| = \frac{L^2}{Q^2} \left\|\sum_{(\o.n) \in \Gamma_1} \|\PT(\bh_\o\vm{\phi}_{\o,n}^*)\|_\F^2\PT \left[\bh_\ell\bh_\ell^*\otimes\vm{\phi}_{\o,n}\vm{\phi}_{\o,n}^*\right]\PT\right\|.
\]
Using \eqref{eq:PTA-fro}, and a calculation that goes along the same lines as before, and employing \eqref{eq:variance-lemma1}, we can compute the upper bound on the variance
\begin{align*}
\sigma_{\mathcal{Z}}^2 &\leq \frac{L^2}{Q^2}\left[3K \|\hat{\vm{B}}\hh\|_\infty^2+3C\mu_{\max}^2 \frac{S}{L}\left(\max_n \|\mm_n\|_2^2\right)\right]\cdot\max_n\left\|\sum_{\o\in \Delta_{n,1}} \PO\bh_\o\bh_\o^* \PO\right\|\\
& \leq 3\left[K\frac{\mu_0^2}{Q} + \mu_{\max}^2 \frac{S}{Q}\rho_0^2 \frac{1}{N}\right].
\end{align*}
The Orlicz norm turns out to be $\|\mathcal{Z}_{\o,n} - \E\mathcal{Z}_{\o,n} \|_{\psi_1} \leq 2 \|\mathcal{Z}_{\o,n}\|_{\psi_1} \leq C (L/Q)\| \|\PT(\bh_{\o}\vm{\phi}_{\o,n}^*)\|_\F^2 \|_{\psi_1}$, which from \eqref{eq:PTAfro-Orlicz} is 
\[
\|\mathcal{Z}_{\o,n} - \E\mathcal{Z}_{\o,n} \|_{\psi_1} \leq  C \left(\mu_0^2 \frac{K}{Q} + \mu_{\max}^2\rho_0^2\frac{S}{QN}\right).
\]
This shows that $\mathcal{Z}_{\o,n}$ are sub-exponential, and thus using $\alpha = 1$ in the definition of $U_1$ in Proposition \ref{prop:matbernpsi}, one obtains 
\[
\log\left(\frac{QNU_1^2}{\sigma_\mathcal{Z}^2}\right) \leq C \log\left(\mu_0^2KN + \mu_{\max}^2 \rho_0^2 S\right) \leq C\log(LN).
\]
With this, we have all the ingredients to compute the deviation bound. Apply Bernstein inequality in Proposition \ref{prop:matbernpsi}, and choose $t = (\beta-1)\log(LN)$ to obtain 
\begin{align}\label{eq:injectivity-rate3}
&\frac{L}{Q}\|\PT\cA_1^*\cA_1\PT-\E\PT\cA_1^*\cA_1\PT\|\notag\\ &\leq C\max\left\{\sqrt{\beta\left(\mu_0^2 \frac{K}{Q} + \mu_{\max}^2\rho_0^2\frac{S}{QN}\right)\log (LN)},  \beta\left(\mu_0^2 \frac{K}{Q} + \mu_{\max}^2\rho_0^2\frac{S}{QN}\right)\log^2(LN)\right\}
\end{align} 
with probability at least $1-(LN)^{-\beta}$. The choice of $Q$, and $N$ in \eqref{eq:Q-bound-injectivity3} for a sufficiently large constant $C$ in there guarantees that the right hand side is smaller than $\frac{1}{8}\sqrt{\frac{Q}{L}}$.
\end{proof}
\subsection{Proof of Corollary \ref{cor:injectivitype1part1}}
\begin{proof}
	The proof is exactly same as in Lemma \ref{lem:injectivitype1part1} except now instead of using $\cA_1$ that is defined over a the subset $\Gamma_1 \subset [L]\times [N]$, we instead use a linear operator $\cA$ defined over the entire set $\Gamma := \{(\o,n) \in [L]\times [N]\}$ of measurements. This means, we only need to replace $\cA_1$ with $\cA$, defined in \eqref{eq:cA}, and $Q = |\Delta_{n,1}|$ with $L=|\Gamma|$. In addition, note that $\E \cA^*\cA = \mathcal{I},$ where $\mathcal{I}$ is an identity operator, i.e., $\mathcal{I}(\X) = \X$ for a matrix $\X$.  Making these changes in the proof of Lemma \ref{lem:injectivitype1part1}, one obtains the result claimed in the statement of the corollary.
\end{proof}
\subsection{Proof of Lemma \ref{lem:injectivity1part2}}
\begin{proof}
 Note that by using the definition of $\cA_1$ in \eqref{eq:cAp}, we have 
\begin{align*}
\PT\E\cA_1^*\cA_1\PT &=  \sum_{(\o,n)\in \Gamma_1}\PT\left(\bh_\o\bh_\o^* \otimes\E \vm{\phi}_{\o,n}\vm{\phi}_{\o,n}^* \right)\PT \\
&=\sum_n \PT \left(\sum_{\o \in \Delta_{n,1}} \bh_\o\bh_\o^* \otimes \vm{D}_n\right)\PT \neq \PT. 
\end{align*}
The operator norm of the quantity of interest can then be simplified using the definition of $\vm{S}_{n,1}$ in Section \ref{sec:Golfing}, and the facts that $\PT = \PO\PT = \PT\PO$, and $\PT\PO\PT = \PT$ as follows
\begin{align}\label{eq:eq1}
&\left\|\frac{L}{Q}\PT\E\cA_1^*\cA_1\PT -\PT\right\| = \left\|\sum_n \frac{L}{Q}\PT \left(\sum_{\o \in \Delta_{n,1}} \bh_\o\bh_\o^* \otimes \vm{D}_n\right)\PT-\PT\right\|\notag\\
& ~~~~~~~~~~= \left\|\sum_n \frac{L}{Q}\PT \left(\sum_{\o \in \Delta_{n,1} }\PO \bh_\o \bh_\o^* \PO\otimes \vm{D}_n\right)\PT-\PT\right\|\notag\\
& ~~~~~~~~~~\leq \frac{L}{Q}\|\PT\| \left\| \sum_n \sum_{\o \in \Delta_{n,1}}\PO \bh_\o \bh_\o^* \PO   \otimes \vm{D}_n -\frac{Q}{L} \PO\right\|\|\PT\|\leq \frac{L}{Q} \cdot \max_n \left\|\sum_{\o \in \Delta_{n,1}}\PO \bh_\o \bh_\o^* \PO- \frac{Q}{L}\PO\right\|. 
\end{align}
We can apply the uniform result in \eqref{eq:RIPresults} to bound the last quantity above, but really a weaker nonuniform result from the compressive sensing literature \cite{candes2007sparsity} suffices, and also results in an overall tighter bound on $Q$ than a simple application of the uniform result in \eqref{eq:QRIPbound} would give. The non-uniform result says that there exists a constant $C$, such that for any $ 0 < \epsilon < 1$, and $ 0 < \delta < 1$, whenever 
\begin{equation*}
 Q \geq C\frac{\mu_{\max}^2 S}{\epsilon^2}\max\{\log L, \log(1/\delta)\}, 
 \end{equation*}
then 
\begin{equation*}
\left\|\sum_{\o \in \Delta_{n,1} } \PO \bh_{\o}\bh_{\o}^* \PO - \frac{Q}{L}\PO\right\| \leq \frac{\epsilon Q}{L}
\end{equation*}
with probability at least $1-\delta$. Specifically, taking $ \epsilon = 1/8\sqrt{Q/L}$, $\delta = (LN)^{-\beta}$ says that taking 
\begin{equation*}
Q \geq C\beta\mu_{\max}^2S (L/Q)\log (LN),
\end{equation*}
is enough to guarantee that   
\begin{equation*}
\left\|\sum_{\o \in \Delta_{n,1} } \PO \bh_{\o}\bh_{\o}^* \PO - \frac{Q}{L}\PO\right\| \leq \frac{1}{8}\left(\frac{Q}{L}\right)^{3/2}
\end{equation*}
with probability at least $1-(LN)^{-\beta}$. Now a union bound over $N$ sets in the partition $\{\Delta_{n,1}\}_n$ means that 
\begin{equation*}
\max_n\left\|\sum_{\o \in \Delta_{n,1} } \PO \bh_{\o}\bh_{\o}^* \PO - \frac{Q}{L}\PO\right\| \leq \frac{1}{8}\left(\frac{Q}{L}\right)^{3/2}
\end{equation*}
with probability at least $1-N(LN)^{-\beta} \geq 1 - (LN)^{1-\beta}$. 

Plugging this result in \eqref{eq:eq1} proves the lemma. 
\end{proof}

\subsection{Proof of Lemma \ref{lem:rho-p-bound}}
\begin{proof}
	 As a direct implication of the restricted isometry property in \eqref{eq:RIPresults}, and later using \eqref{eq:S-eigen-bounds}, one obtains
	\begin{align*}
	\sum_{\ot \in \Delta_{\nt,p}} \|\bh_{\ot}^* \vm{S}_{\nt,p}^\ddagger \w_p \vm{D}_{\nt}\|_2^2 \leq \frac{5Q}{4L} \|\vm{S}_{\nt,p}^\ddagger \w_p \vm{D}_{\nt}\|_2^2 \leq \frac{20}{9}\frac{L}{Q}\|\w_p\vm{D}_{\nt}\|_\F^2.
	\end{align*}
	From the definition of $\w_p$ in \eqref{eq:Wp-def2}, and Cauchy Schwartz inequality, we have 
	\begin{align*}
	\| \w_p \vm{D}_{\nt}\|_\F^2 \leq\left[\prod_{k=2}^p\|\PT\cA_p^*\cA_p \mathcal{S}_p^\ddagger \PT-\PT\|^2\right]\left[ \left\|\frac{L}{Q}\PT\cA_1^*\cA_1\PT-\PT\right\|^2\right]\max_{\nt}\|\w_0\vm{D}_{\nt}\|_\F^2.
	\end{align*}
	Using \eqref{eq:injectivitypart1-complete}, and \eqref{eq:injectivity-result1}, we can verify that 
	\[
	\rho_p^2 \leq  4^{-p}  \frac{Q}{L}N\max_{\nt} \|\w_0\vm{D}_{\nt}\|_\F^2, ~\mbox{for each}~ p \in \{ 1, 2,3,\ldots, P\}. 
	\]
	Finally, from the definition of $\rho_0^2$, and that $\w_0 = -\hh\mm^*$, we can conclude that
	\[
	\|\w_0\vm{D}_{\nt}\|_\F^2 = \|\hh\|_2^2 \|\mm^*\vm{D}_{\nt}\|_2^2 = \|\mm_{\nt}\|_2^2 \leq \rho_0^2/N,
	\]
	which, if plugged back in the above bound on $\rho_p^2$, completes the proof of the lemma. 
\end{proof}
\subsection{Proof of Lemma \ref{lem:nup-bound}}
\begin{proof}
	Our exposition is different for $p \in \{2, 3,\ldots, P\},$ and $p =1$ owing to the difference in the iterative construction of the dual certificate for these choice of $p$. We start by considering $\nu_p^2$ for $p \in \{2,3,\ldots, P\}$.
	The following lemma provides an upper bound on $\|\bh_{\ot}^*\vm{S}_{\nt,p+1}^\ddagger\w_p\vm{D}_{\nt}\|_2^2$ that is in turn used to bound $\nu_p^2$.
	\begin{lem}\label{lem:qln-bound}
		Let 
		\begin{align}\label{eq:tVZ}
		\Pi^2_p :=  C\frac{Q}{L} K |\bh_{\ot}^*\vm{S}^\ddagger_{\nt,p+1}\hh|^2&\left[\max_{\o \in \Delta_{\nt,p}}\|\bh_\o^*\vm{S}_{\nt,p}^\ddagger \w_{p-1}\vm{D}_{\nt}\|_2^2\right]\notag\\
		&\qquad\qquad+ C \mu_{\max}^2 \frac{S}{Q}\|\mm_{\nt}\|_2^2 \left[\max_n \max_{\o \in \Delta_{n,p}} \|\bh_{\o}^*\vm{S}_{n,p}^\ddagger\w_{p-1}\vm{D}_n\|_2^2\right], 
		\end{align}
		and
		\begin{align}\label{eq:tUZ}
		\Lambda^2_p &:= C \left(\frac{\mu^2_0K}{L}|\bh_{\ot}^*\vm{S}_{\nt,p+1}^\ddagger\hh|^2+\mu^4_{\max}\rho_0^2\frac{S^2}{Q^2N}\|\mm_{\nt}\|_2^2\right)\left[\max_{\o \in \Delta_{\nt,p}}\left\|\bh_\o^* \vm{S}_{\nt,p}^\ddagger\w_{p-1}\vm{D}_{\nt}\right\|_2^2\right].
		\end{align}
		Let $\w_p$, $\vm{S}_{n,p}^\ddagger$ be as defined in \eqref{eq:Wp-def}, and \eqref{eq:Snp}; and assume that the restricted isometry property in \eqref{eq:RIPresults} holds. Fix $\beta \geq 1$. Then there exists a constant $C$ such that
		\begin{equation}\label{eq:qln-l2-norm-bound}
		\|\bh_{\ot}^*\vm{S}_{\nt,p+1}^\ddagger \w_p\vm{D}_{\nt}\|_2^2  \leq C\max\left\{\beta \Pi^2_p\log (LN), \beta^2 \Lambda^2_p\log^4 (LN)\right\} , ~\mbox{for every}~ p \in \{2,3, \ldots, P\}.
		\end{equation}
		with probability at least $1- (LN)^{-\beta}$. 
	\end{lem}
	Lemma \ref{lem:nup-bound} is established in Section \ref{sec:Supporting-Lemmas}. Using \eqref{eq:qln-l2-norm-bound}, it is clear that the upper bound on $\nu_p^2$ in \eqref{eq:nu_Wp} can be obtained by evaluating the maximum of the quantities $\Pi_p$, and $\Lambda_p$  over $\ot \in \Delta_{\nt,p+1}$, and $\nt \in [N]$. 

	Putting this together with \eqref{eq:coherence_h}, \eqref{eq:coherence_m}, and \eqref{eq:nu_Wp} directly implies that 
	\[
	\max_{\nt}\left[\max_{\ot \in \Delta_{\nt,p+1}}\Pi^2_p \right]\leq C\left(\mu_0^2 \frac{K}{Q}+\mu_{\max}^2\rho_0^2 \frac{S}{QN} \right)\nu_{p-1}^2 \frac{L}{Q^2N}, 
	\]
	and
	\[
	\max_{\nt}\left[\max_{\ot \in \Delta_{\nt,p+1}}\Lambda_p^2\right] \leq C\left(\mu_0^4 \frac{K}{Q^2} + \rho_0^4 \mu_{\max}^4 \frac{S^2}{Q^2N^2}\right)\nu_{p-1}^2 \frac{L}{Q^2N}.
	\]
	Choosing $Q$, and $N$ as in Lemma \ref{lem:nup-bound} results in 
	\begin{align}\label{eq:nup2g}
	\nu_p^2\leq \frac{1}{4}\nu_{p-1}^2 \mbox{~for each~} p \in \{2, 3,\ldots, P\}
	\end{align}
	with probability at least $1-(LN)^{-\beta}$. 
	
	For the remaining case of $p=1$, 
	\begin{align*}
	\nu_1^2 &= \frac{Q^2}{L}N \max_{\nt} \left[\max_{\ot \in \Delta_{\nt,2}}\left\|\bh_{\ot}^* \vm{S}_{\nt,2}^\ddagger \w_1 \vm{D}_{\nt}\right\|_2^2\right], 
	\end{align*}
	and using the definition of $\w_1$ in \eqref{eq:Wp-def2}
	\begin{align}\label{eq:nu1-two-parts}
	\nu_1^2 &=\frac{Q^2}{L}N \max_{\nt} \left[\max_{\ot \in \Delta_{\nt,2}}\left\|\bh_{\ot}^* \vm{S}_{\nt,2}^\ddagger\left[ \frac{L}{Q}\PT \cA_1^*\cA_1 (\hh\mm^*)-\hh\mm^*\right] \vm{D}_{\nt}\right\|_2^2\right]\notag\\
	&\leq 2\frac{Q^2}{L}N \max_{\nt}\left[\max_{\ot \in \Delta_{\nt,2}}\left\|\bh_{\ot}^* \vm{S}_{\nt,2}^\ddagger\left[ \frac{L}{Q}\PT \cA_1^*\cA_1 (\hh\mm^*)-\frac{L}{Q}\E\PT \cA_1^*\cA_1 (\hh\mm^*)\right] \vm{D}_{\nt}\right\|_2^2\right] + \notag\\
	&~~~~~~~~~ ~~~~~~~~2\frac{Q^2}{L}N \max_{\nt}\left[\max_{\ot \in \Delta_{\nt,2}}\left\|\bh_{\ot}^* \vm{S}_{\nt,2}^\ddagger\left[\frac{L}{Q}\E\PT \cA_1^*\cA_1 (\hh\mm^*)-\hh\mm^*\right] \vm{D}_{\nt}\right\|_2^2\right].
	\end{align}
	The following corollary provides a bound on the first term in the sum above. 
	
   \begin{cor}[Corollary of Lemma \ref{lem:qln-bound}]\label{cor:qln-bound}
   	Let 
		\begin{align}\label{eq:G1}
		\Pi^2_1 :=  C\frac{Q}{L} K |\bh_{\ot}^*\vm{S}^\ddagger_{\nt,2}\hh|^2&\left[ \max_n \max_{\o \in \Delta_{n,1}}\|\bh_\o^*\frac{L}{Q}\mathcal{P}\w_0\vm{D}_{\nt}\|_2^2\right]\notag\\
		&\qquad\qquad+ C \mu_{\max}^2 \frac{S}{Q}\|\mm_{\nt}\|_2^2 \left[\max_n \max_{\o \in \Delta_{n,1}} \|\bh_\o^*\frac{L}{Q}\mathcal{P}\w_0\vm{D}_n\|_2^2\right], 
		\end{align}
		and
		\begin{align}\label{eq:H1}
		\Lambda^2_1 &:= C \left(\frac{\mu^2_0K}{L}|\bh_{\ot}^*\vm{S}_{\nt,2}^\ddagger\hh|^2+\mu^4_{\max}\rho_0^2\frac{S^2}{Q^2N}\|\mm_{\nt}\|_2^2\right)\left[\max_n\max_{\o \in \Delta_{n,1}}\left\|\bh_\o^* \frac{L}{Q}\mathcal{P}\w_{0}\vm{D}_{\nt}\right\|_2^2\right].
		\end{align}
		Let $Q$ be as in Lemma \ref{lem:qln-bound}. Let 
		\begin{align}\label{eq:sln}
		\vm{s}_{\ot,\nt}^* &:= \bh_{\ot}^*\vm{S}^\ddagger_{\nt,2}\left[\frac{L}{Q}\PT\cA_1^*\cA_1(\hh\mm^*)- \E\frac{L}{Q}\PT\cA_1^*\cA_1(\hh\mm^*)\right]\vm{D}_{\nt},
		\end{align}
		Assume that the restricted isometry property in \eqref{eq:RIPresults} holds. Fix $\beta \geq 1$. Then there exists a constant $C$ such that 
		\[
		\|\vm{s}_{\ot,\nt}\|_2^2 \leq C\max\left\{\beta \Pi_1^2\log (LN) , \beta^2\Lambda_1^2\log^4 (LN)\right\}
		\]
		with probability at least $1- (LN)^{-\beta}$. 
	\end{cor}	
	 The proof of this corollary is provided in Section \ref{sec:Supporting-Lemmas}.
    
    It is easy to see using the definitions \eqref{eq:coherence_h}, and \eqref{eq:coherence_m} that 
    \[
    \max_{\nt}\max_{\ot \in \Delta_{\nt,2}} \Pi_1^2 \leq C\left(\mu_0^2\frac{K}{Q}+\mu_{\max}^2\rho_0^2 \frac{S}{QN}\right)\mu_0^2\rho_0^2 \frac{L}{Q^2N}, 
    \]
    and 
    \[
     \max_{\nt}\max_{\ot \in \Delta_{\nt,2}} \Lambda_1^2 \leq C\left(\mu_0^4 \frac{K}{Q^2} + \rho_0^4\mu_{\max}^4\frac{S}{QN}\right) \mu_0^2\rho_0^2 \frac{L}{Q^2N}.
    \]
    Using these calculations, the first term in \eqref{eq:nu1-two-parts} can be bounded by applying Corollary \ref{cor:qln-bound}, and choosing $Q$ as in Lemma \ref{eq:nup-bound} for a large enough constant $C$ to achieve
	\begin{equation}\label{eq:nu1-bound-part1}
	\frac{Q^2}{L}N\max_{\nt} \left[\max_{\ot \in \Delta_{\nt,2}}\|\vm{s}_{\ot,\nt}\|_2^2 \right] \leq \frac{1}{4}\mu_0^2\rho_0^2.
	\end{equation}
	
	As for the second term in the sum in \eqref{eq:nu1-two-parts}, the lemma below provides an upper bound. 
	\begin{lem}\label{lem:rln}
		Define 
		\begin{equation}\label{eq:rln}
		\vm{r}^*_{\ot,\nt} :=  \bh_{\ot}^*\vm{S}^\ddagger_{\nt,2}\left[\frac{L}{Q}\E\PT\cA_1^*\cA_1(\hh\mm^*)- \hh\mm^*\right]\vm{D}_{\nt},
		\end{equation}
		and assume that \eqref{eq:S-eigen-bounds} holds. Then 
		\begin{equation}\label{eq:rln-bound}
		\|\vm{r}_{\ot,\nt}\|_2^2 \leq \left(\frac{9}{2} \|\bh_{\ot}^*\vm{S}_{\nt,2}^\ddagger\hh\|_2^2 + 2\frac{L^2}{Q^2}\max_n \|\bh_{\ot}^*\vm{S}_{\nt,2}^\ddagger\vm{S}_{n,1}\hh\|_2^2\right)\|\mm_{\nt}\|_2^2.
		\end{equation}
	\end{lem}
	\begin{proof}
		The proof of this lemma is provided in Section \ref{sec:Supporting-Lemmas}.
	\end{proof}
	
	As for the second term in \eqref{eq:nu1-two-parts}, we appeal to Lemma \ref{lem:rln}, and directly obtain after evaluating the maximum over $\ot \in \Delta_{\nt,2}$, and $\nt \in [N]$ in \eqref{eq:rln-bound}, and using the definition of coherences $\mu_0^2$, and $\rho_0^2$, that  
	\begin{equation}\label{eq:nu1-bound-part2}
	\frac{Q^2}{L}N \max_{\nt} \left[\max_{\ot \in \Delta_{\nt,2}}\|\vm{r}_{\ot,\nt}\|_2^2 \right] \leq 6.5 \mu_0^2\rho_0^2.
	\end{equation}
	Plugging \eqref{eq:nu1-bound-part1}, and \eqref{eq:nu1-bound-part2} in \eqref{eq:nu1-two-parts} shows that $\nu_1 \leq 4 \mu_0\rho_0$. Combining this fact with \eqref{eq:nup2g} completes the proof of the lemma. 
\end{proof}
\subsection{Proof of Lemma \ref{lem:mup-bound}}
\begin{proof}
	We first consider the case for $p \in \{2, 3,\ldots, P\}$.  Using Lemma \ref{lem:qln-bound}, we have a bound on the quantity $\|\bh_{\ot}^* \vm{S}_{n,p}^\ddagger \w_p \vm{D}_{\nt}\|_2^2$ for $p \geq 2$ in terms of $\Lambda^2_p$, and $\Pi^2_p$ in \eqref{eq:tUZ}, and \eqref{eq:tVZ}, respectively. A corresponding bound on $\mu_p^2$ above is then 
	\begin{equation}\label{eq:mup-bound-intermediate}
		\mu_p^2 \leq C\frac{Q^2}{L}\max\left\{\sum_{\nt}\left[\max_{\ot\in \Delta_{\nt,p}}\Pi^2_p\right] \beta\log (LN) , \sum_{\nt}\left[\max_{\ot\in \Delta_{\nt,p}}\Lambda^2_p\right] \beta^2\log^4 (LN)\right\}.
	\end{equation} 
	Using the definitions of the coherences $\nu_p^2$, and $\mu_p^2$ of the iterates $\w_p$ in \eqref{eq:nu_Wp}, and  \eqref{eq:mup-def}; the definition of the coherence $\mu_0^2$ of $\hh$ in \eqref{eq:coherence_h}; and the fact that $\sum_{\nt} \|\mm_{\nt}\|_2^2 = 1$, one can verify from \eqref{eq:tUZ}, and \eqref{eq:tVZ} that 
	\begin{align*}
	\sum_{\nt}\left[\max_{\ot \in \Delta_{\nt,p}}\Pi^2_p\right] \leq C\left(\mu_0^2\mu_{p-1}^2 K\frac{L}{Q^3}+ \frac{16}{9}\mu_{\max}^2\nu_{p-1}^2 \frac{SL}{Q^3N}\right),
	\end{align*}
	and 
	\begin{align*}
	\sum_{\nt}\left[\max_{\ot \in \Delta_{\nt,p}}\Lambda^2_p\right]\leq C\left(\mu_0^4\frac{K}{Q^2} + \mu_{\max}^4\rho_0^4 \frac{S^2}{Q^2N^2} \right)\mu^2_{p-1}\frac{L}{Q^2}.
	\end{align*}
	Plugging back in \eqref{eq:mup-bound-intermediate},  and choosing $Q$ and $N$ as in Lemma \ref{lem:nup-bound} for a large enough $C$ is sufficient to guarantee that 
	\begin{equation}\label{eq:mup-intermediate-result2}
	\mu_p^2 \leq \delta\left(\mu_{p-1}^2 + \frac{1}{N}\nu_{p-1}^2\right), \mbox{for each}~ p \in  \{2,3,\ldots, P\}.
	\end{equation}
	for arbitrarily small number $\delta$  that lies between zero and one. Bounding $\mu_p^2$ iteratively using the above relation gives $\mu_p^2 \leq \delta^{p-1}\mu_{1}^2 + (2\delta/N)\nu_{p-1}^2,  ~ p \in \{2, 3,\ldots, P\},$ and now using \eqref{eq:nup-bound}, and the fact that $\rho_0^2 \leq N$, we have 
	\begin{equation}\label{eq:mup-bound-pg2}
	\mu_p^2 \leq \delta^{p-1}\mu_1^2 + 2\delta 4^{-p+3} \mu_0^2, ~ \mbox{for each}~ p \in \{ 2,3,\ldots, P\}.
	\end{equation}
	 All that remains now is to bound $\mu_1^2$, which from \eqref{eq:mup-def} is
	\begin{align*}
	\mu_1^2 &:= \frac{Q^2}{L} \sum_{\nt} \left[\max_{\ot \in \Delta_{\nt,2}} \left\|\bh_{\ot}^* \vm{S}_{\nt,2}^\ddagger \w_1\vm{D}_{\nt}\right\|_2^2\right].
	\end{align*}
	Using the definition of $\w_1$ in \eqref{eq:Wp-def} followed by an application of a simple identity $(a+b)^2 \leq 2(a^2+b^2)$, we have 
	\begin{align*}
	\left\|\bh_{\ot}^* \vm{S}_{\nt,2}^\ddagger \w_1\vm{D}_{\nt}\right\|_2 \leq 2 \left\|\bh_{\ot}^* \vm{S}_{\nt,2}^\ddagger\left[ \frac{L}{Q}\PT \cA_1^*\cA_1 (\hh\mm^*)-\frac{L}{Q}\E\PT \cA_1^*\cA_1 (\hh\mm^*)\right] \vm{D}_{\nt}\right\|_2^2 +\\
	 2\left\|\bh_{\ot}^* \vm{S}_{\nt,2}^\ddagger\left[\frac{L}{Q}\E\PT \cA_1^*\cA_1 (\hh\mm^*)-\hh\mm^*\right] \vm{D}_{\nt}\right\|_2^2.\notag
	\end{align*}
   An application of Corollary \ref{cor:qln-bound} gives
	\begin{align}\label{eq:mupbound-stage1}
	\frac{Q^2}{L} \sum_{\nt}\max_{\ot\in \Delta_{\nt,2}} &\left\|\bh_{\ot}^* \vm{S}_{\nt,2}^\ddagger\left[ \frac{L}{Q}\PT \cA_1^*\cA_1 (\hh\mm^*)-\frac{L}{Q}\E\PT \cA_1^*\cA_1 (\hh\mm^*)\right] \vm{D}_{\nt}\right\|_2^2\notag\\
	 &\qquad \qquad\leq C\frac{Q^2}{L}\sum_{\nt=1}^N\max\left\{\beta\max_{\ot \in \Delta_{\nt,2}} \Pi_1^2\log (LN), \beta^2 \max_{\ot \in \Delta_{\nt,2}}\Lambda_1^2 \log^4(LN)\right\}.
	\end{align}
	Using the fact that $\w_0 = \hh\mm^*$, it is easy to see that 
	\begin{align*}
	\frac{Q^2}{L}\sum_{\nt=1}^N \max_{\ot \in \Delta_{\nt,2}} \Pi_1^2& \leq C\left(KQ\left[\max_{\nt}\max_{\ot \in \Delta_{\nt,2}} |\bh_{\ot}^*\vm{S}_{\nt,2}^\ddagger \hh|^2\right]+\mu_{\max}^2\frac{SL}{Q}  \max_n \|\mm_n\|_2^2\right)\left[\max_n \max_{\o \in \Delta_{n,1}} |\bh_\o^*\hh|^2\right], 
	\end{align*}
	where we have used the fact that the fact that $\sum_{\nt} \|\mm_{\nt}\|_2^2 = 1$. By definitions \eqref{eq:coherence_h}, and \eqref{eq:coherence_m}, we have the upper bound
	\begin{align*}
	\frac{Q^2}{L}\sum_{\nt=1}^N \max_{\ot \in \Delta_{\nt,2}} \Pi_1^2& \leq C\left(\mu_0^4 \frac{K}{Q} +\mu_{\max}^2 \mu_0^2 \rho_0^2 \frac{S}{QN}\right).
	\end{align*}
	In a similar manner, one can show that 
	\begin{align*}
	\frac{Q^2}{L}\sum_{\nt=1}^N \max_{\ot \in \Delta_{\nt,2}} \Lambda_1^2 \leq C\left(\mu_0^4\frac{K}{Q^2} +\mu_{\max}^4 \frac{S^2}{Q^2} \rho_0^4 \frac{1}{N^2} \mu_0^2\right). 
	\end{align*}
	Now with the choice of $Q$, and $N$ in the statement of lemma for some suitably large constant, one can show that 
    \begin{align*}
    &\frac{Q^2}{L} \sum_{\nt}\max_{\ot\in \Delta_{\nt,2}} \left\|\bh_{\ot}^* \vm{S}_{\nt,2}^\ddagger\left[ \frac{L}{Q}\PT \cA_1^*\cA_1 (\hh\mm^*)-\frac{L}{Q}\E\PT \cA_1^*\cA_1 (\hh\mm^*)\right] \vm{D}_{\nt}\right\|_2^2 \leq \frac{1}{2}\mu_0^2,
    \end{align*}
    where the last inequality follows from $\rho_0^2 \leq N$.  For the remaining term, using Lemma \ref{lem:rln}, and taking summation over $\nt$ followed by maximum over $\ot \in \Delta_{\nt,2}$ on both sides, and using the definition \eqref{eq:coherence_h}, we obtain 
    \[
    \frac{Q^2}{L} \sum_{\nt}\left[\max_{\ot\in \Delta_{\nt,2}} \left\|\bh_{\ot}^* \vm{S}_{\nt,2}^\ddagger\left[\frac{L}{Q}\E\PT \cA_1^*\cA_1 (\hh\mm^*)-\hh\mm^*\right] \vm{D}_{\nt}\right\|_2^2\right] \leq \left(\frac{9}{2} \mu_0^2 + 2 \mu_0^2\right) = 6.5\mu_0^2.
    \]
	Plugging back returns $\mu_1^2 \leq 14\mu_0^2$. Combining this with \eqref{eq:mup-bound-pg2} for a small enough $\delta$ means that we can bound 
	\[
	\mu_p^2 \leq 4^{-p+2}\mu_0^2, ~ \mbox{for each}~ p \in \{1,2,3,\ldots, P\}.
	\]
	This completes the proof. 
	
\end{proof}

\subsection{Proof of Lemma \ref{lem:concentration}}
\begin{proof}
The proof concerns bounding deviation of $\cA_p^*\cA_p\mathcal{S}_p^\ddagger(\w_{p-1})$ from its mean, and we resort to the matrix Bernstein inequality to control it. By definition of $\cA_p$, and $\mathcal{S}_p^\ddagger$ in Section \ref{sec:Main_results}, we have 
\[
\cA_p^*\cA_p\mathcal{S}_{p}^\ddagger (\w_{p-1}) = \sum_{(\o,n) \in \Gamma_p} (\bh_\o\bh_\o^*\vm{S}_{n,p}^\ddagger\otimes \vm{\phi}_{\o,n}\vm{\phi}_{\o,n}^*)(\w_{p-1}) = \sum_{(\o,n) \in \Gamma_p}\bh_\o\bh_\o^*\vm{S}_{n,p}^\ddagger\w_{p-1}  \vm{\phi}_{\o,n}\vm{\phi}_{\o,n}^*,
\]
and its expected value is 
\begin{align}\label{eq:expectation-lemma2}
\E \cA_p^*\cA_p\mathcal{S}_p^\ddagger (\w_{p-1})  &= \sum_{(\o,n) \in \Gamma_p}\bh_\o\bh_\o^*\vm{S}_{n,p}^\ddagger\w_{p-1} \E \vm{\phi}_{\o,n}\vm{\phi}_{\o,n}^* = \sum_{(\o,n) \in \Gamma_p}\bh_\o\bh_\o^*\vm{S}_{n,p}^\ddagger\w_{p-1} \vm{D}_n.
\end{align}
Given this, the quantity of interest can now be expressed as the sum
\[
\cA_p^*\cA_p\mathcal{S}_{p}^\ddagger (\w_{p-1}) -\E \cA_p^*\cA_p\mathcal{S}_{p}^\ddagger (\w_{p-1}) = \sum_{(\o,n) \in \Gamma_p}\left( \bh_\o\bh_\o^*\vm{S}_{n,p}^\ddagger\w_{p-1}  \vm{\phi}_{\o,n}\vm{\phi}_{\o,n}^* - \bh_\o\bh_\o^*\vm{S}_{n,p}^\ddagger\w_{p-1}\vm{D}_n\right)
\]
of mean zero, independent random matrices
\[
\z_{\o,n} := \bh_\o\bh_\o^*\vm{S}_{n,p}^\ddagger\w_{p-1}  \vm{\phi}_{\o,n}\vm{\phi}_{\o,n}^* - \bh_\o\bh_\o^*\vm{S}_{n,p}^\ddagger\w_{p-1}\vm{D}_n.
\]
An application of the matrix Bernstein inequality in  Proposition \ref{prop:matbernpsi} requires us to compute a bound on the variance $\sigma_{\vm{Z}}^2$. To this end
\begin{align*}
\left\|\sum_{(\o,n) \in \Gamma_p}\E \z_{\o,n}^*\z_{\o,n} \right\| &\leq \left\|\sum_{(\o,n) \in \Gamma_p}\E \left(\bh_\o\bh_\o^*\vm{S}_{n,p}^\ddagger\w_{p-1}  \vm{\phi}_{\o,n}\vm{\phi}_{\o,n}^* \right)^*\left(\bh_\o\bh_\o^*\vm{S}_{n,p}^\ddagger\w_{p-1}  \vm{\phi}_{\o,n}\vm{\phi}_{\o,n}^*\right) \right\|\\
& = \left\|\sum_{(\o,n) \in \Gamma_p}\E |\bh_\o^*\vm{S}_{n,p}^\ddagger\w_{p-1}  \vm{\phi}_{\o,n}|^2 \vm{\phi}_{\o,n}\vm{\phi}_{\o,n}^*\right\|,
\end{align*}
where again the first inequality is the result of the fact that $\|\vm{A}-\vm{B}\| \leq \|\vm{A}\|$ when $\vm{A}$, $\vm{B}$, and $\vm{A}-\vm{B}$ are PSD. An application of Lemma \ref{lem:Supporting1} shows that 
\[
\E |\bh_\o^*\vm{S}_{n,p}^\ddagger\w_{p-1}  \vm{\phi}_{\o,n}|^2 \vm{\phi}_{\o,n}\vm{\phi}_{\o,n}^* \preccurlyeq 3 \|\bh_\o^*\vm{S}_{n, p}^\ddagger\w_{p-1}\vm{D}_n\|_2^2\vm{D}_n,
\]
 and taking the summation through returns
\[
\sum_{(\o,n) \in \Gamma_p}\E |\bh_\o^*\vm{S}_{n,p}^\ddagger\w_{p-1}  \vm{\phi}_{\o,n}|^2\vm{\phi}_{\o,n} \vm{\phi}_{\o,n}^* \preccurlyeq3\sum_{(\o,n) \in \Gamma_p}\|\bh_\o^*\vm{S}_{n,p}^\ddagger\w_{p-1}\vm{D}_n\|_2^2
\vm{D}_n,
\]
and thus the operator norm produces the variance
\begin{align}\label{eq:conc-var1}
\left\|\sum_{(\o,n) \in \Gamma_p}\E \z_{\o,n}\z_{\o,n}^*\right\| &\leq 3\left\| \sum_{(\o,n) \in \Gamma_p}\|\bh_\o^*\vm{S}_{n,p}^\ddagger\w_{p-1}\vm{D}_n\|_2^2
\vm{D}_n\right\| \notag\\
&= 3\max_{n}\sum_{\o \in \Delta_{n,p}} \|\bh_\o^*\vm{S}_{n,p}^\ddagger\w_{p-1}\vm{D}_n\|_2^2\leq 3 \rho_{p-1}^2\frac{L}{QN},
\end{align} 
where the last line follows from the definition of $\rho_p^2$ in \eqref{eq:rho_Wp}. 
In a similar manner, we can compute
\begin{align*}
\left\|\sum_{(\o,n) \in \Gamma_p}\E \z_{\o,n}\z_{\o,n}^*\right\| &\leq\left\|\sum_{n}\sum_{\o\in \Delta_{n,p}}\E  \left(\bh_\o\bh_\o^*\vm{S}_{n,p}^\ddagger\w_{p-1}  \vm{\phi}_{\o,n}\vm{\phi}_{\o,n}^* \right)\left(\bh_\o\bh_\o^*\vm{S}_{n, p}^\ddagger\w_{p-1}  \vm{\phi}_{\o,n}\vm{\phi}_{\o,n}^* \right)^*\right\|,
\end{align*}
where the first inequality follows for exact same reasoning as before. 
The summand simplifies to
\[
\bh_\o\bh_\o^*\E |\bh_\o^*\vm{S}_{n, p}^\ddagger\w_{p-1}  \vm{\phi}_{\o,n}|^2\|\vm{\phi}_{\o,n}\|_2^2, 
\]
and once the expectation is moved inside, we obtain 
\[
\E |\bh_\o^*\vm{S}_{n, p}^\ddagger\w_{p-1}  \vm{\phi}_{\o,n}|^2\|\vm{\phi}_{\o,n}\|_2^2
= 3K \|\bh_\o^*\vm{S}_{n,p}\w_{p-1}\vm{D}_n\|_2^2,
\]
and using the orthogonality of $\{\bh_\o\}_\o$, the operator norm simplifies to
\begin{align}\label{eq:conc-var2}
\left\|\sum_{(\o,n) \in \Gamma_p}\E \z_{\o,n}\z_{\o,n}^*\right\|
&\leq 3K\cdot \max_{1 \leq p \leq P}\max_{\o \in \Delta_{n,p}} \left(\sum_{n} \|\bh_\o^*\vm{S}_{n,p}^\ddagger\w_{p-1}\vm{D}_n\|_2^2 \cdot \left\|\sum_{\o\in \Delta_{n,p}}\bh_\o\bh_\o^*\right\|\right)\notag\\
& \leq 3\mu_{p-1}^2\frac{KL}{Q^2},
\end{align}
which follows by the definition of the coherence $\mu_p^2$ in \eqref{eq:mup-def}. As per \eqref{eq:matbernsigma}, the maximum of \eqref{eq:conc-var1}, and \eqref{eq:conc-var2} is the variance 
\begin{equation}\label{eq:variance-conc}
\sigma^2_{\vm{Z}} \leq 3\left(\mu_{p-1}^2\frac{KL}{Q^2}+\rho_{p-1}^2\frac{L}{QN}\right).
\end{equation}
The fact that $\z_{\o,n}$ are sub-exponential can be proven by showing that $\max_{\o,n}\|\z_{\o,n}\|_{\psi_1} < \infty$. First, note that 
\begin{align*}
\|\z_{\o,n}\|_{\psi_1} &= \|\bh_\o\bh_\o^*\vm{S}_p^\ddagger\w_{p-1}  \vm{\phi}_{\o,n}\vm{\phi}_{\o,n}^* - \bh_\o\bh_\o^*\vm{S}_p^\ddagger\w_{p-1}\vm{D}_n\|_{\psi_1} \leq 2\|\bh_\o\bh_\o^*\vm{S}_{n,p}^\ddagger\w_{p-1}  \vm{\phi}_{\o,n}\vm{\phi}_{\o,n}^*\|_{\psi_1}.
\end{align*}
Second, the operator norm of the matrix under consideration is 
\begin{align*}
\|\bh_\o\bh_\o^*\vm{S}_{n,p}^\ddagger\w_{p-1}  \vm{\phi}_{\o,n}\vm{\phi}_{\o,n}^*\| &= \|\bh_\o\|_2 \|\vm{\phi}_{\o,n}\|_2 |\bh_\o^*\vm{S}_{n,p}^\ddagger\w_{p-1}  \vm{\phi}_{\o,n}|.
\end{align*}
It is well-known that $|\bh_\o^*\vm{S}_{n,p}^\ddagger\w_{p-1}  \vm{\phi}_{\o,n}|$ is a subgaussian random variable for an arbitrary matrix $\w_{p-1}$ with $\|\bh_\o^*\vm{S}_{n, p}^\ddagger\w_{p-1}  \vm{\phi}_{\o,n}\|_{\psi_2} \leq C\|\bh_\o^*\vm{S}_{n,p}^\ddagger\w_{p-1}\vm{D}_n\|_{2} \leq C \nu_{p-1}\sqrt{L/Q^2N}$. Also $\|\vm{\phi}_{\o,n}\|_2$ is subgaussian with $\|\|\vm{\phi}_{\o,n}\|_2\|_{\psi_2} \leq \sqrt{K}$. This implies $U_{\alpha}$ in Proposition \ref{prop:matbernpsi} is 
\begin{align}\label{eq:Zw-orlicz-concentration}
U_1 &:= \max_{n} \left[\max_{\o \in \Delta_{n,p}}\|\|\vm{\phi}_{\o,n}\|_2 |\bh_\o^*\vm{S}_{n,p}^\ddagger\w_{p-1}  \vm{\phi}_{\o,n}|\|_{\psi_1}\right]\notag\\
 &\leq \max_{n}\left[\max_{\o \in \Delta_{n,p}}\left(\|\bh_\o^*\vm{S}_{n,p}^\ddagger\w_{p-1}  \vm{\phi}_{\o,n}\|_{\psi_2}\|\|\vm{\phi}_{\o,n}\|\|_{\psi_2} \right)\right]\leq C\nu_{p-1}\sqrt{\frac{KL}{Q^2N}},
\end{align}
which means
\[
\log\left(\frac{(QN) U_1^2}{\sigma_{\vm{Z}}^2} \right)  \leq C\log M,~ \mbox{where}~
M : = \frac{\nu_{p-1}^2 KQN}{\mu_{p-1}^2KN+\rho_{p-1}^2Q}.
\]
Applying the coherence bounds in \eqref{eq:rhop-bound}, \eqref{eq:nup-bound}, and \eqref{eq:mup-bound}, we have 
\[
M \leq C(\rho_0^2Q+\mu_0^2KN) \leq C LN,
\]
where the last inequality follows from the fact that $\rho_0^2 \leq N$, $\mu_0^2 \leq L$, and $L \geq K$. 
Combining all the ingredients gives us the final result and choosing $t = (\beta-1)\log (LN)$ in the Bernstein inequality shows that
\begin{align}\label{eq:Bern-bound-pg2}
&\|\cA_p^*\cA_p\mathcal{S}_p^\ddagger(\w_{p-1})-\E \cA_p^*\cA_p\mathcal{S}_p^\ddagger(\w_{p-1}) \|\leq\notag\\
&~~~C\max\left\{\sqrt{\beta \left(\mu_{p-1}^2\frac{KL}{Q^2}+\rho_{p-1}^2\frac{L}{QN}\right)\log (LN)},\sqrt{\frac{ \beta^2\nu_{p-1}^2 KL\log^4(LN)}{Q^2N}}\right\}
\end{align}
holds with probability at least $1-(LN)^{-\beta}$. Using the union bound, it follows from Lemma \ref{lem:mup-bound}, \ref{lem:rho-p-bound}, and \ref{lem:nup-bound} that at least one of the coherence bounds in \eqref{eq:rhop-bound}, \eqref{eq:nup-bound}, and \eqref{eq:mup-bound} fails with probability at most $(P-1)^3(LN)^{-\beta}$. This means that all of the coherence bounds hold with probability at least $1-(P-1)^3(LN)^{-\beta} \geq 1-(LN)^{-\beta+3}$. Plugging the coherence bounds in \eqref{eq:Bern-bound-pg2}, and choosing $Q$, and $N$ as in Lemma \ref{lem:concentration} for appropriately large constant $C$ ensures that 
\[
\|\cA_p^*\cA_p\mathcal{S}_p^\ddagger(\w_{p-1})-\E \cA_p^*\cA_p\mathcal{S}_p^\ddagger(\w_{p-1}) \|\leq 2^{-p-1}, ~~ \mbox{for each}~ p \in \{2,3, \ldots, P\}
\]
holds with probability at least $1-(LN)^{-\beta+3}$. Using union bound for $P-1$ choices of $p$, we can show that the above conclusion holds for all $p \in \{2,3,\ldots, P\}$ with probability at least $1-(P-1)(LN)^{-\beta+3} \geq 1-(LN)^{-\beta+4}$. 
\end{proof}
\subsection{Proof of Corollary \ref{cor:concentration}}
\begin{proof}
The proof follows essentially from the proof of Lemma \ref{lem:concentration} by taking $p=1$, and afterwards taking $\mathcal{S}_1^\ddagger = (L/Q) \PO$, or equivalently, $\vm{S}_{n,1}^\ddagger = (L/Q)\PO$. The final bound is obtained by making the above changes in \eqref{eq:Bern-bound-pg2}, and is 
\begin{align}
&\|\cA_1^*\cA_1(L/Q)\PO(\hh\mm^*)-\E \cA_1^*\cA_1(L/Q)\PO(\hh\mm^*) \|\leq\notag\\
&~~~C\max\left\{\sqrt{\beta \left(\mu_{0}^2\frac{KL}{Q^2}+\rho_{0}^2\frac{L}{QN}\right)\log (LN)},\sqrt{\frac{ \beta^2\nu_{0}^2 KL\log^4(LN)}{Q^2N}}\right\}
\end{align}
for some constant $C$ that may differ from the one in \eqref{eq:Bern-bound-pg2}. Choosing $Q$, and $N$ as in Corollary \ref{cor:concentration} for a large enough $C$ proves the corollary.
\end{proof}

\subsection{Proof of Lemma \ref{lem:conectrationpg2part2}}
\begin{proof}
Begin by noting that it is clear from \eqref{eq:expectation-lemma2} that for a fixed $\Gamma_p$, we have 
\[
\E\cA_p^*\cA_p\mathcal{S}_p^\ddagger (\w_{p-1}) \neq \w_{p-1}, 
\]
and if one takes the random construction of $\Gamma_p$ into account, the dependence between $\Gamma_1, \Gamma_2,\ldots, \Gamma_{p-1}$, and $\Gamma_p$ that in turn means that $\w_{p-1}$ is dependent on $\Gamma_p$. This means there is no simple way to write this quantity as a sum of independent random matrices and apply the matrix Bernstein inequality to control the size as before. Fortunately, we can work a uniform bound using restricted isometry property that works for all matrices $\w_{p-1}$, and thus overcome the issues of intricate dependencies between $\w_{p-1}$ and $\Gamma_p$. By the equivalence between operator, and Frobenius norm, we have 
\begin{align}\label{eq:intermediate}
&\|\E\cA_p^*\cA_p\mathcal{S}_p^\ddagger (\w_{p-1})-\w_{p-1}\|^2 \leq \|\E\cA_p^*\cA_p\mathcal{S}_p^\ddagger (\w_{p-1})-\w_{p-1}\|_\F^2\notag\\
&\qquad\qquad= \left\|\PO\left[\E\cA_p^*\cA_p\mathcal{S}_p^\ddagger (\w_{p-1})-\w_{p-1}\right]\right\|_\F^2 + \left\|\POc\left[\E\cA_p^*\cA_p\mathcal{S}_p^\ddagger (\w_{p-1})\right]\right\|_\F^2,
\end{align}
where the projection operator $\PO$ is defined in \eqref{eq:POmega}, and $\POc$ is the orthogonal complement. Note that 
\[
\PO\left[\E \cA_p^*\cA_p\mathcal{S}_p^\ddagger (\w_{p-1})\right] = \sum_n\left[\sum_{ \o \in \Delta_{n,p}}\PO \bh_{\o}\bh_{\o}^*\PO \vm{S}_{n,p}^\ddagger\right]\w_{p-1}\vm{D}_n.
\]
Using the definition of $\vm{S}_{n,p}^\ddagger$ in \eqref{eq:Snp}, it is clear that 
\[
\left[\sum_{ \o \in \Delta_{n,p}}\PO \bh_{\o}\bh_{\o}^*\PO \vm{S}_{n,p}^\ddagger\right] = \PO,
\]
which implies, using the fact that $\PO \w_{p-1} = \w_{p-1}$, $\PO\left[\E \cA_p^*\cA_p\mathcal{S}_p^\ddagger (\w_{p-1})\right]  = \w_{p-1}$. This means \eqref{eq:intermediate} reduces to
\[
\|\E\cA_p^*\cA_p\mathcal{S}_p^\ddagger (\w_{p-1})-\w_{p-1}\|^2 \leq \left\|\POc\left[\E\cA_p^*\cA_p\mathcal{S}_p^\ddagger (\w_{p-1})\right]\right\|_\F^2.
\]
As far as the second term in the above expression is concerned, using \eqref{eq:expectation-lemma2}, one obtains
\begin{align*}
\left\|\sum_{(\o,n) \in \Gamma_p}  \POc \bh_\o\bh_\o^*\vm{S}_{n,p}^\ddagger\w_{p-1}\vm{D}_n\right\|_\F^2 &= \left\|\sum_{(\o,n) \in \Gamma_p} \bh_\o\bh_\o^*\vm{S}_{n,p}^\ddagger\w_{p-1}\vm{D}_n\right\|_\F^2 - \left\|\sum_{(\o,n) \in \Gamma_p} \PO\bh_\o\bh_\o^*\vm{S}_{n,p}^\ddagger\w_{p-1}\vm{D}_n\right\|_\F^2 \notag\\
&= \sum_{(\o,n) \in \Gamma_p}  \|\bh_\o^*\vm{S}_{n,p}^\ddagger \w_{p-1}\vm{D}_n\|_2^2 - \|\w_{p-1}\|_\F^2.
\end{align*}
The last equality follows from the fact that $\{\bh_\o\}_\o$ are orthonormal vectors, $\w_{p-1}\vm{D}_n$ is orthogonal to $\w_{p-1}\vm{D}_{n^\prime}$ for $n \neq n^\prime$ and by the definition of $\vm{S}_{n,p}^\ddagger$ in \eqref{eq:Snp}. The matrix $\w_{p-1}$ is dependent on the sets $\Delta_{n,1}, \Delta_{n,2},\ldots, \Delta_{n,p-1}$ that are in turn dependent on $\Delta_{n,p}$ by construction. However, we can avoid this dependence issue here as the result in \eqref{eq:RIPresults} is uniform in nature in the sense that it holds for all $S$-sparse vectors. Employing this result on every column of $\w_{p-1} \in \Omega$, we obtain 
\begin{align*}
\sum_{(\o,n) \in \Gamma_p}  \|\bh_\o^*\vm{S}_{n,p}^\ddagger \w_{p-1}\vm{D}_n\|_2^2 &= \sum_n \sum_{\o \in \Delta_{n,p}}\|\bh_\o^*\vm{S}_{n,p}^\ddagger \w_{p-1}\vm{D}_n\|_2^2 \leq \frac{5Q}{4L}\sum_n\|\vm{S}_{n,p}^\ddagger\w_{p-1}\vm{D}_n\|_\F^2\\
& \leq \frac{5Q}{4L}\max_n \|\vm{S}_{n,p}^\ddagger\|_2^2 \sum_n\|\w_{p-1}\vm{D}_n\|_\F^2 \leq \frac{20}{9}\frac{L}{Q} \|\w_{p-1}\|_\F^2,
\end{align*}
which holds with probability at least $1-(LN)^{-\beta}$.  This implies using the calculation above that 
\begin{align*}
\left\|\sum_{(\o,n) \in \Gamma_p}\POc \bh_\o\bh_\o^*\vm{S}_{n,p}^\ddagger\w_{p-1}\vm{D}_n\right\|^2
& \leq \left(\frac{20}{9}\frac{L}{Q}  - 1\right) \|\w_{p-1}\|_\F^2.
\end{align*}
Now the  decay rate of $\|\w_p\|_\F$ in the statement of the lemma is sufficient to guarantee that 
\[
\|\E\cA_p^*\cA_p\mathcal{S}_p^\ddagger (\w_{p-1})-\w_{p-1}\| \leq 2^{-p-1}
\]
for every $p \in \{2,3,\ldots, P\}$. Using the union bound as before, the statement can be extended to all $p \in \{2,3,\ldots, P\}$ with probability at least $1-(LN)^{-\beta+1}$.
\end{proof}
\subsection{Proof of Lemma \ref{lem:concentrationpe1part2}}
\begin{proof}
Using the definition of $\cA_1$, and after evaluating the expectation, it is easy to see that
\[
\E \cA_1^*\cA_1(\hh\mm^*) = \sum_{(\o,n) \in \Gamma_1}\bh_{\o}\bh_{\o}^* \hh\mm^*\vm{D}_n.
\]
From Section \ref{sec:Golfing}, we know that $\Gamma_1 = \{(\Delta_{n,1},n)\}_n$, where the sets $|\Delta_{n,1}| = Q$ for every $n$, and are chosen uniformly at random. Define a set $\tilde{\Gamma}_1 := \{(\tilde{\Delta}_{n,1},n)\}_n$, where $\tilde{\Delta}_{n,1}$ for each $n$ are independent Bernoulli sets defined as 
\[
\tilde{\Delta}_{n,1} : = \{ \o \in [L] ~|~ \delta_{\o,n} = 1 \},
\]
where $\delta_{\o,n}$ is an independently chosen Bernoulli number for every $\o$, and $n$ that takes value one with probability $Q/L$. Since the probability of failure the event
\begin{equation*}
\left\|\frac{L}{Q}\sum_{n}\sum_{\o \in \Delta_{n,1}}\bh_{\o}\bh_{\o}^* \hh\mm^*\vm{D}_n - \hh\mm^* \right\| \leq \epsilon
\end{equation*}
for some number $\epsilon > 0$ is a nonincreasing function of $Q = |\Delta_{n,1}|$ for every $n$; this follows by the orthogonality of $\{\vm{b}_\ell\}_{\ell}$, and by the fact that increasing $|\Delta_{n,1}|$ only increases the range of the projector  $\sum_{\o \in \Delta_{n,1}}\bh_{\o}\bh_{\o}^*$, and hence the distance in the operator norm above can either decrease or stay the same. It now follows using Lemma 2.3 in \cite{candes2006robust} that the probability of failure of the event above is less than or equal to twice the probability of failure of the event
\begin{align}\label{eq:Bernoulli-event}
\left\|\frac{L}{Q}\sum_{(\o,n) \in \tilde{\Gamma}_1} \bh_{\o}\bh_{\o}^* \hh\mm^*\vm{D}_n- \vm{h}\vm{m}^*\right\| \leq \epsilon.
\end{align}
Therefore, in the rest of the proof it suffices to only consider the event above where the index sets are $\{\tilde{\Delta}_{n,1}\}_n$. Using the definition of the Bernoulli sets $\tilde{\Delta}_{n,1}$ defined above, we can write 
\begin{align*}
\sum_{n}\sum_{\o \in \tilde{\Delta}_{n,1}} \bh_{\o}\bh_{\o}^* \hh\mm^*\vm{D}_n = \sum_{\o,n}\delta_{\o,n} \bh_{\o}\bh_{\o}^* \hh\mm^*\vm{D}_n,
\end{align*}
where on the right hand side the summation is over all $\o \in [L]$. Note that 
\[
\E\frac{L}{Q}\sum_{\forall (\o,n)}\delta_{\o,n} \bh_\o\bh_\o^*\hh\mm^*\vm{D}_n = \sum_{\forall (\o,n)}\bh_\o\bh_\o^*\hh\mm^*\vm{D}_n = \hh\mm^*, 
\]
where we have used the fact that $\E \delta_{\o,n} = \frac{Q}{L}$. 

A simple application of matrix Bernstein is now enough to show that the event in \eqref{eq:Bernoulli-event} holds for the desired $\epsilon$ with high probability. To this end, the calculation for the variance is laid out as follows. Denote the centered random matrices 
\[
\vm{Z}_{\o,n} := \left(\frac{L}{Q}\delta_{\o,n}\bh_\o\bh_\o^*\vm{h}\mm^*\vm{D}_n- \bh_\o\bh_\o^*\vm{h}\mm^*\vm{D}_n\right). 
\]
The matrices $\vm{Z}_{\o,n}$ are independent not only for $\ell \in \Delta_{n,1}$ but for every $n$ as the set $\Delta_{n,1}$ is chosen independently for each $n$. The variance $\sigma_{\vm{Z}}^2$ is as before the maximum of the operator norms of two quantities; firstly, 
\begin{align*}
&\left\|\sum_{\o,n}\E\vm{Z}_{\o,n}\vm{Z}_{\o,n}^* \right\| \leq  \frac{L^2}{Q^2} \left\|\sum_{\o,n}\E\delta_{\o,n}^2|\bh_\o^*\hh|^2 \|\mm_n\|_2^2 \bh_\o\bh_\o^*\right\|\\
 &~~~~~\leq \frac{L}{Q} \max_{\o}|\bh_\o^*\hh|^2\left[\sum_n\|\mm_n\|_2^2\right]\left\|\sum_{\o}\bh_\o\bh_\o^*\right\| \leq  \mu_{\max}^2 \frac{S}{Q},
\end{align*}
where the second last inequality follows from the fact that $\E\delta_{\o,n}^2 = \frac{Q}{L}$, and the last inequality is obtained by applying the definitions in \eqref{eq:coherence_h}, and \eqref{eq:coherence_m}; secondly,
\begin{align*}
&\left\|\sum_{\o,n}\E\vm{Z}_{\o,n}^*\vm{Z}_{\o,n}\right\| \leq  \frac{L^2}{Q^2} \left\|\sum_{\o,n}\E\delta_{\o,n}^2|\bh_\o^*\hh|^2 \|\bh_\o\|_2^2 \vm{D}_n \mm\mm^*\vm{D}_n \right\| \\
&~~~~\leq \frac{L}{Q}\left[\sum_{\o}|\bh_{\o}^* \hh |^2 \right] \left\|\sum_n \vm{D}_n \mm\mm^*\vm{D}_n \right\|\leq  \frac{L}{Q} \max_n \|\mm_n\|_2^2 \leq \frac{L\rho_0^2}{QN},
\end{align*}
where we have used the facts that $\{\bh_{\o}\}_\o$ is a complete orthonormal basis, and that the operator norm of a block diagonal matrix with $n$th block being $\mm_n \mm_n^*$ is upper bounded by $\max_n \|\mm_n\|_2^2$. The last inequality is the result of \eqref{eq:coherence_m}. Thus, the variance $\sigma_{\vm{Z}}^2$ being the maximum of the operator norm of the two results above is bounded by
\[
\sigma_{\vm{Z}}^2 \leq \mu_{\max}^2 \frac{S}{Q}+\frac{L}{Q}\cdot\rho_0^2\frac{1}{N}.
\]
The last ingredient required to apply the Bernstein inequality in Proposition \ref{prop:matrixBernstein-uniform} is then 
\begin{align*}
\max_{\o,n} \|\vm{Z}_{\o,n}\| &= \max_{\o,n} \left\| \frac{L}{Q}\delta_{\o,n} \bh_{\o}\bh_{\o}^*\hh\mm^* \vm{D}_n - \bh_{\o}\bh_{\o}^*\hh\mm^* \vm{D}_n \right\| \leq 2 \max_{\o,n}\frac{L}{Q}\left\| \bh_\o\bh_\o^* \hh\mm^*\vm{D}_n \right\| \\
&\leq 2\frac{L}{Q}  \max_{\o}(\|\bh_\o\|_2|\bh_\o^*\hh|) \cdot \max_n \|\mm_n\|_2\leq 2\frac{L}{Q} \cdot \mu_{\max}\sqrt{\frac{S}{L}}\cdot\rho_0 \frac{1}{\sqrt{N}}.
\end{align*}
With all the ingredients in place, an application of the uniform version of the Bernstein bound with $t = (\beta-1)\log(LN)$ tells us that 
\begin{align*}
&\left\|(L/Q)\E\cA_1^*\cA_1(\hh\mm^*) - \hh\mm^*\right\| \leq \\
& \max\left\{ \sqrt{\left(\mu_{\max}^2 \frac{S}{Q}+\frac{L}{Q}\cdot\rho_0^2\frac{1}{N}\right)\beta \log (LN)}, 2\left(\mu_{\max}\rho_0\frac{L}{Q}\sqrt{\frac{S}{LN}}\right) \beta \log(LN)\right\}.
\end{align*}
The right hand side can be driven to the desired small number by choosing $N \geq C \rho_0^2(L/Q) \log(LN)$, and $L \geq C \beta\mu_{\max}^2 S(L/Q)\log(LN)$ for an appropriately large constant $C$. The probability that the above inequlaity holds is $1-(LN)^{-\beta}$ and follows by plugging in the choice $t=(\beta-1)\log(LN)$ in Proposition \ref{prop:matrixBernstein-uniform}.
\end{proof}

\section{Supporting Lemmas}\label{sec:Supporting-Lemmas}
This section proves Lemma \ref{lem:qln-bound}, Corollary \ref{cor:qln-bound}, and Lemmas \ref{lem:rln}, \ref{lem:Supporting1}. 
\subsection{Proof of Lemma \ref{lem:qln-bound}}
\begin{proof}
	We start with the proof of Lemma \ref{lem:qln-bound} that concerns bounding the quantity \footnote{The use of $\ell^\prime$, and $n^\prime$ as index variables is to avoid conflict with $\ell$, and $n$ that are reserved to index the set $\Gamma_p$ in the proof below.} $\|\vm{b}_{\ell^\prime}^*\vm{S}_{n^\prime,p+1}^\ddagger \w_p \vm{D}_{n^\prime}\|_2^2$ for $p \geq 2$. Let $\vm{q}^* := \bh_{\ot}^*\vm{S}_{\nt,p+1}^\ddagger \w_p\vm{D}_{\nt}$. The quantity can be expanded using the definition of $\w_p$ for $p \geq 2$ in \eqref{eq:Wp-def2} 
	\begin{align*}
	\vm{q}^* &= \bh_{\ot}^*\vm{S}^\ddagger_{\nt,p+1}\left[\PT\cA_p^*\cA_p\mathcal{S}_p^\ddagger(\w_{p-1})-\w_{p-1} \right]\vm{D}_{\nt}\\
	&=\bh_{\ot}^*\vm{S}^\ddagger_{\nt,p+1}\left[\PT\cA_p^*\cA_p\mathcal{S}_p^\ddagger(\w_{p-1})-\PT\E\cA_p^*\cA_p\mathcal{S}_p^\ddagger(\w_{p-1}) \right]\vm{D}_{\nt},
	\end{align*}
	where the second equality follows from a previously shown fact that $\E \PT\cA_p^*\cA_p \mathcal{S}_p^\ddagger(\w_{p-1}) = \w_{p-1}$. Thus $\vm{q}$ is just a zero-mean random vector that can be expanded further as a sum of zero-mean, independent random vectors using the definition of map $\cA_p^*\cA_p$ in \eqref{eq:cAp} as follows
	\begin{align}\label{eq:q-summands}
	\vm{q}^* & = \sum_{(\o,n) \in \Gamma_p} \bh_{\ot}^*\vm{S}_{\nt,p+1}^\ddagger\left[\PT\left(\bh_\o\bh_\o^*\vm{S}_{n,p}^\ddagger \w_{p-1}\vm{\phi}_{\o,n}\vm{\phi}_{\o,n}^*\right) - \PT\left(\bh_\o\bh_\o^*\vm{S}_{n,p}^\ddagger \w_{p-1}\E\vm{\phi}_{\o,n}\vm{\phi}_{\o,n}^*\right)\right]\vm{D}_{\nt}.
	\end{align}
	From here on, we use the matrix Bernstein inequality to find the range in which the $\ell_2$-norm of the random vector $\vm{q}$ lies with high probability. Let us define random vectors 
	\[
	\vm{z}^*_{\o,n} := \bh_{\ot}^*\vm{S}_{\nt,p+1}^\ddagger\PT\left(\bh_\o\bh_\o^*\vm{S}_{n,p}^\ddagger \w_{p-1}\vm{\phi}_{\o,n}\vm{\phi}_{\o,n}^*\right)\vm{D}_{\nt}
	\]
	 then by Proposition \ref{prop:matbernpsi}, it suffices to compute the following upper bound on the variance
	\begin{align*}
	\sigma_{\vm{z}}^2 \leq  \max\left\{\left\|\sum_{(\o,n) \in \Gamma_p} \E \vm{z}_{\o,n}\vm{z}_{\o,n}^*\right\|, \left| \sum_{(\o,n) \in \Gamma_p} \E \|\vm{z}_{\o,n}\|_2^2\right|\right\} \leq \left| \sum_{(\o,n) \in \Gamma_p} \E\|\vm{z}_{\o,n}\|_2^2\right|.
	\end{align*}
	Note that the vectors $\vm{z}_{\o,n}$ can be rewritten as a scalar times a vector as follows
	\[
	\vm{z}^*_{\o,n} = \bh_\o^*\vm{S}_{n,p}^\ddagger \w_{p-1}\vm{\phi}_{\o,n}  \cdot \bh_{\ot}^*\vm{S}_{\nt,p+1}^\ddagger\PT\left(\bh_\o\vm{\phi}_{\o,n}^*\right)\vm{D}_{\nt},
	\]
	and using \eqref{eq:PT-def}, the vector part above can be expanded as 
	\begin{align*}
	\bh_{\ot}^*\vm{S}^\ddagger_{\nt,p+1}\PT(\bh_\o\vm{\phi}_{\o,n}^*)\vm{D}_{\nt}  = (\bh_\o^* \hh)^*(\bh_{\ot}^*\vm{S}^\ddagger_{\nt,p+1}\hh)\vm{\phi}^*_{\o,n}\vm{D}_{\nt}&+(\vm{\phi}^*_{\o,n}\mm)(\bh_{\ot}^*\vm{S}^\ddagger_{\nt,p+1}\bh_\o)\mm^*\vm{D}_{\nt}\\
	&\qquad\qquad-(\bh_\o^* \hh)^*(\vm{\phi}^*_{\o,n}\mm)(\bh_{\ot}^*\vm{S}^\ddagger_{\nt, p+1}\hh)\mm^*\vm{D}_{\nt},
	\end{align*}
	and its the $\ell_2$-norm can then easily be shown to be upper bounded as 
	\begin{align*}
	&\|\bh_{\ot}^*\vm{S}_{\nt,p+1}^\ddagger\PT\left(\bh_\o\vm{\phi}_{\o,n}^*\right)\vm{D}_{\nt}\|_2^2\\
	 &\qquad\qquad\leq 3|\bh_\o^* \hh|^2|\bh_{\ot}^*\vm{S}^\ddagger_{\nt,p+1}\hh|^2 \left(\vm{\phi}^*_{\o,n}\vm{D}_{\nt}\vm{D}_{\nt}^*\vm{\phi}_{\o,n}\right) + 3|\vm{\phi}^*_{\o,n}\mm|^2 |\bh_{\ot}^*\vm{S}^\ddagger_{\nt,p+1}\bh_\o|^2 \|\mm_{\nt}\|_2^2 \\
	&\qquad\qquad+ 3|\bh_\o^* \hh|^2 |\vm{\phi}^*_{\o,n}\mm|^2 |\bh_{\ot}^*\vm{S}^\ddagger_{\nt,p+1}\hh|^2 \|\mm_{\nt}\|_2^2\\
	&\qquad\qquad\leq 3|\bh_\o^* \hh|^2|\bh_{\ot}^*\vm{S}^\ddagger_{\nt,p+1}\hh|^2 \left(\vm{\phi}^*_{\o,n}\vm{D}_{\nt}\vm{D}_{\nt}^*\vm{\phi}_{\o,n}\right) + 6|\vm{\phi}^*_{\o,n}\mm|^2 |\bh_{\ot}^*\vm{S}^\ddagger_{\nt,p+1}\bh_\o|^2 \|\mm_{\nt}\|_2^2,
	\end{align*}
	where the last line follows from the fact that $|\bh_\o^*\hh|^2 \leq 1$. For a standard Gaussian vector $\vm{g}$, it can easily be verified that $\E \left(|\vm{g}^*\vm{x}|^2 |\vm{g}^*\vm{y}|^2\right)\leq 3\|\vm{x}\|_2^2\|\vm{y}\|_2^2$. Using this fact, one has 
	\[
	\E |\vm{\phi}^*_{\o,n}\mm|^2 |\bh_\o^*\vm{S}_{n,p}^\ddagger \w_{p-1}\vm{\phi}_{\o,n}|^2 \leq 3  \|\mm_n\|_2^2\|\bh_\o^*\vm{S}_{n,p}^\ddagger \w_{p-1}\vm{D}_n\|_2^2,
	\]
	and
	\[
	\sum_{n=1}^N\left[\max_{\o \in \Delta_{n,p}}\E\left(\vm{\phi}^*_{\o,n}\vm{D}_{\nt}\vm{D}_{\nt}^*\vm{\phi}_{\o,n}\right) |\bh_\o^*\vm{S}_{n,p}^\ddagger \w_{p-1}\vm{\phi}_{\o,n}|^2\right]  \leq 3K\left[\max_{\o \in \Delta_{\nt,p}}\|\bh_\o^*\vm{S}_{\nt,p}^\ddagger \w_{p-1}\vm{D}_{\nt}\|_2^2\right].
	\]
	Note the change of index variable from $n$ to $\nt$ on the right hand side. Moreover, using \eqref{eq:RIPresults}
	\[
	\sum_{\o \in \Delta_{n,p}} |\bh_{\ot}^*\vm{S}_{\nt,p+1}^\ddagger \bh_\o|^2 \leq  \frac{5Q}{4L}\|\bh_{\ot}^*\vm{S}_{\nt,p+1}^\ddagger\|_2^2 \leq \frac{5Q}{4L} \|\mathcal{P}\bh_{\ot}\|_2^2\|\vm{S}^\ddagger_{\nt,p+1}\|_2^2,
	\]
	where the inclusion of the projection operator $\mathcal{P}$, defined in \eqref{eq:POmega}, on the vector $\bh_{\ot}$ is due to the fact that the rows of matrix $\vm{S}_{n,p}^\ddagger$ are supported on $\Omega$. Furthermore, using definitions in \eqref{eq:coherence_B}, and \eqref{eq:S-eigen-bounds}, we can finally bound the above result as 
	\[
	\sum_{\o \in \Delta_{n,p}} |\bh_{\ot}^*\vm{S}_{\nt,p+1}^\ddagger \bh_\o|^2 \leq \frac{20}{9}\mu_{\max}^2\frac{S}{Q}.
    \]
	Putting the above identities together with
	\[
	\sum_{n=1}^N \|\mm_n\|_2^2 = \|\mm\|_2^2 = 1, \qquad \mbox{and}~ \sum_{\o \in \Delta_{n,p}} |\bh_\o^*\hh|^2 \leq \frac{5Q}{4L}\|\hh\|_2^2 = \frac{5Q}{4L}
	\]
	one directly obtains
	\begin{align}\label{eq:Vz}
	 \sum_{(\o,n) \in \Gamma_p} \|\vm{z}_{\o,n}\|_2^2 &= \sum_{n=1}^N\sum_{\o \in \Delta_{n,p}} \|\vm{z}_{\o,n}\|_2^2\leq 4\frac{Q}{L} K |\bh_{\ot}^*\vm{S}^\ddagger_{\nt,p+1}\hh|^2\left[ \max_{\o \in \Delta_{\nt,p}}\|\bh_\o^*\vm{S}_{\nt,p}^\ddagger \w_{p-1}\vm{D}_{\nt}\|_2^2\right]\notag\\
	 &\qquad\qquad\qquad\qquad+ 14 \mu_{\max}^2 \frac{S}{Q}\|\mm_{\nt}\|_2^2 \left[\max_n \max_{\o \in \Delta_{n,p}} \|\bh_{\o}^*\vm{S}_{n,p}^\ddagger\w_{p-1}\vm{D}_n\|_2^2\right] := \Pi^2_p.
	\end{align}
	The only ingredient left to apply the Bernstein bound in Proposition \ref{prop:matbernpsi} is the Orlicz norm of the summands $\vm{z}_{\o,n}$.  To this end, the $\psi_2$ norm of the vector $\bh_{\ot}^*\vm{S}_{n,p+1}^\ddagger\PT(\bh_\o\vm{\phi}_{\o,n}^*)\vm{D}_{\nt}$ can be evaluated as follows
	\begin{align*}
	&\|\bh_{\ot}^*\vm{S}_{\nt,p+1}^\ddagger\PT(\bh_\o\vm{\phi}_{\o,n}^*)\vm{D}_{\nt}\|_{\psi_2} 
	\leq |\bh_\o^*\hh||\bh_{\ot}^*\vm{S}_{\nt,p+1}^\ddagger\hh|\|\vm{\phi}_{\o,n}\|_{\psi_2}+2\|\vm{\phi}^*_{\o,n}\mm\|_{\psi_2}|\bh_{\ot}^*\vm{S}_{\nt,p+1}^\ddagger\bh_\o| \|\mm_{\nt}\|_2.
	\end{align*}
	Since $\vm{\phi}_{\o,n}$ are Gaussian vectors, the discussion on the Orlicz-norms in Section \ref{sec:Conc-Ineq} tells us that $\|\vm{\phi}^*_{\o,n}\mm\|_{\psi_2} \leq C \|\mm_n\|_2$, and $\|\vm{\phi}_{\o,n}\|_{\psi_2} \leq C\sqrt{K}$. This means
	\begin{align*}
\|\bh_{\ot}^*\vm{S}_{\nt,p+1}^\ddagger\PT(\bh_\o\vm{\phi}_{\o,n}^*)\vm{D}_{\nt}\|_{\psi_2} &\leq  C\left(\sqrt{K}|\bh_\o^*\hh|\cdot|\bh_{\ot}^*\vm{S}_{\nt,p+1}^\ddagger\hh| + \|\mm_n\|_{2}|\bh_{\ot}^*\vm{S}_{\nt,p+1}^\ddagger\bh_\o|\|\mm_{\nt}\|_2\right) \\
	&\leq C\left( \sqrt{K}|\bh_\o^*\hh||\bh_{\ot}^*\vm{S}_{\nt,p+1}^\ddagger\hh| +\|\PO\bh_\o\|_2\|\bh_{\ot}^*\vm{S}_{\nt,p+1}^\ddagger\|_2\|\mm_n\|_2\|\mm_{\nt}\|_2\right),
	\end{align*}
	where we have used the fact that $\|\bh_{\ot}^*\vm{S}_{\nt,p+1}^\ddagger\bh_\o\|_2 = \|\bh_{\ot}^*\PO\vm{S}_{\nt,p+1}^\ddagger\PO\bh_\o\|_2 \leq \|\PO\bh_\o\|_2\|\bh_{\ot}^*\vm{S}_{\nt,p+1}^\ddagger\|_2$.  Again as  $\vm{\phi}_{\o,n}$ is a Gaussian vector, the Orlicz norm of its inner product with a fixed vector $\bh_\o^*\vm{S}_{n,p}^\ddagger \w_{p-1}$ is 
	\begin{align*}
	\left\|\bh_\o^*\vm{S}_{n,p}^\ddagger \w_{p-1}\vm{\phi}_{\o,n}\right\|_{\psi_2} \leq C\left\|\bh_\o^* \vm{S}_{n,p}^\ddagger\w_{p-1}\vm{D}_n\right\|_2.
	\end{align*}
	Using these facts, we can show that random summands in \eqref{eq:q-summands} are sub-exponential vectors by computing their $\psi_1$-norm and showing that it is bounded. Note that 
	\begin{align*}
	&\left\|\bh_{\ot}^*\vm{S}_{\nt,p+1}^\ddagger\left[\PT\left(\bh_\o\bh_\o^*\vm{S}_{n,p}^\ddagger \w_{p-1}\vm{\phi}_{\o,n}\vm{\phi}_{\o,n}^*\right)- \PT\left(\bh_\o\bh_\o^*\vm{S}_{n,p}^\ddagger \w_{p-1}\E\vm{\phi}_{\o,n}\vm{\phi}_{\o,n}^*\right)\right]\vm{D}_{\nt}\right\|_{\psi_1} \\
	&\qquad\qquad\leq 2\left\|\bh_{\ot}^*\vm{S}_{\nt,p+1}^\ddagger\PT\left(\bh_\o\bh_\o^*\vm{S}_{n,p}^\ddagger \w_{p-1}\vm{\phi}_{\o,n}\vm{\phi}_{\o,n}^*\right)\vm{D}_{\nt}\right\|_{\psi_1} ,
	\end{align*}
	where the inequality follows by using the identity in \eqref{eq:Orlicz-norm-centered-X}. Now note that 
	\begin{align*}
	\max_n \max_{\o \in \Delta_{n,p}}&\left\|\bh_{\ot}^*\vm{S}_{\nt,p+1}^\ddagger\PT\left(\bh_\o\bh_\o^*\vm{S}_{n,p}^\ddagger \w_{p-1}\vm{\phi}_{\o,n}\vm{\phi}_{\o,n}^*\right)\vm{D}_{\nt}\right\|_{\psi_1}\\ &\qquad\qquad\leq \max_{\o \in \Delta_{\nt,p}}\left\|\bh_{\ot}^*\vm{S}_{\nt,p+1}^\ddagger\PT\left(\bh_\o\bh_\o^*\vm{S}_{\nt,p}^\ddagger \w_{p-1}\vm{\phi}_{\o,\nt}\vm{\phi}_{\o,\nt}^*\right)\right\|_{\psi_1},
	\end{align*}
	and the  change of indices from $n$ to $\nt$ is justified as $\phi_{\o,n}[\kt]$ is zero when $\kt \notin \{(K-1)\nt+1,\ldots, \nt K\}$.  Now using the result in \eqref{eq:Product-Subgaussians}, we can write 
	\begin{align*}
	&\left\|\bh_{\ot}^*\vm{S}_{\nt,p+1}^\ddagger\PT\left(\bh_\o\bh_\o^*\vm{S}_{\nt,p}^\ddagger \w_{p-1}\vm{\phi}_{\o,\nt}\vm{\phi}_{\o,\nt}^*\right)\right\|_{\psi_1} \leq \left\|\bh_{\ot}^*\vm{S}_{\nt,p+1}^\ddagger\PT\left(\bh_\o\vm{\phi}^*_{\o,\nt}\right)\right\|_{\psi_2}\left\|\bh_\o^*\vm{S}_{\nt,p}^\ddagger \w_{p-1}\vm{\phi}_{\o,\nt}\right\|_{\psi_2}\\
	& \qquad\qquad \leq \left(\sqrt{K}|\bh_\o^* \hh||\bh_{\ot}^*\vm{S}_{\nt,p+1}^\ddagger\hh| +\|\PO\bh_\o\|_2\|\bh_{\ot}^*\vm{S}_{\nt,p+1}^\ddagger\|_2 \|\mm_{\nt}\|_2^2\right)\left(\left\|\bh_\o^* \vm{S}_{\nt,p}^\ddagger\w_{p-1}\vm{D}_{\nt}\right\|_2\right).
	\end{align*}
	 Setting $\alpha = 1$ in Proposition \ref{prop:matbernpsi}, one obtains the upper bound $U_1$
	\begin{align}\label{eq:Uz}
	 U_1^2  &= \max_{n}\max_{\o\in \Delta_{n,p}} \|\vm{z}_{\o,n}\|^2_{\psi_1} \leq \notag C\left(\frac{\mu^2_0K}{L}|\bh_{\ot}^*\vm{S}_{\nt,p+1}^\ddagger\hh|^2+\mu^2_{\max}\rho_0^2\frac{S}{LN}\|\mm_{\nt}\|_2^2\|\bh_{\ot}^* \vm{S}_{\nt,p+1}^\ddagger\|_2^2\right)\notag\\
	&~~~~~~~~~~~~~~~~~~~~~~~~~\qquad\qquad\qquad\left[\max_{\o \in \Delta_{\nt,p}}\left\|\bh_\o^* \vm{S}_{\nt,p}^\ddagger\w_{p-1}\vm{D}_{\nt}\right\|_2^2\right]\notag\\
	&\leq C \left(\frac{\mu^2_0K}{L}|\bh_{\ot}^*\vm{S}_{\nt,p+1}^\ddagger\hh|^2+\mu^4_{\max}\rho_0^2\frac{S^2}{Q^2N}\|\mm_{\nt}\|_2^2\right)\left[\max_{\o \in \Delta_{\nt,p}}\left\|\bh_\o^* \vm{S}_{\nt,p}^\ddagger\w_{p-1}\vm{D}_{\nt}\right\|_2^2\right]:= \Lambda^2_p,
	\end{align}
	where the last inequality follows from plugging in the bound for $\|\bh_{\ot}^* \vm{S}_{\nt,p+1}^\ddagger\|_2^2$ calculated earlier in the proof of this lemma. The logarithmic factor in the Bernstein bound can be crudely bounded as follows
	\[
	\log\left( \frac{(QN)U_{1}^2}{\sigma_{\vm{z}}^2}\right) \leq C\log(LN)
	\]
	for some constant $C$, which follows from the fact that $U_1^2 \leq |\Gamma_p| \sigma_{\vm{z}}^2$. This completes all the ingredients to apply the Bernstein bound with $t = \beta \log (LN)$ to obtain 
	\begin{align*}
	&\|\vm{b}_{\ell^\prime}^*\vm{S}_{n^\prime,p+1}^\ddagger \w_p \vm{D}_{n^\prime}\|_2^2 \leq C\max\left\{\beta\Pi_p^2\log (LN), \beta^2 \Lambda_p^2\log^4 (LN)\right\},
	\end{align*}
	which holds with probability at least $1-(LN)^{-\beta}$. This completes the proof of the lemma. 
	 \end{proof}
\subsection{Proof of Corollary \ref{cor:qln-bound}} 

\begin{proof}
	Note that from \eqref{eq:sln}, we can equivalently write
	\[
	\vm{s}_{\ot,\nt}^* := \bh_{\ot}^*\vm{S}^\ddagger_{\nt,2}\left[\PT\cA_1^*\cA_1(L/Q)\PO(\hh\mm^*)- \E\PT\cA_1^*\cA_1(L/Q)\PO(\hh\mm^*)\right]\vm{D}_{\nt}.
	\]
	using the fact that $\hh\mm^* \in \Omega$. In comparison, Lemma \ref{lem:qln-bound} was concerned with bounding the $\ell_2$ norm of the term 
	\[
	\bh_{\ot}^*\vm{S}^\ddagger_{\nt,p+1}\left[\PT\cA_p^*\cA_p\mathcal{S}_p^\ddagger(\w_{p-1})-\E \PT\cA_p^*\cA_p \mathcal{S}_p^\ddagger(\w_{p-1}) \right]\vm{D}_{\nt}. 
	\]
	All we need to do over here is to replace $\vm{S}_{\nt,p+1}$, $\cA_p$, $\mathcal{S}_p^\ddagger$, and $\w_{p-1}$ with $\vm{S}_{\nt,2}$, $\cA_1$, $(L/Q)\PO$, and $\hh\mm^*$ and repeat the same argument as in the proof of Lemma \ref{lem:qln-bound}. This leads to the bound on the $\ell_2$ norm of $\vm{s}_{\o,n}$ in the statement of the corollary in terms of $\Pi^2_1$, and $\Lambda^2_1$ defined in \eqref{eq:G1}, and \eqref{eq:H1}, respectively. Compared to Lemma \ref{lem:qln-bound} the $\Pi^2_p$, and $\Lambda^2_p$ are replaced in the result by $\Pi^2_1$, and $\Lambda^2_1$, respectively. The quantities $\Pi_1^2$ and $\Lambda_1^2$ are just $\Pi^2_p$, and $\Lambda^2_p$ evaluated at $p=1$, and afterwards $\vm{S}_{n,1}^\ddagger$ therein replaced by $(L/Q)\PO$. 
\end{proof}
\subsection{Proof of Lemma \ref{lem:rln}}

\begin{proof}
	Using the definition \eqref{eq:rln}, and the triangle inequality, we have 
	\begin{align}\label{eq:rln-l2-norm}
	\|\vm{r}^*_{\ot,\nt}\|_2 \leq  \left\|\frac{L}{Q}\bh_{\ot}^*\vm{S}^\ddagger_{\nt,2}\PT\E\cA_1^*\cA_1(\hh\mm^*)\vm{D}_{\nt}\right\|_2 + \|\bh_{\ot}^*\vm{S}_{\nt,2}^\ddagger \hh\mm^* \vm{D}_{\nt}\|_2,
	\end{align}
	where using the definition of $\vm{S}_{n,1}$ defined in Section \ref{sec:Golfing}, we have 
	\begin{align*}
	&\frac{L}{Q}\bh_{\ot}^*\vm{S}^\ddagger_{\nt,2}\PT\E\cA_1^*\cA_1(\hh\mm^*) \vm{D}_{\nt} = \frac{L}{Q} \bh_{\ot}^*\vm{S}_{\nt,2}^\ddagger \left[\sum_{n}\PT\left(\vm{S}_{n,1}\hh\mm^*\vm{D}_n\right)\right]\vm{D}_{\nt} = \\
	&\frac{L}{Q} \bh_{\ot}^*\vm{S}_{\nt,2}^\ddagger \left[\sum_{n} \left(\hh^*\vm{S}_{n,1}\hh\right)\hh\mm^*\vm{D}_n+\vm{S}_{n,1}\hh\mm^* \|\mm_n\|_2^2-  \left(\hh^*\vm{S}_{n,1}\hh\right)\hh\mm^*\|\mm_n\|_2^2\right]\vm{D}_{\nt}\leq \\
	&~~~~~~\frac{L}{Q} \bh_{\ot}^*\vm{S}_{\nt,2}^\ddagger \left[\frac{5Q}{4L}\hh\mm^*\sum_n\vm{D}_n+\max_n\vm{S}_{n,1}\hh\mm^* \sum_n\|\mm_n\|_2^2- \frac{3Q}{4L}\hh\mm^*\sum_n \|\mm_n\|_2^2\right]\vm{D}_{\nt},
	\end{align*}
	where the second last equality results from the definition of $\PT$ in \eqref{eq:PT-def}, and from \eqref{eq:S-eigen-bounds} that directly imply that 
	\[
	\frac{3Q}{4L} \leq \hh^*\vm{S}_{n,1}\hh \leq \frac{5Q}{4L},
	\]
	for the choice of $Q$ in the lemma. Since $\sum_n \|\mm_n\|_2^2 = 1$, and $\sum_n \vm{D}_n = \vm{I}_{KN}$, the result above simplifies to 
	\begin{align*}
	\frac{L}{Q}\bh_{\ot}^*\vm{S}^\ddagger_{\nt,2}\PT\E\cA_1^*\cA_1(\hh\mm^*) \vm{D}_{\nt} \leq  \frac{L}{Q} \bh_{\ot}^*\vm{S}_{\nt,2}^\ddagger \left[ \frac{Q}{2L} \hh\mm^* + \max_n\vm{S}_{n,1}\hh\mm^*\right]\vm{D}_{\nt}.
	\end{align*}
	Finally, the operator norm of the above quantity returns 
	\[
	\left\|\frac{L}{Q}\bh_{\ot}^*\vm{S}^\ddagger_{\nt,2}\PT\E\cA_1^*\cA_1(\hh\mm^*) \vm{D}_{\nt} \right\|_2 \leq \left(\frac{1}{2} |\bh_{\ot}^*\vm{S}_{\nt,2}^\ddagger\hh| + \frac{L}{Q}\max_n |\bh_{\ot}^*\vm{S}_{\nt,2}^\ddagger\vm{S}_{n,1}\hh|\right) \|\mm_{\nt}\|_2.
	\]
	Using \eqref{eq:rln-l2-norm}, we obtain 
	\[
	\|\vm{r}_{\ot,\nt}\|_2 \leq  \left(\frac{3}{2} |\bh_{\ot}^*\vm{S}_{\nt,2}^\ddagger\hh| + \frac{L}{Q}\max_n |\bh_{\ot}^*\vm{S}_{\nt,2}^\ddagger\vm{S}_{n,1}\hh|\right)\|\mm_{\nt}\|_2,
	\]
	and squaring both sides results in 
	\[
	\|\vm{r}_{\ot,\nt}\|_2^2 \leq  \left(\frac{9}{2} |\bh_{\ot}^*\vm{S}_{\nt,2}^\ddagger\hh|^2 + 2\frac{L^2}{Q^2}\max_n |\bh_{\ot}^*\vm{S}_{\nt,2}^\ddagger\vm{S}_{n,1}\hh|^2\right)\|\mm_{\nt}\|_2^2.
	\]
	This completes the proof.
\end{proof}
\begin{lem}\label{lem:Supporting1}
	Let $\vm{\phi}_{\o,n} $ be as in defined in \eqref{eq:meas}, and $\X$ be a fixed matrix. Then 
	\[
	\E |\<\vm{X},\bh_\o\vm{\phi}_{\o,n}^*\>|^2 \vm{\phi}_{\o,n}\vm{\phi}_{\o,n}^*  \preccurlyeq 3\|\bh_\o^*\X\vm{D}_n\|_2^2\vm{D}_n.
	\]
\end{lem}
\begin{proof}
	Note that
	\begin{align*} 
	\E\left( |\<\vm{X},\bh_\o\vm{\phi}_{\o,n}^*\>|^2\vm{\phi}_{\o,n}\vm{\phi}_{\o,n}^*\right)&= \E\left( |\bh_{\o}^*\vm{X}\vm{\phi}_{\o,n}|^2\vm{\phi}_{\o,n}\vm{\phi}_{\o,n}^*\right) \\
	&= \E \left(\left|\sum_{k \sim_K n} \{\bh_\o^*\vm{X}\vm{\phi}_{\o,n}\}[k]\right|^2\vm{\phi}_{\o,n}\vm{\phi}_{\o,n}^*\right)\\
	&= \|\bh_\o^*\X\vm{D}_n\|_2^2\vm{D}_n + 2 \bh_\o^*\vm{X}\vm{X}^*\bh_\o \vm{D}_n \preccurlyeq 3 \|\bh_\o^*\X\vm{D}_n\|_2^2 \vm{D}_n.
	\end{align*}
\end{proof}

\bibliographystyle{plain}
\bibliography{Bibliography}
\end{document}